%% file: main.tex
\documentclass[10pt]{article}

\usepackage[margin=2.5cm]{geometry}
\PassOptionsToPackage{numbers,compress}{natbib}
\usepackage{preamble}
\usepackage{macros}
\usepackage{pifont}     
\newcommand{\cmark}{\ding{51}}%
\newcommand{\xmark}{\ding{55}}%

\newcommand*\samethanks[1][\value{footnote}]{\footnotemark[#1]}

\title{
\fontseries{b}\selectfont
No-Regret Learning in Games with Noisy Feedback:\\
Faster Rates and Adaptivity via Learning Rate Separation}

\author{
\begin{tabular}{cc}
Yu-Guan Hsieh\thanks{Univ. Grenoble Alpes}
& Kimon Antonakopoulos\thanks{EPFL}\\
\texttt{yu-guan.hsieh@univ-grenoble-alpes.fr}
& \texttt{kimon.antonakopoulos@epfl.ch}\\[0.75em]
Volkan Cevher\samethanks
& Panayotis Mertikopoulos\samethanks[1]\hspace{0.35em}\thanks{CNRS, Inria, LIG \& Criteo AI Lab}\\
\texttt{volkan.cevher@epfl.ch}
&\texttt{panayotis.mertikopoulos@imag.fr}
\end{tabular}
}

\allowdisplaybreaks

\begin{document}

\maketitle



\etocdepthtag.toc{mtchapter}
\etocsettagdepth{mtchapter}{subsection}
\etocsettagdepth{mtappendix}{none}

\input{content}

\input{appendices/apx-content}


\end{document}

%% file: content.tex



\vspace{-0.5em}
\input{sections/abstract}

\acused{OG+}
\acused{OptDA+}

\section{Introduction}
\label{sec:introduction}
\input{sections/introduction}

\input{tables/results-summary}

\section{Problem Setup}
\label{sec:setup}
\input{sections/preliminaries}

\section{\Acl{OG} methods: Definitions, difficulties, and a test case}
\label{sec:motivating}

\input{sections/motivating}

\section{Regret minimization with noisy feedback}
\label{sec:regret}
\input{sections/regret}

\section{Adaptive learning rates}
\label{sec:adaptive}
\input{sections/adaptive}

\section{Trajectory analysis}
\label{sec:trajectory}
\input{sections/trajectory}


\section{Concluding remarks}
\label{sec:conlcusion}
\input{sections/conclusion}

\section*{Acknowledgments}
\small
\input{sections/acknowledgment}
\normalsize

\bibliographystyle{plainnat}
\bibliography{bibtex/references}

%% file: sections/abstract.tex
\begin{abstract}
We examine the problem of regret minimization when the learner is involved in a continuous game with other optimizing agents:
in this case, if all players follow a no-regret algorithm, it is possible to achieve significantly lower regret relative to fully adversarial environments.
We study this problem in the context of variationally stable games (a class of continuous games which includes all convex-concave and monotone games), and when the players only have access to noisy estimates of their individual payoff gradients.
If the noise is additive, the game-theoretic and purely adversarial settings enjoy similar regret guarantees;
however, if the noise is \emph{multiplicative}, we show that the learners can, in fact, achieve \emph{constant} regret.
We achieve this faster rate
via an optimistic gradient scheme with \emph{learning rate separation} \textendash\ that is, the method's extrapolation and update steps are tuned to different schedules, depending on the noise profile.
Subsequently, to eliminate the need for delicate hyperparameter tuning, we propose a fully adaptive method that attains nearly the same guarantees as its non-adapted counterpart, while operating without knowledge of either the game or of the noise profile.

\end{abstract}

%% file: sections/introduction.tex

Owing to its simplicity and versatility, the notion of regret has been the mainstay of online learning ever since the field's first steps \cite{Blackwell56,Hannan57}.
Stated abstractly,
it concerns processes of the following form:\!
\begin{enumerate}
\item 
At each stage $\run=\running$, the learner selects an action $\vt[\action]$ from some $\vdim$-dimensional real space.
\item
The environment determines a convex loss function $\loss_{\run}$ and the learner incurs a loss of $\loss_{\run}(\vt[\action])$.
\item
Based on this loss (and any other piece of information revealed), the learner updates their action $\vt[\action] \gets \update[\action]$ and the process repeats.
\end{enumerate}
In this general setting, the agent's regret $\vt[\reg][\nRuns]$ is defined as the difference between the cumulative loss incurred by the sequence $\action_{\run}$, $\run=\running,\nRuns$, versus that of the best fixed action over the horizon of play $\nRuns$.
Accordingly, the learner's objective is to minimize the growth rate of $\vt[\reg][\nRuns]$, guaranteeing in this way that the chosen sequence of actions becomes asymptotically efficient over time.

Without further assumptions on the learner's environment or the type of loss functions encountered, it is not possible to go beyond the well-known minimax regret bound of $\Omega(\sqrt{\nRuns})$ \cite{Shalev11,Hazan16}, which is achieved by the \ac{OGD} policy of \citet{Zin03}.
However, this lower bound concerns environments that are ``adversarial'' and loss functions that may vary arbitrarily from one stage to the next:
if the environment is ``smoother'' \textendash\ and not actively seeking to sabotage the learner's efforts \textendash\ one could plausibly expect faster regret minimization rates.

\input{figures/illustration-main}

This question is particularly relevant \textendash\ and has received significant attention \textendash\ in the backdrop of \emph{multi-agent} learning in games.
Here, the learners' environment is no longer arbitrary:
instead, each player interacts with other regret minimizing players, and every player's individual loss function is determined by the actions chosen by all players via a fixed underlying mechanism \textendash\ that of a \emph{non-cooperative game}.
Because of this mechanism \textendash\ and the fact that players are changing their actions incrementally from one round to the next \textendash\ the learners are facing a much more ``predictable'' sequence of events.
As a result, there has been a number of research threads in the literature showing that it is possible to attain \emph{near-constant regret} (\ie at most polylogarithmic) in different classes of games, from the work of \cite{DDK11,KHSC18} on finite two-player zero-sum games,
to more recent works on general-sum finite games \cite{DFG21,ADFF+21,AFKL+22}, extensive form games \cite{FLLK22}, and even continuous games \cite{HAM21}.

\para{Our contributions in the context of related work}

The enabling technology for this range of near-constant regret guarantees is the \acdef{OG} algorithmic template, itself a variant of the \acf{EG} algorithm of \citet{Kor76}.
The salient feature of this method \textendash\ first examined by \citet{Pop80} in a game-theoretic setting and subsequently popularized by \citet{RS13-NIPS} in the context of online learning \textendash\ is that players use past gradient information to take a more informed ``look-ahead'' gradient step that stabilizes the method and leads to lower regret.
This, however, comes with an important caveat:
all of the above works crucially rely on the players' having access to exact payoff gradients, an assumption which is often violated in practice.
When the players' feedback is corrupted by noise (or other uncertainty factors), the very same algorithms discussed above may incur \textit{superlinear regret} (\cf \cref{fig:illustration}).
We are thus led to the following natural question:
\begin{center}
\emph{Is it possible to achieve constant regret in the presence of noise and uncertainty?}
\end{center}

Our paper seeks to address this question in a class of continuous games that satisfy a variational stability condition in the spirit of \cite{MZ19,HAM21}.
This class contains all bilinear min-max games (the unconstrained analogue of two-player, zero-sum finite games), cocoercive and monotone games, and it is one of the settings of choice when considering applications to generative models and robust reinforcement learning \cite{KHHR+20,CGFLJ19,HIMM20,LZMJ20}.
As for the noise contaminating the players' gradient feedback, we consider two standard models that build on a classical distinction by \citet{Pol87}:
\begin{enumerate*}
[(\itshape a\upshape)]
\item
\emph{additive};
and
\item
\emph{multiplicative}
\end{enumerate*}
gradient noise.
The first model is more common when dealing with problem-agnostic first-order oracles \cite{Nes04};
the latter arises naturally in the study of randomized coordinate descent \citep{Nes04}, asynchronous player updating schemes \cite{APKMC21}, signal processing and control \cite{SFPP10}, etc.

In this general context, our contributions can be summarized as follows:
\begin{enumerate}
\item
We introduce a learning rate separation mechanism that effectively disjoins the extrapolation and update steps of the \ac{OG} algorithm.
The resulting method, which we call \ac{OG+},
guarantees $\bigoh(\sqrt{\nRuns})$ regret in the presence of additive gradient noise;
however, if the noise is multiplicative and the method is tuned appropriately,
it achieves \emph{constant} $\bigoh(1)$ regret.
\item
On the downside, \ac{OG+} may fail to achieve sublinear
regret in an adversarial environment.
To counter this, we propose a ``primal-dual'' variant of \ac{OG+}, which we call \ac{OptDA+}, and which retains the above properties of \ac{OG+}, while achieving $\bigoh(\sqrt{\nRuns})$ regret in the adversarial case.
\item
Subsequently, to obviate the need for delicate hyperparameter tuning,
we propose a fully adaptive method that 
enjoys nearly the same regret guarantees as mentioned above, without any prior knowledge of the game or of the uncertainties involved.
Interestingly, our method features a trade-off between achieving small regret when facing adversarial opponents and achieving small regret when facing opponents that adopt the same prescribed strategy, which prevents us from obtaining the optimal $\bigoh(\sqrt{\nRuns})$ regret bound in the former situation.
\item
Finally, we complement our analysis with a series of equilibrium convergence results for the range of algorithms presented above under both additive and multiplicative noise.
\end{enumerate}

To the best of our knowledge, our work is the first in the literature to point out 
that constant regret may still be achievable in the presence of stochasticity (even in the simplest case where the noise profile is known in advance).
In this regard, it can be seen as a first estimation of the degree of uncertainty that can enter the process before the aspiration of constant (or polylogarithmic) regret becomes an impossible proposition.\footnote{%
We also note that under the additional assumption of \emph{cocoercivity}, a constant regret bound can be derived from 
\cite[Th. 4.4]{LZMJ20}.
That being said, extending this result to the broader family of variationally stable games that we address here
requires non-trivial modifications (to both the algorithm and the analysis).}

A summary of our results is presented in \cref{tab:result-summary}.
In the paper's appendix, we discuss some further related works that are relevant but not directly related to our work.
We also mention here that our paper focuses on the unconstrained setting, as this simplifies considerably the presentation and treatment of multiplicative noise models.
We defer the constrained case (where players must project their actions to a convex subset of $\R^{\vdim}$), to future work.

%% file: figures/illustration-main.tex
\begin{figure}
    \centering
    \includegraphics[width=0.95\textwidth]{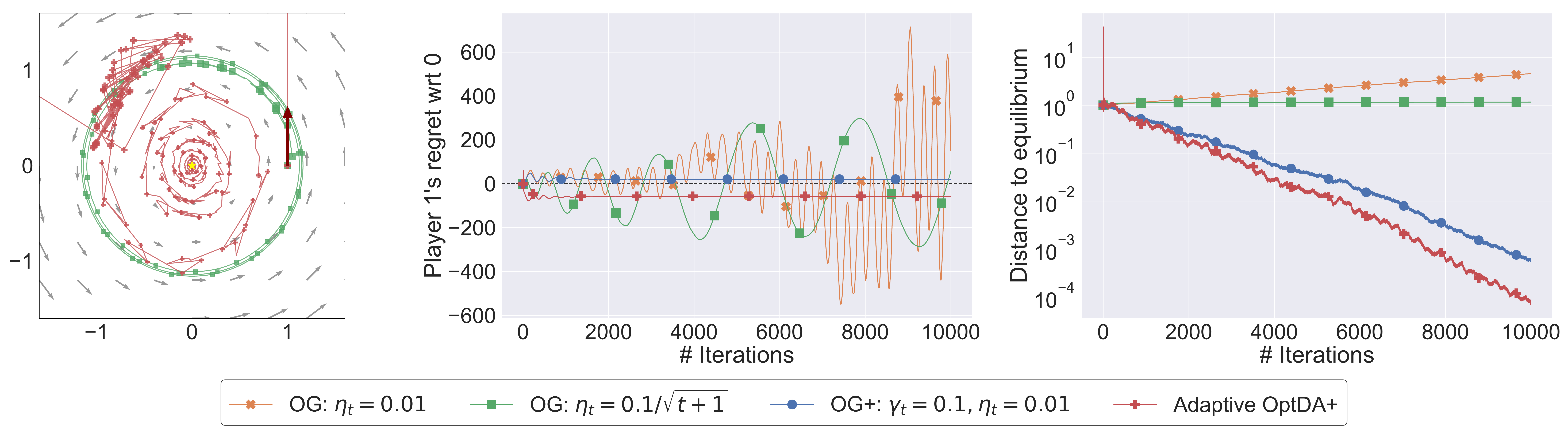}
    \caption[Caption]{The behavior of different algorithms on the game $\min_{\minvar\in\R}\max_{\maxvar\in\R}\minvar\maxvar$ when the feedback is corrupted by noise.
    Left: trajectories of play.
    Center: regret of Player $1$.
    Right: distance to equilibrium.
    Adaptive \ac{OptDA+} is run with $\exponent=1/4$.
    See \cref{ex:bilinear} for the details of the model
    and \cref{apx:figures} for additional figures.}
    \label{fig:illustration}
    \vspace{-1ex}
\end{figure}

%% file: tables/results-summary.tex
\begin{table}
\renewcommand{\arraystretch}{1.3}
\setlength\tabcolsep{0.49em}
\centering
\begin{tabular}{lccccc}
\toprule
	& \multicolumn{1}{c}{Adversarial}
	& \multicolumn{4}{c}{All players run the same algorithm}
	\\
	\cmidrule(lr){2-2}
	\cmidrule(lr){3-6}
	& Bounded feedback
	& \multicolumn{2}{c}{Additive noise}
	& \multicolumn{2}{c}{Multiplicative noise}
	\\
    & Regret
	& Regret
	& Convergence
	& Regret
	& Convergence
	\\
	\midrule
	\ac{OG} 
	& \xmark
	& \xmark
	& \xmark
	& \xmark
	& \xmark
	\\
	\ac{OG+} 
	& \xmark
	& $\sqrt{\run}\log\run$
	& \cmark
	& constant
	& \cmark
	\\
	\ac{OptDA+}
    & $\sqrt{\run}$
    & $\sqrt{\run}$
    & \textendash
    & constant
    & \cmark
	\\
    AdaOptDA+ ($\exponent=1/4$)
    & $\run^{3/4}$
    & $\sqrt{\run}$
    & \textendash
    & constant
    & \cmark
    \\
	\bottomrule
\end{tabular}
\vspace{1em}
\caption{Summary of the results obtained in the paper. The cross \xmark~ indicates a negative result (divergence of trajectory or potentially unbounded regret with decreasing stepsize) while a dash ``\textendash'' means that the behavior of the algorithm is unknown. Our methods improve upon vanilla \ac{OG} by separating the two step-size schedules.}
\label{tab:result-summary}
\vspace{-1.5em}
\end{table}

%% file: sections/preliminaries.tex
Throughout this paper, we focus on deriving optimal regret minimization guarantees for multi-agent game-theoretic settings with noisy feedback.
Starting with the single-agent case, given a sequence of actions $\vt[\action] \in \points = \R^{\vdim}$ and a sequence of loss functions $\obj_{\run} \from \points \to \R$, we define the associated \textit{regret} induced by $\vt[\action]$ relative to a benchmark action $\arpoint\in\points$ as
\vspace{-0.1em}
\begin{equation}
\label{eq:reg}
\vt[\reg][\nRuns](\arpoint)
	= \sum_{\run=\start}^{\nRuns} \bracks{\obj_{\run}(\action_{\run}) - 
	\obj_{\run}(\arpoint)}.
\end{equation}
We then say that learner has \emph{no regret} if 
$\vt[\reg][\nRuns](\arpoint) = o(\nRuns)$ for all $\arpoint\in\points$.
In the sequel, we extend this basic framework to the multi-agent, game-theoretic case, and we discuss the various feedback model available to the optimizer(s).

\vspace{-0.3em}
\para{No-regret learning in games}
The game-theoretic analogue of the above framework is defined as follows.
We consider a finite set of players indexed by $\play \in \players = \{1,\dotsc,\nPlayers\}$, each with their individual action space $\points^{i}=\R^{\vdim^i}$ and their associated loss function $\vp[\loss] \from \points \to \R$, where $\points=\Pi_{\allplayers}\vp[\points]$ denotes the game's joint action space.
For clarity, any ensemble of actions or functions whose definition involves multiple players will be typeset in bold.
In particular, we will write $\jaction=(\vw[\action], \vw[\jaction][\playexcept])\in \points$ for the action profile of all players, where $\vw[\action]$ and $\vw[\jaction][\playexcept]$ respectively denote the action of player $\play$ and the joint action of all players other than $\play$. 
In this way, each player $\play\in\players$ incurs at round $\run$ a loss $\vp[\loss](\vt[\jaction])$ which is determined not only by their individual action $\vpt[\action]$, but also by the actions $\vpt[\jaction][\playexcept]$ of all other players.
Thus, by drawing a direct link with \eqref{eq:reg}, given a sequence of play $\vpt[\action]$, the \emph{individual regret} of each player $\play\in\players$ is defined as
\vspace{-0.1em}
\begin{equation}
\label{eq:reg-ind}
\vpt[\reg][\play][\nRuns](\vp[\arpoint])
	= \sum_{\run=\start}^{\nRuns} \loss(\point_{\run}^{i},\jaction^{\playexcept}_{\run})-\loss(\vp[\arpoint],\jaction^{\playexcept}_{\run}),
\end{equation}
%
From a static viewpoint, the most widely spread solution concept in game theory is that of a \emph{Nash equilibrium}, \ie a state from which no player has incentive to deviate unilaterally.
Formally, a point $\bb{\sol}\in\points$ is a \acl{NE} if for all $\play\in\players$ and all $\vp[\point]\in\vp[\points]$, we have
$\vp[\loss](\oneandother[\sol][\sol[\jaction]])\le\vp[\loss](\oneandother[\point][\sol[\jaction]])$.
In particular, if the players' loss functions are assumed individually convex (see below), \aclp{NE} coincide precisely with the zeros of the players' individual gradient field, denoted by $\vp[\vecfield]=\grad_{\vp[\action]}\vp[\loss]$.
That is, $\sol[\jaction]$ is a \acl{NE} if and only if $\jvecfield(\sol[\jaction]) = 0$.
We will make the following blanket assumptions for all this:


\begin{assumption}[Convexity and Smoothness]
\label{asm:lips}
For all $\allplayers$,
$\vp[\loss](\cdot, \vp[\jaction][\playexcept])$ is convex at all $\vp[\jaction][\playexcept]$
and the individual gradient of each player
$\grad_{\vp[\action]}\vp[\loss]$ is $\lips$-Lipschitz continuous.
\end{assumption}


\begin{assumption}[Variational Stability]
\label{asm:VS}
The solution set $\sols = \setdef{\jaction\in\points}{\jvecfield(\jaction)=0}$ of the game is nonempty, and for all $\jaction\in\points$, $\jsol\in\sols$, we have $\product{\jvecfield(\jaction)}{\jaction-\jsol} =\sum_{i\in \mathcal{N}}\product{\vp[\vecfield](\jaction)}{\point^{i}-\sol^{i}}\ge0$.
\end{assumption}

The convexity requirement in \cref{asm:lips} is crucial in the literature of online learning;
otherwise, it is not possible to transform iterative gradient bounds to bona fide regret guarantees.
In a similar vein, variational stability can be seen as a variant of the convexity assumption for multi-agent environments, where unilateral convexity assumptions do not suffice to give rise to a learnable game \textendash~ for example, finite games are unilaterally linear, but finding a Nash equilibrium of a finite game is a PPAD-complete problem \cite{DGP09}.
Our work thus focuses on games that satisfy the variational stability condition.
Some important families of games that are covered by this criterion are monotone games (\ie $\jvecfield$ is monotone), which in their turn include convex-concave zero-sum games, zero-sum polymatrix games, 
Cournot oligopolies, 
etc.

It is also worth noting that several recent works \cite{DFG21,ADFF+21,farina2022near} have managed to bypass \cref{asm:VS}
when the players have access to perfect feedback;
whether these techniques are applicable in the stochastic setup is an open question.
In any case, \cref{asm:VS} seems crucial for the last-iterate convergence presented in \cref{sec:trajectory}.

\para{Oracle feedback and noise models}

In terms of feedback, we will assume that players have access to noisy estimates of their individual payoff gradients, and we will consider two noise models, \textit{additive noise} and \textit{multiplicative noise}.
To illustrate the difference between these two models, suppose we wish to estimate the value of some quantity $v\in\R$.
Then, an estimate of $v$ with additive noise is a random variable $\hat v_{\mathrm{add}}$ of the form $\hat v_{\mathrm{add}} = v + \xi_{\mathrm{add}}$ for some zero-mean noise variable $\xi_{\mathrm{add}}$;
analogously, a multiplicative noise model for $v$ is a random variable of the form $\hat v_{\mathrm{mult}} = v(1 + \xi_{\mathrm{mult}})$ for some zero-mean noise variable $\xi_{\mathrm{mult}}$.
The two models can be compared directly via the additive representation of the multiplicative noise model as $\hat v_\mathrm{mult} = v + \xi_\mathrm{mult} v$, which gives $\mathrm{Var}[\xi_\mathrm{add}] = v^2 \mathrm{Var}[\xi_\mathrm{mult}]$.

With all this in mind, we will consider the following oracle feedback model:
let $\vpt[\gvec]=\vp[\vecfield](\vt[\jaction])+\vpt[\noise]$ denote the gradient feedback to player $\play$ at round $\run$, where $\vpt[\noise]$ represents the aggregate measurement error relative to $\vp[\vecfield](\vt[\jaction])$.
Then, with $\seqinf[\filter]$ denoting the natural filtration associated to $\seqinf[\jaction]$ and $\ex_{\run}[\cdot]=\exof{\cdot\given{\vt[\filter]}}$ representing the corresponding conditional expectation, we make the following standard assumption for the measurement error vector $\vt[\jnoise]=(\vpt[\noise])_{\allplayers}$.
\begin{assumption}
\label{asm:noises}
The noise vector $(\vt[\jnoise])_{\run\in\N}$ satisfies the following requirements for some $\noisedev,\noisedevmul\ge0$.
\begin{enumerate}[(a), leftmargin=*]
    \item \textit{Zero-mean:} 
    For all $\allplayers$ and $\run\in\N$, $\ex_{\run}[\vpt[\noise]]=0$.
    \label{asm:noises-unbiased}
    \item \textit{Finite variance:}
    For all $\allplayers$ and $\run\in\N$,
    $\ex_{\run}[\norm{\vpt[\noise]}^2]\le\noisevar+\noisecontrol\norm{\vp[\vecfield](\vt[\jaction])}^2$.
    \label{asm:noises-variance}
\end{enumerate}
\end{assumption}

As an example of the above,
the case $\noisedev,\noisedevmul=0$ corresponds to ``perfect information'', \ie when players have full access to their payoff gradients.
The case $\noisedev > 0$, $\noisedevmul=0$, is often referred to as ``absolute noise'', and it is a popular context-agnostic model for stochastic first-order methods, \cf \cite{NJLS09,JNT11} and references therein.
Conversely, the case $\noisedev=0$, $\noisedevmul>0$, is sometimes called ``relative noise'' \cite{Pol87}, and it is widely used as a model for randomized coordinate descent methods \cite{Nes04}, randomized player updates in game theory \cite{APKMC21}, physical measurements in signal processing and control \cite{SFPP10}, etc.
In the sequel, we will treat both models concurrently, and we will use the term ``noise'' to tacitly refer to the presence of both additive and multiplicative components.

%% file: sections/motivating.tex
To illustrate some of the difficulties faced by first-order methods in a game-theoretic setting, consider the standard bilinear problem $\min_{\minvar\in\R}\max_{\maxvar\in\R}\minvar\maxvar$, \ie $\vp[\loss][1](\minvar,\maxvar)=\minvar\maxvar=-\vp[\loss][2](\minvar,\maxvar)$.
This simple game has a unique Nash equilibrium at $(0,0)$ but, despite this uniqueness, it is well known that standard gradient descent/ascent methods diverge on this simple problem \cite{DISZ18,MLZF+19}.
To remedy this failure, one popular solution consists of incorporating an additional \textit{extrapolation step} at each iteration of the algorithm, leading to the \acdef{OG} method
%
\begin{equation}
    \notag
    \vptupdate[\action] = 
    \vpt[\action] - 2\vptupdate[\step]\vpt[\gvec]
    +\vpt[\step]\vptlast[\gvec],
\end{equation}
where $\vpt[\step]$ is player $\play$'s learning rate at round $\run$.
For posterity, it will be convenient to introduce the auxiliary iterate $\vpt[\state]$ and write $\vptinter=\vpt[\action]$
for the actual sequence of actions.
The above update rule then becomes
\begin{equation}
    \label{eq:OG}
    \tag{OG}
    \vptinter = \vpt[\state] - \vpt[\step]\vptlast[\gvec],
    ~~~~~
    \vptupdate = \vpt[\state] - \vptupdate[\step]\vpt[\gvec].
\end{equation}
This form of the algorithm effectively decouples the learner's \textit{extrapolation step} (performed with $\vptlast[\gvec]$, which acts here as an \emph{optimistic} guess for the upcoming feedback),
and the bona fide \textit{update step}, which exploits the received feedback $\vpt[\gvec]$ to update the player's action state from $\vpt[\state]$ to $\vptupdate[\state]$.
This mechanism helps the players attain
\begin{enumerate*}
[\itshape a\upshape)]
\item
\emph{lower regret} when their utilities vary slowly (from an online learning viewpoint) \cite{CYLM+12,RS13-NIPS};
and
\item
\emph{near-constant regret} when all players employ the said algorithm in certain classes of games \cite{DFG21,ADFF+21,AFKL+22,FLLK22,HAM21}.
\end{enumerate*}

However, the above guarantees concern only the case of \emph{perfect gradient feedback}, and may fail completely when the feedback is contaminated by noise, as illustrated in the following example.
\begin{example}
\label{ex:bilinear}
Suppose that the game's objective is an expectation over $\sadobj_1(\minvar,\maxvar)=3\minvar\maxvar$ and $\sadobj_2(\minvar,\maxvar)=-\minvar\maxvar$ so that $\vp[\loss][1]=-\vp[\loss][2]=(\sadobj_1+\sadobj_2)/2$.
At each round, we randomly draw $\sadobj_1$ or $\sadobj_2$ with probability $1/2$ and return the gradient of the sampled function as feedback.
\cref{asm:noises} is clearly satisfied here with $\noisedev=0$ and $\noisedevmul=2$,
\ie the noise is multiplicative;
however, as shown in \cref{fig:illustration}, running \eqref{eq:OG} with either constant or decreasing learning rate leads to
\begin{enumerate*}
[\itshape i\upshape)]
\item divergent trajectories of play;
and
\item regret oscillations that grow linearly or even superlinearly in magnitude over time.\footnote{%
By superlinear we mean that the regret grows faster than $\Theta(\nRuns)$, and this is possible here because neither the action set nor the feedback magnitude is bounded.
}

\end{enumerate*}
\end{example}
%
In view of the above negative results, we propose in the next section a simple fix of the algorithm that allows us to retain its constant regret guarantees in the presence of multiplicative noise.

%% file: sections/regret.tex
In this section, we introduce \ac{OG+} and \ac{OptDA+}, our backbone algorithms for learning under uncertainty, and we present their guarantees in different settings.
All proofs are deferred to the appendix.

\paragraph{Learning rate separation and the role of averaging\afterhead}
Viewed abstractly, the failure of \ac{OG}
in the face of uncertainty should be attributed to its inability of separating noise from the expected variation of utilities.
In fact, in a noisy environment, 
the two consecutive pieces of feedback are only close \emph{in expectation},
so a player can only exploit this similarity when the noise is mitigated appropriately.

To overcome this difficulty, we adopt a learning rate separation strategy originally proposed for the \ac{EG} algorithm by \citet{HIMM20}.
The key observation here is that by taking a larger extrapolation step, the noise effectively becomes an order of magnitude smaller relative to the expected variation of utilities.
We refer to this generalization of \ac{OG} as \ac{OG+}, and we define it formally as
\begin{equation}
    \label{OG+}
    \tag{OG+}
    \vptinter = \vpt[\state] - \vpt[\stepalt]\vptlast[\gvec],
    ~~~~~
    \vptupdate = \vpt[\state] - \vptupdate[\step]\vpt[\gvec],
\end{equation}
where $\vpt[\stepalt]\ge\vpt[\step]>0$ are the player's learning rates (assumed $\filter_{\run-1}$-measurable throughout the sequel).
Nonetheless, the design of \ac{OG+} is somehow counter-intuitive because the players' feedback enters the algorithm with decreasing weights.
This feature opens up the algorithm to adversarial attacks that can drive it to a suboptimal regime in early iterations, as formally shown in \cite[Thm.~3]{OP18}.

To circumvent this issue, we also consider a \textit{dual averaging} variant of \ac{OG+} that we refer to as \ac{OptDA+}, and which treats the gradient feedback used to update the players' chosen actions with the same weight.
Specifically, \ac{OptDA+} combines the mechanisms of optimism \cite{DISZ18,RS13-NIPS}, dual averaging \cite{Nes09,Xia10,HAM21}, and learning rate separation \cite{HIMM20} as follows
%
\begin{equation}
    \label{OptDA+}
    \tag{OptDA+}
    \vptinter = \vpt[\state] - \vpt[\stepalt]\vptlast[\gvec],
    ~~~~~
    \vptupdate = \vpt[\state][\play][1]
    -\vptupdate[\step]\sum_{\runalt=1}^{\run}\vpt[\gvec][\play][\runalt].
\end{equation}
%
As we shall see below, these mechanisms dovetail in an efficient manner and allow the algorithm to achieve sublinear regret even in the adversarial regime.
[Of course, \ac{OG+} and \ac{OptDA+} coincide when the update learning rate $\vpt[\step]$ is taken constant.]

\paragraph{Quasi-descent inequality\afterhead}
\label{subsec:quasi-descent}

Before stating our main results on the regret incurred by \ac{OG+} and \ac{OptDA+}, we present the key quasi-descent inequality that underlies our analysis, as it provides theoretical evidence on how the separation of learning rates can lead to concrete performance benefits.


\begin{lemma}
\label{main:lem:quasi-descent-individual}
Let \cref{asm:noises,asm:lips} hold and all players run either \eqref{OG+} or \eqref{OptDA+} with non-increasing learning rate sequences $\vpt[\stepalt]$, $\vpt[\step]$.
Then, for all $\allplayers$, $\run\ge2$, and $\vp[\arpoint]\in\vp[\points]$, we have
\begin{subequations}
\begin{align}
\label{eq:quasi-descent-individual-a}
    \ex_{\run-1}\Bigg[
    \frac{\norm{\vptupdate[\state]-\vp[\arpoint]}^2}
    {\vptupdate[\step]}
    \Bigg]
    \le
    \ex_{\run-1}\Bigg[&
    \frac{\norm{\vpt[\state]-\vp[\arpoint]}^2}{\vpt[\step]}
    +
    \left(
    \frac{1}{\vptupdate[\step]}-\frac{1}{\vpt[\step]}
    \right)\norm{\vpt[\uvec]-\vp[\arpoint]}^2
    \\
\label{eq:quasi-descent-individual-b}
    &-
    2\product{\vp[\vecfield](\inter[\jstate])}{\vptinter[\state]-\vp[\arpoint]}
    \\
\label{eq:quasi-descent-individual-c}
    &-
    \vpt[\stepalt](\norm{\vp[\vecfield](\inter[\jstate])}^2
    +\norm{\vp[\vecfield](\past[\jstate])}^2)
    \\
\label{eq:quasi-descent-individual-d}
    &
    -\norm{\vpt[\state]-\vptupdate[\state]}^2/2\vpt[\step]
    +\vpt[\stepalt]\norm{\vp[\vecfield](\inter[\jstate])-\vp[\vecfield](\past[\jstate])}^2
    \\
\label{eq:quasi-descent-individual-e}
    &
    +(\vpt[\stepalt])^2\lips\norm{\vptpast[\noise]}^2
    +\lips
    \norm{\past[\jnoise]}_{(\vt[\jstep]+\vt[\jstepalt])^2}^2
    +2\vpt[\step]\norm{\vpt[\gvec]}^2
    \Bigg],
\end{align}
\end{subequations}
where 
\begin{enumerate*}[\itshape i\upshape)]
    \item
    $\norm{\past[\jnoise]}_{(\vt[\jstep]+\vt[\jstepalt])^2}^2
    \defeq \sumplayers[\playalt](\vpt[\step][\playalt]+\vpt[\stepalt][\playalt])^2\norm{\vptpast[\noise][\playalt]}^2$,
    and
    \item $\vpt[\uvec]=\vpt[\state]$ if player $\play$ runs \eqref{OG+} and $\vpt[\uvec]=\vpt[\state][\play][1]$ if player $\play$ runs \eqref{OptDA+}.
\end{enumerate*}
\end{lemma}

\cref{main:lem:quasi-descent-individual} indicates how the (weighted) distance between the player's chosen actions and a fixed benchmark action evolves over time.
In order to provide some intuition on how this inequality will be used to derive our results, we sketch below the role that each term plays in our analysis.
%
\begin{enumerate}[leftmargin=*]
    \item Thanks to the convexity of the players' loss functions, the regret of each player can be bounded by the sum of the pairing terms in \eqref{eq:quasi-descent-individual-b}. On the other hand, taking $\jsol\in\sols$, $\vp[\arpoint]=\vp[\sol]$, and summing from $\play=1$ to $\nPlayers$, we obtain $-2\product{\jvecfield(\inter[\jstate])}{\inter[\jstate]-\jsol}$, which is non-positive by \cref{asm:VS}, and can thus be dropped from the inequality.
    \item The weighted squared distance to $\vp[\arpoint]$, \ie $\norm{\vpt[\state]-\vp[\arpoint]}^2/\vpt[\step]$, telescopes when controlling the regret (\cref{subsec:regret-nonadapt}) and serves as a Lyapunov function for equilibrium convergence (\cref{sec:trajectory}).
    \item The negative term in \eqref{eq:quasi-descent-individual-c} provides a consistent negative drift that partially cancels out the noise.
    \item The difference in \eqref{eq:quasi-descent-individual-d} can be bounded using the smoothness assumption and leaves out terms that are in the order of $\vpt[\stepalt](\vpt[\stepalt][\playalt])^2$.
    \item Line \eqref{eq:quasi-descent-individual-e} contains a range of positive terms of the order $(\vpt[\stepalt][\playalt])^2+\vpt[\step]$.
    To ensure that they are sufficiently small with respect to the decrease of \eqref{eq:quasi-descent-individual-c}, both $(\vpt[\stepalt][\playalt])_{\allplayers[\playalt]}$ and $\vpt[\step]/\vpt[\stepalt]$ should be small.
    Applying \cref{asm:noises} gives $\ex[\norm{\vpt[\gvec]}^2]\le\ex[(1+\noisecontrol)\norm{\vp[\vecfield](\inter[\jstate])}+\noisevar]$, revealing that $\vpt[\stepalt]/\vpt[\step]$ needs to be at least in the order of $(1+\noisecontrol)$.
    \item Last but not least, 
    $(1/\vptupdate[\step]-1/\vpt[\step])\norm{\vpt[\uvec]-\vp[\arpoint]}^2$ simply telescopes for \ac{OptDA+} (in which case $\vpt[\uvec]=\vpt[\state][\play][1]$) but is otherwise difficult to control for \ac{OG+} when $\vptupdate[\step]$ differs from $\vpt[\step]$.
    This additional difficulty forces us to use a global learning rate common across all players when analysing \ac{OG+}.
\end{enumerate}

To summarize, \ac{OG+} and \ac{OptDA+} are more suitable for learning in games with noisy feedback because the scale separation between the extrapolation and the update steps delivers a consistent negative drift \eqref{eq:quasi-descent-individual-c}  that is an order of magnitude greater relative to the deleterious effects of the noise.
We will exploit this property to derive our main results for \ac{OG+} and \ac{OptDA+} below.

\vspace{-0.4em}
\paragraph{Constant regret under uncertainty\afterhead}
\label{subsec:regret-nonadapt}


We are now in a position to state our regret guarantees:

\begin{theorem}
\label{thm:OG+-regret}
Suppose that \cref{asm:noises,asm:lips,asm:VS} hold and all players run \eqref{OG+} with non-increasing learning rate sequences $\vt[\stepalt]$ and $\vt[\step]$ such that
\begin{equation}
    \label{OG+-lr}
    \vt[\stepalt]
    \le \min
    \left(\frac{1}{3\lips\sqrt{2\nPlayers(1+\noisecontrol)}}
    ,\frac{1}{2(4\nPlayers+1)\lips\noisecontrol}\right)
    \quad
    \text{and}
    \quad
    \vt[\step]\le
    \frac{\vt[\stepalt]}{2(1+\noisecontrol)}
    \quad
    \text{for all $\run\in\N$}.
\end{equation}
Then, for all $\allplayers$ and all $\vp[\arpoint]\in\vp[\points]$, we have
\begin{enumerate}[(a), leftmargin=*]
    \item
    \label{thm:OG+-regret-add-main}
    If $\vt[\stepalt]=\bigoh(1/(\run^{\frac{1}{4}}\sqrt{\log\run}))$ and $\vt[\step]=\Theta(1/(\sqrt{\run}\log\run))$, then
    $\ex\left[\vpt[\reg][\play][\nRuns](\vp[\arpoint])
    \right]
    =\tbigoh(\sqrt{\nRuns}).$
    \item
    \label{thm:OG+-regret-mul-main}
    If the noise is multiplicative and the learning rates are constant, then
    $\ex\left[\vpt[\reg][\play][\nRuns](\vp[\arpoint])
    \right]
    = \bigoh(1)$.
\end{enumerate}
\end{theorem}

The first part of \cref{thm:OG+-regret} guarantees the standard $\tbigoh(\sqrt{\nRuns})$ regret in the presence of additive noise, in accordance with existing results in the literature.
What is far more surprising is the second part of \cref{thm:OG+-regret} which shows that when the noise is multiplicative (\ie when $\noisedev=0$), it is \emph{still} possible to achieve constant regret.
This represents a dramatic improvement in performance, which we illustrate in \cref{fig:illustration}:
by simply taking the extrapolation step to be $10$ times larger, the player's regret becomes completely stabilized.
In this regard, \cref{thm:OG+-regret} provides fairly conclusive evidence that having access to exact gradient payoffs \emph{is not} an absolute requisite for achieving constant regret in a game-theoretic context.

On the downside, the above result requires all players to use the same learning rate sequences, a technical difficulty that we overcome below by means of the dual averaging mechanism of \ac{OptDA+}.

\begin{theorem}
\label{thm:OptDA+-regret}
Suppose that \cref{asm:noises,asm:lips,asm:VS} hold
and all players run \eqref{OptDA+} with non-increasing learning rate sequences $\vpt[\stepalt]$ and $\vpt[\step]$
such that
%
\begin{equation}
    \label{OptDA+-lr}
    \vpt[\stepalt]
    \le \frac{1}{2\lips} \min
    \left(\frac{1}{\sqrt{3\nPlayers(1+\noisecontrol)}}
    ,\frac{1}{(4\nPlayers+1)\noisecontrol}\right)
    \quad
    \text{and}
    \quad
    \vpt[\step]
    \le\frac{\vpt[\stepalt]}{4(1+\noisecontrol)}
    \quad
    \text{for all $\run\in\N$, $\play\in\players$}.
\end{equation}
Then, for any $\allplayers$ and $\vp[\arpoint]\in\vp[\points]$, we have:
\begin{enumerate}[(a), leftmargin=*]
    \item
    \label{thm:OptDA+-regret-add-main}
    If $\vpt[\stepalt][\playalt] = \bigoh(1/\run^{\frac{1}{4}})$
    and $\vpt[\step][\playalt] = \Theta(1/\sqrt{\run})$ for all $\allplayers[\playalt]$, then
    $\ex\left[\vpt[\reg][\play][\nRuns](\vp[\arpoint])
    \right]
    =\bigoh(\sqrt{\nRuns}).$
    \item
    \label{thm:OptDA+-regret-mul-main}
    If the noise is multiplicative and the learning rates are constant, then
    $\ex\left[\vpt[\reg][\play][\nRuns](\vp[\arpoint])
    \right]
    = \bigoh(1)$.
\end{enumerate}
\end{theorem}

The similarity between \cref{thm:OG+-regret,thm:OptDA+-regret} suggests that \ac{OptDA+} enjoys nearly the same regret guarantee as \ac{OG+} while allowing for the use of \textit{player-specific} learning rates.
As \ac{OptDA+} and \ac{OG+} coincide when run with constant learning rates, \cref{thm:OG+-regret}\ref{thm:OG+-regret-mul-main} is in fact a special case of \cref{thm:OptDA+-regret}\ref{thm:OptDA+-regret-mul-main}.
However, when the algorithms are run with decreasing learning rates, they actually lead to different trajectories.
In particular, when the feedback is corrupted by additive noise, this difference translates into the removal of logarithmic factors in the regret bound.
More importantly, as we show below, it also helps to achieve sublinear regret when the opponents do not follow the same learning strategy, \ie in the fully arbitrary, adversarial case.

\begin{proposition}
\label{prop:OptDA+-regret-adversarial}
Suppose that \cref{asm:noises} holds and player $\play$ runs \eqref{OptDA+} with non-increasing learning rates $\vpt[\stepalt] = \Theta(1/\run^{\frac{1}{2}-\exponent})$ and
$\vpt[\step] = \Theta(1/\sqrt{\run})$
for some $\exponent\in[0,1/4]$.
If
$\sup_{\vp[\action]\in\vp[\points]}\norm{\vp[\vecfield](\vp[\action])}<+\infty$, we have $\ex[\vpt[\reg][\play][\nRuns](\vp[\arpoint])] = \bigoh(\nRuns^{\frac{1}{2}+\exponent})$ for every benchmark action $\vp[\arpoint]\in\vp[\points]$.
\end{proposition}

We introduce the exponent $\exponent$ in
\cref{prop:OptDA+-regret-adversarial} because, as suggested by \cref{thm:OptDA+-regret}\ref{thm:OptDA+-regret-add-main}, the whole range of $\exponent\in[0,1/4]$ leads to the optimal $\bigoh(\sqrt{\nRuns})$ regret bound for additive noise when all the players adhere to the use of \ac{OptDA+}.
However, it turns out that taking smaller $\exponent$ (\ie smaller extrapolation step) is more favorable in the adversarial regime.
This is because arbitrarily different successive feedback may make the extrapolation step harmful rather than helpful.
On the other hand, our previous discussion also suggests that taking larger $\exponent$ (\ie larger extrapolation steps), should be more beneficial when all the players use \ac{OptDA+}.
We will quantify this effect in \cref{sec:trajectory};
however, before doing so, we proceed in the next section to show how the learning rates of \cref{prop:OptDA+-regret-adversarial} can lead to the design of a fully adaptive, parameter-agnostic algorithm.


%% file: sections/adaptive.tex
So far, we have focused exclusively on algorithms run with \emph{predetermined} learning rates, whose tuning requires knowledge of the various parameters of the model.
Nonetheless, even though a player might be aware of their own loss function, there is little hope that the noise-related parameters are also known by the player.
Our goal in this section will be to address precisely this issue through the design of \emph{adaptive} methods enjoying the following desirable properties:
\begin{itemize}[leftmargin=*]
    \item
    The method should be implementable by every individual player using only local information and without any prior knowledge of the setting's parameters (for the noise profile and the game alike).
    \item
    The method should guarantee sublinear individual regret against any bounded feedback sequence.
    \item
    When employed by all players, the method should guarantee $\bigoh(\sqrt{\nRuns})$ regret under additive noise and $\bigoh(1)$ regret under multiplicative noise.
\end{itemize}
%
In order to achieve the above,
inspired by the learning rate requirements of \cref{thm:OptDA+-regret} and \cref{prop:OptDA+-regret-adversarial},
we fix $\exponent\in(0,1/4]$ and consider the following Adagrad-style~\cite{DHS11} learning rate schedule.
\begin{equation}
    \label{adaptive-lr}
    \tag{Adapt}
    \begin{aligned}
    \vpt[\stepalt]
    =
    \frac{1}{
    \left(
    1+\sum_{\runalt=1}^{\run-2}
    \norm{\vpt[\gvec][\play][\runalt]}^2
    \right)^{\frac{1}{2}-\exponent}},
    ~~~~
    \vpt[\step]
    =
    \frac{1}{\sqrt{
     1+\sum_{\runalt=1}^{\run-2}
    \left(
    \norm{\vpt[\gvec][\play][\runalt]}^2
    +\norm{\vpt[\state][\play][\runalt]
    -\vptupdate[\state][\play][\runalt]}^2
    \right)
    }}.
    \end{aligned}
\end{equation}
As in Adagrad, the sum of the squared norm of the feedback appears in the denominator.
This helps controlling the various positive terms appearing in \cref{main:lem:quasi-descent-individual}, such as $\lips\norm{\past[\jnoise]}_{(\vt[\jstep]+\vt[\jstepalt])^2}^2$ and $2\vpt[\step]\norm{\vpt[\gvec]}^2$.
Nonetheless, this sum is not taken to the same exponent in the definition of the two learning rates.
This scale separation ensures that the contribution of 
the term $-\vpt[\stepalt]\norm{\vp[\vecfield](\inter[\jstate])}^2$ appearing in \eqref{eq:quasi-descent-individual-c} remains negative, and it is the key for deriving constant regret under multiplicative noise.
As a technical detail, the 
term $\norm{\vpt[\state][\play][\runalt]-\vptupdate[\state][\play][\runalt]}^2$ is involved in the definition of $\vpt[\step]$
for controlling the difference of \eqref{eq:quasi-descent-individual-d}.
Finally, we do not include the previous received feedback $\vptlast[\gvec]$ in the definition of $\vpt[\stepalt]$ and $\vpt[\step]$. This makes these learning rates $\last[\filter]$-measurable,
which in turn implies $\ex[\vpt[\stepalt]\vpt[\step]\vptpast[\noise]]=0$.

From a high-level perspective, the goal with \eqref{adaptive-lr} is to recover automatically the learning rate schedules of \cref{thm:OptDA+-regret}.
This in particular means that $\vpt[\stepalt]$ and $\vpt[\step]$ should at least be in the order of $\Omega(1/\run^{\frac{1}{2}-\exponent})$ and $\Omega(1/\sqrt{\run})$, suggesting the following boundedness assumptions on the feedback.


\begin{assumption}
\label{asm:boundedness}
There exists $\gbound,\noisebound\ge0$ such that \begin{enumerate*}[\itshape i\upshape)]
    \item $\norm{\vp[\vecfield](\vp[\action])}\le\gbound$ for all $\allplayers$, $\vp[\action]\in\vp[\points]$; and
    \item $\norm{\vpt[\noise]}\le\noisebound$ for all $\allplayers$, $\run\in\N$ with probability $1$.
\end{enumerate*}
\end{assumption}

These assumptions are standard in the literature on adaptive methods, \cf \cite{BL19,ABM19,KLBC19,EN22}.

\paragraph{Regret\afterhead}
We begin with the method's fallback guarantees, deferring all proofs to the appendix.

\begin{proposition}
\label{prop:OptDA+-adapt-regret-adversarial}
Suppose that \cref{asm:boundedness} holds and a player $\play\in\players$ follows \eqref{OptDA+} with learning rates given by \eqref{adaptive-lr}.
Then, for any benchmark action $\vp[\arpoint]\in\vp[\points]$, we have
$\ex[\vpt[\reg][\play][\nRuns](\vp[\arpoint])] = \bigoh(\nRuns^{\frac{1}{2}+\exponent})$.
\end{proposition}

\cref{prop:OptDA+-adapt-regret-adversarial}
provides exactly the same rate as \cref{prop:OptDA+-regret-adversarial},
illustrating in this way the benefit of taking a smaller $\exponent$ for achieving smaller regret against adversarial opponents.
Nonetheless, as we see below, taking smaller $\exponent$ may incur higher regret when adaptive OptDA+ is employed by all players.
In particular, we require $\exponent>0$ in order to obtain constant regret under multiplicative noise, and this prevents us from obtaining the optimal $\bigoh(\sqrt{\nRuns})$ regret in fully adversarial environments.

\begin{theorem}
\label{thm:OptDA+-adapt-regret}
Suppose that \cref{asm:noises,asm:lips,asm:VS,asm:boundedness} hold
and all players run \eqref{OptDA+} with learning rates given by \eqref{adaptive-lr}.
Then, for any $\play\in\players$ and point $\vp[\arpoint]\in\vp[\points]$ , we have $\ex[\vpt[\reg][\play][\nRuns](\vp[\arpoint])]=\bigoh(\sqrt{\nRuns})$.
Moreover, if the noise is multiplicative ($\noisedev=0$), we have $\ex[\vpt[\reg][\play][\nRuns](\vp[\arpoint])]=\bigoh(\exp(1/(2\exponent)))$.
\end{theorem}

The proof of \cref{thm:OptDA+-adapt-regret} is based on \cref{main:lem:quasi-descent-individual};
we also note that the $\bigoh(\sqrt{\nRuns})$ regret guarantee can in fact be derived for any $\exponent\le1/4$ (even negative ones).
The main difficulty here consists in bounding \eqref{eq:quasi-descent-individual-d}, which does not directly cancel out since $\vpt[\stepalt]\vpt[\step]$ might not be small enough.
To overcome this challenge, we have involved the squared difference $\norm{\vpt[\state][\play][\runalt]-\vptupdate[\state][\play][\runalt]}^2$ in the definition of $\vpt[\step]$ so that the sum of these terms cannot be too large when $\vpt[\step]$ is not small enough.
More details on this aspect can be found in the proof of \cref{lem:OptDA+-adapt-sum-cancel-terms} in the appendix.


Importantly, the $\bigoh(\sqrt{\nRuns})$ guarantee above does not depend on the choice of $\exponent$.
This comes in sharp contrast to the constant regret bounds (in $\nRuns$) that we obtain for multiplicative noise.
In fact, a key step for proving this is to show that for some
(environment-dependent) constant $\Cst$, we have
\begin{equation}
    \label{eq:adapt-main-ineq}
    \sumplayers
    \ex\left[
    \left(1+\sum_{\runalt=1}^{\run}\norm{\vpt[\gvec][\play][\runalt]}^2
    \right)^{\frac{1}{2}+\exponent}
    \right]
    \le
    \Cst
    \sumplayers
    \ex\left[
    \sqrt{1+\sum_{\runalt=1}^{\run}\norm{\vpt[\gvec][\play][\runalt]}^2}
    \right]
    \qquad
    \text{for all $\run\in\N$}
\end{equation}
This inequality is derived from \cref{main:lem:quasi-descent-individual}
by carefully bounding \eqref{eq:quasi-descent-individual-a}, \eqref{eq:quasi-descent-individual-d}, \eqref{eq:quasi-descent-individual-e} from above and bounding \eqref{eq:quasi-descent-individual-b}, \eqref{eq:quasi-descent-individual-c} from below. 
Applying Jensen's inequality, we then further deduce that the \acl{RHS} of inequality \eqref{eq:adapt-main-ineq} is bounded by some constant.
This constant, however, is exponential in $1/\exponent$.
This leads to an inherent trade-off in the choice of $\exponent$:
larger values of $\exponent$ favor the situation where all players adopt adaptive OptDA+ under multiplicative noise, while smaller values of $\exponent$ provide better fallback guarantees in adversarial environments.


%

%% file: sections/trajectory.tex
In this section, we shift our focus to the analysis of the joint trajectory of play when all players follow the same learning strategy.
We derive the convergence of the trajectory of play induced by the algorithms (\cf \cref{fig:illustration})
and provide bounds on the sum of the players' payoff gradient norms
$\sum_{\run=1}^{\nRuns}\norm{\jvecfield(\inter[\jstate])}^2$.
This may be regarded as a relaxed convergence criterion, and by the design of the algorithms, a feedback sequence of smaller magnitude also suggests a more stable trajectory.

\paragraph{Convergence of trajectories under multiplicative noise\afterhead}

When the noise is multiplicative, its effect is in expectation absorbed by the progress brought by the extrapolation step.
We thus expect convergence results that are similar to the noiseless case. This is confirmed by the following theorem.

\begin{theorem}
\label{thm:cvg-relative-noise}
Suppose that \cref{asm:noises,asm:lips,asm:VS} hold with $\noisedev=0$
and
all players run \eqref{OG+} / \eqref{OptDA+} with learning rates given in \cref{thm:OptDA+-regret}\ref{thm:OptDA+-regret-mul-main}.\footnote{Recall that OG+ and OptDA+ are equivalent when run with constant learning rates.}
Then,
$\inter[\jstate]$ converges almost surely to a Nash equilibrium and enjoys the stabilization guarantee $\sum_{\run=1}^{+\infty}\ex[\norm{\jvecfield(\inter[\jstate])}^2]<+\infty$.

\end{theorem}
\vspace{-1em}
\begin{proof}[Idea of proof]
The proof of \cref{thm:cvg-relative-noise} follows the following steps.
\begin{enumerate}[leftmargin=*, itemsep=0pt]
    \item We first show $\sum_{\run=1}^{+\infty}\ex[\norm{\jvecfield(\inter[\jstate])}^2]<+\infty$ using \cref{main:lem:quasi-descent-individual}.
    This implies $\sum_{\run=1}^{\infty}\norm{\jvecfield(\inter[\jstate])}^2$ is finite almost surely, and thus with probability $1$, $\norm{\jvecfield(\inter[\jstate])}$ converges to $0$ and all cluster point of $(\inter[\jstate])_{\run\in\N}$ is a solution.
    \item 
    Applying the Robbins\textendash Siegmund theorem to a suitable quasi-descent inequality then gives the almost sure convergence of $\ex_{\run-1}[\sum_{\allplayers}\norm{\vpt[\state]-\vp[\sol]}^2/\vp[\step]]$ to finite value for any $\jsol\in\sols$.
    \item The conditioning on $\last[\filter]$ makes the above quantity not directly amenable to analysis.
    This difficulty is specific to the optimistic algorithms that we consider here as they make use of past feedback in each iteration.
    We overcome this issue by introducing a virtual iterate $\vt[\surr{\jstate}]=(\vpt[\surr{\state}])_{\allplayers}$
    with $\vpt[\surr{\state}]=\vpt[\state]+\vp[\step]\vptpast[\noise]$ that serves as a $\last[\filter]$-measurable surrogate for $\vpt[\state]$.
    We then derive the almost sure convergence of $ \sum_{\allplayers}\norm{\vpt[\surr{\state}]-\vp[\sol]}^2/\vp[\step]$.
    \item To conclude, along with the almost sure convergence of $\norm{\inter[\jstate]-\vt[\surr{\jstate}]}$ and $\norm{\jvecfield(\inter[\jstate])}$ to $0$ we derive the almost sure convergence of $\inter[\jstate]$ to a Nash equilibrium.
    \qedhere
\end{enumerate}
\end{proof}
\vspace{-0.5em}

In case where the players run the adaptive variant of OptDA+, we expect the learning rates to behave as constants asymptotically and thus similar reasoning can still apply.
Formally, we show in the appendix that under multiplicative noise the learning rates of the players converge almost surely to positive constants, and prove the following results concerning the induced trajectory.

\begin{theorem}
\label{thm:cvg-relative-noise-adapt}
Suppose that \cref{asm:noises,asm:lips,asm:VS,asm:boundedness} hold with $\noisedev=0$
and all players run \eqref{OptDA+} with learning rates \eqref{adaptive-lr}.
Then, 
\begin{enumerate*}[\itshape i\upshape)]
    \item $\sum_{\run=1}^{+\infty}\norm{\jvecfield(\inter[\jstate])}^2<+\infty$ with probability $1$, and 
    \item $\inter[\jstate]$ converges almost surely to a Nash equilibrium.
\end{enumerate*}
\end{theorem}

Compared to \cref{thm:cvg-relative-noise}, we can now only bound $\sum_{\run=1}^{\infty}\norm{\jvecfield(\inter[\jstate])}^2$ in an almost sure sense.
This is because in the case of adaptive learning rates, our proof relies on inequality \eqref{eq:adapt-main-ineq},
and deriving a bound on $\sum_{\run=1}^{\infty}\ex[\norm{\jvecfield(\inter[\jstate])}^2]$ from this inequality does not seem possible.
Nonetheless, with the almost sure convergence of the learning rates to positive constants, we still manage to prove almost sure last-iterate convergence of the trajectory of play towards a Nash equilibrium.

Such last-iterate convergence results for adaptive methods are relatively rare in the literature, and most of them assume perfect oracle feedback.
To the best of our knowledge, the closest antecedents to our result are \cite{LZMJ20,AFKL+22}, but both works make the more stringent cocoercive assumptions and consider adaptive learning rate that is the same for all the players. In particular, their learning rates are  computed with global feedback and are thus less suitable for the learning-in-game setup.

\paragraph{Convergence of trajectories under additive noise\afterhead}

To ensure small regret under additive noise, we take vanishing learning rates.
This makes the analysis much more difficult as the term $(1/\vptupdate[\step]-1/\vpt[\step])\norm{\vpt[\uvec]-\vp[\arpoint]}^2$ appearing on the \acl{RHS} of inequality \eqref{eq:quasi-descent-individual-a} is no longer summable.
Nonetheless, it is still possible to provide bound on the sum of the squared operator norms. 

\begin{theorem}
\label{thm:boundv2-abs-noise}
Suppose that \cref{asm:noises,asm:lips,asm:VS} hold and either
\begin{enumerate*}[\itshape i\upshape)]
\item all players run \eqref{OG+} with learning rates described in \cref{thm:OG+-regret}\ref{thm:OG+-regret-add-main} and 
$\vt[\stepalt]=\Omega(1/\run^{\frac{1}{2}-\exponent})$ for some $\exponent\in[0,1/4]$;
\item all players run \eqref{OptDA+} with learning rates described in \cref{thm:OptDA+-regret}\ref{thm:OptDA+-regret-add-main}
and $\vpt[\stepalt][\play] = \Omega(1/\run^{\frac{1}{2}-\exponent})$ for all $\allplayers$ for some $\exponent\in[0,1/4]$; or
\item all players run \eqref{OptDA+} with learning rates \eqref{adaptive-lr} and \cref{,asm:boundedness} holds.
\end{enumerate*}
Then, $\sum_{\run=1}^{\nRuns}\ex[\norm{\jvecfield(\inter[\jstate])}^2]=\tbigoh(\nRuns^{1-\exponent})$.
\end{theorem}

\cref{thm:boundv2-abs-noise} suggests that the convergence speed of $\norm{\jvecfield(\inter[\jstate])}^2$ under additive noise actually depends on $\exponent$.
Therefore, though the entire range of $\exponent\in[0,1/4]$ leads to $\bigoh(\sqrt{\nRuns})$ regret, taking larger $\exponent$ 
may result in a more stabilized trajectory.
This again goes against \cref{prop:OptDA+-regret-adversarial,prop:OptDA+-adapt-regret-adversarial}, which suggests smaller $\exponent$ leads to smaller regret in the face of adversarial opponents.

Finally, we also show last-iterate convergence of the trajectory of OG+ under additive noise.

\begin{theorem}
\label{thm:cvg-OG+}
Suppose that \cref{asm:noises,asm:lips,asm:VS} hold and all players run \eqref{OG+} with non-increasing learning rate sequences $\vt[\stepalt]$ and $\vt[\step]$ satisfying \eqref{OG+-lr} and $\vt[\stepalt] = \Theta(1/(\run^{\frac{1}{2}-\exponent}\sqrt{\log\run}))$,
$\vt[\step] = \Theta(1/(\sqrt{\run}\log\run))$
for some $\exponent\in(0,1/4]$.
%
Then, $\vt[\jstate]$ converges almost surely to a Nash equilibrium.
Moreover,
if $\sup_{\run\in\N}\ex[\norm{\vt[\jnoise]}^4]<+\infty$, then $\inter[\jstate]$ converges almost surely to a Nash equilibrium.
\end{theorem}

\cref{thm:cvg-OG+},
in showing that the sequence $\vt[\jstate]$ generated by \ac{OG+} converges under suitable learning rates,
resolves an open question of \cite{HIMM20}.
By contrast,
the analysis of OG+ is much more involved due to the use of past feedback, as explained in the proof of \cref{thm:cvg-relative-noise}.
Going further,
in the second part of statement, we show that $\inter[\jstate]$
also converges to a Nash equilibrium as long as the $4$-th moment of the noise is bounded.
Compared to OptDA+, it is possible to show last-iterate convergence for OG+
under additive noise because we can use
$\ex_{\run-1}[\norm{\vt[\jstate]-\jsol}^2]$ (with $\jsol\in\sols$) as a Lyapunov function.
The same strategy does not apply to OptDA+ due to summability issues.
This is a common challenge shared by trajectory convergence analysis of the dual averaging template under additive noise.

%% file: sections/conclusion.tex
In this paper, we look into the fundamental problem of no-regret learning in games under uncertainty.
We exhibited algorithms that enjoy constant regret under multiplicative noise.
Building upon this encouraging result, we further studied an adaptive variant and proved trajectory convergence of the considered algorithms.
A central element that is ubiquitous in our work is the trade-off between robustness in the fully adversarial setting and faster convergence in the game-theoretic case, as encoded by the exponent $\exponent$.
Whether this trade-off is inherent to the problem or an artifact of the algorithm design warrants further investigation.

Moving forward, there are many important problems that remain to be addressed.
On the technical side, a first goal would be to deepen our understanding on the convergence behavior of \ac{OptDA+} under additive noise.
Extension of our results to learning in other type of games and/or under different types of uncertainty \textendash\ such as learning in finite games with sampling- or payoff-based feedback \textendash\ would likewise be a valuable contribution.
Going one step further,
analyzing the situation where only a fraction of players deviate is practically relevant (the cases studied in this paper represent the extreme of this spectrum).
Taking into account other type of regret that may be more suitable for game-theoretic settings
is yet another fruitful research direction to pursue.


%% file: sections/acknowledgment.tex
This work received financial support from
MIAI@Grenoble Alpes (ANR-19-P3IA-0003),
the European Research Council (ERC) under the European Union's Horizon 2020 research and innovation program (grant agreement n° 725594 - time-data),
and
the Swiss National Science Foundation (SNSF) under grant number 200021\_205011.
P. Mertikopoulos was also supported by
the grant ALIAS (ANR-19-CE48-0018-01).

%% file: appendices/apx-content.tex

\newpage
\appendix

\noindent\rule{\textwidth}{1pt}
\begin{center}
\vspace{7pt}
{\Large \fontseries{bx}\selectfont Appendix}
\end{center}
\noindent\rule{\textwidth}{1pt}


\renewcommand{\contentsname}{Table of Contents}
\etocdepthtag.toc{mtappendix}
\etocsettagdepth{mtchapter}{none}
\etocsettagdepth{mtappendix}{subsection}
\tableofcontents



\newpage

\section{Prelude}
\label{apx:overview}
\input{appendices/overview}

\vspace{-1em}
\section{Further Related Work}
\label{apx:related}
\input{appendices/related}

\section{Additional Figures}
\label{apx:figures}
\input{appendices/figures}

\section{Technical Details and Notations}
\label{apx:notations}
\input{appendices/notations}

\section{Preliminary Analysis for OG+ and OptDA+}
\label{apx:prelim}
\input{appendices/shared}

\section{Regret Analysis with Predetermined Learning Rates}
\label{apx:regret}
\input{appendices/regret}

\section{Regret Analysis with Adaptive Learning Rates}
\label{apx:regret-adaptive}
\input{appendices/regret-adaptive}

\section{Last-iterate Convergence}
\label{apx:convergence}
\input{appendices/convergence}

%% file: appendices/overview.tex
The appendix is organized as follows.
In \cref{apx:related} we complement our introduction with an overview on other related works.
In \cref{apx:figures} we expand on our plots for better visibility. We also provide some additional figures there.
Subsequently, we build toward the proofs of our main results in \cref{apx:notations,apx:prelim,apx:regret,apx:regret-adaptive,apx:convergence}.
\cref{apx:notations} introduces the notations used in the proofs.
Some technical details concerning the measurability of the noises and learning rates are discussed as well.
\cref{apx:prelim} contains elementary energy inequalities that are repeatedly used through out our analysis.
\cref{apx:regret,apx:regret-adaptive} are dedicated to the regret analysis of the non-adaptive and the adaptive variants.
Bounds on the expectation of the sum of the squared operator norms $\sum_{\run=1}^{\nRuns}\norm{\jvecfield(\inter[\jstate])}^2$ are also established in these two sections, as bounding this quantity often consists in an important step for bounding the regret.
Finally, proofs on the trajectory convergence are presented in \cref{apx:convergence}.

Importantly, in the appendix we present our results in a way that fits better the analysis.
Hence, both the organization and the ordering of the these results differ from those in the main paper.
For the ease of the reader, we summarize below how the results in the appendix correspond to those in the main paper.

\input{tables/theorems-correspondence}

%% file: tables/theorems-correspondence.tex
\begin{table}[htb]
    \renewcommand{\arraystretch}{1.2}
    \setlength\tabcolsep{1em}
    \centering
    \begin{tabularx}{0.85\textwidth}{l|l}
        \toprule
        Results of main paper
        &
        Results of appendix
        \\
        \midrule
        \cref{main:lem:quasi-descent-individual} &
        \cref{lem:OG+-quasi-descent-individual};
        \cref{lem:OptDA+-quasi-descent-individual}
        \\
        \cref{thm:OG+-regret} &
        \cref{thm:OG+-regret-apx};
        \cref{lem:linearized-regret}
        \\
        \cref{thm:OptDA+-regret} &
        \cref{thm:OptDA+-regret-apx};
        \cref{lem:linearized-regret}
        \\
        \cref{thm:OptDA+-adapt-regret} &
        \cref{thm:OptDA+-adapt-absnoise-linregret};
        \cref{thm:OptDA+-adapt-relnoise-linregret};
        \cref{lem:linearized-regret}
        \\
        \cref{thm:cvg-relative-noise} &
        \cref{thm:OptDA+-bound-V2}\,\ref{thm:OptDA+-bound-V2-mul};
        \cref{thm:cvg-OptDA+}
        \\
        \cref{thm:cvg-relative-noise-adapt} &
        \cref{thm:OptDA+-adapt-as-convergence}
        \\
        \cref{thm:boundv2-abs-noise} &
        \cref{thm:OG+-bound-V2}\,\ref{thm:OG+-bound-V2-add};
        \cref{thm:OptDA+-bound-V2}\,\ref{thm:OptDA+-bound-V2-add};
        \cref{thm:OptDA+-adapt-abs-noise-sumV2}
        \\
        \cref{thm:cvg-OG+} &
        \cref{thm:OG+-as-convergence};
        \cref{thm:OG+-as-convergence-inter}
        \\
        \cref{prop:OptDA+-regret-adversarial} &
        \cref{prop:OptDA+-regret-adversarial-apx};
        \cref{lem:linearized-regret}
        \\
        \cref{prop:OptDA+-adapt-regret-adversarial} &
        \cref{prop:OptDA+-adapt-adversarial-apx};
        \cref{lem:linearized-regret}
        \\
        \bottomrule
    \end{tabularx}
    \vspace{0.8em}
    \caption{Correspondence between results presented in the appendix and results presented in the main paper.}
    \label{tab:correspondence}
\end{table}

%% file: appendices/related.tex
On the algorithmic side,
both \ac{OG} and \ac{EG} have been extensively studied over the past decades in the contexts of, among others, variational inequalities \cite{Nes07,Nem04}, online optimization \cite{CYLM+12}, and learning in games \cite{RS13-NIPS,DFG21}.
While the original design of these methods considered the use of the same learning rate for both the extrapolation and the update step, several recent works have shown the benefit of scale separation between the two steps.
Our method is directly inspired by \cite{HIMM19}, which proposed a double step-size variant of \ac{EG} for achieving last-iterate convergence in stochastic variationally stable games.
Among the other uses of learning rate separation of optimistic gradient methods, we should mention here \cite{ZY20,FMPV21} for faster convergence in bilinear games, \cite{LK21,DDJ21,PPFC+22} for performance guarantees under weaker assumptions, and \cite{FOCM+21,HIMM22} for robustness against delays. 


Concerning the last-iterate convergence of no-regret learning dynamics in games with noisy feedback, most existing results rely on the use of vanishing learning rates and are established under more restrictive assumptions such as strong monotonicity \cite{KS19,HIMM19,AIMM21} or strict variational stability \cite{MZ19,MLZF+19}.
Our work, in contrast, studies learning with potentially non-vanishing learning rates in variationally stable games.
This is made possible thanks to a clear distinction between additive and multiplicative noise; the latter has only been formerly explored in the game-theoretic context by \cite{LZMJ20,ABM21} for the class of cocoercive games.\footnote{In the said works they use the term absolute random noise and relative random noise for additive noise and multiplicative noise.}
Relaxing the cocoercivity assumption is a nontrivial challenge, as testified by the few number of works that establish last-iterate convergence results of stochastic algorithms for monotone games.
Except for \cite{HIMM20} mentioned above, this was achieved either through mini-batching \cite{IJOT17,BMSV21}, Tikhonov regularization / Halpen iteration \cite{KNS12}, or both \cite{CSGD22}.

%% file: appendices/figures.tex
In this section we provide the complete version of \cref{fig:illustration}.
In additional to the algorithms already considered in the said figure, we also present results for the case where the two players follow the vanilla gradient descent methods, which we mark as \acdefs{GDA}.

To begin, we complement the leftmost plot of \cref{fig:illustration} by \cref{fig:trajectories}, where we present individual plots of the trajectories induced by different algorithms for better visibility.
For optimistic algorithm, we present the trajectory both of the sequence of play $\vt[\jaction]=\inter[\jstate]$ and of the auxiliary iterate $\vt[\jstate]$.
The two algorithms \ac{GDA} and \ac{OG} have their iterates spiral out, indicating a divergence behavior, conformed to our previous discussions.
For \ac{OG+} run with constant learning rate and adaptive \ac{OptDA+}, we observe that the trajectory of $\vt[\jstate]$ is much ``smoother'' than that of $\vt[\jaction]=\inter[\jstate]$.
This is because the extrapolation step is taken with a larger learning rate.
Finally, adaptive \ac{OptDA+} has its iterates go far away from the equilibrium in the first few iterations due to the initialization with large learning rates, but eventually finds the right learning rates itself and ends up with a convergence speed and regret that is competitive with carefully tuned OG+.


Next, In \cref{fig:performances}, we expand on the right two plots of \cref{fig:illustration} with additional curves for \ac{GDA}.
\ac{GDA} and \ac{OG} run with the same decreasing learning rate sequences $\vt[\step]=0.1/\sqrt{\run+1}$ turn out to have similar performance.
This suggests that without learning rate separation, the benefit of the extrapolation step may be completely lost in the presence of noise.
\vspace{-0.5em}

\input{figures/trajectories}
\vspace{-0.5em}

\input{figures/performances}

%% file: figures/trajectories.tex
\begin{figure}[!h]
    \centering
    \begin{subfigure}{\linewidth}
        \centering
        \includegraphics[width=0.235\linewidth]{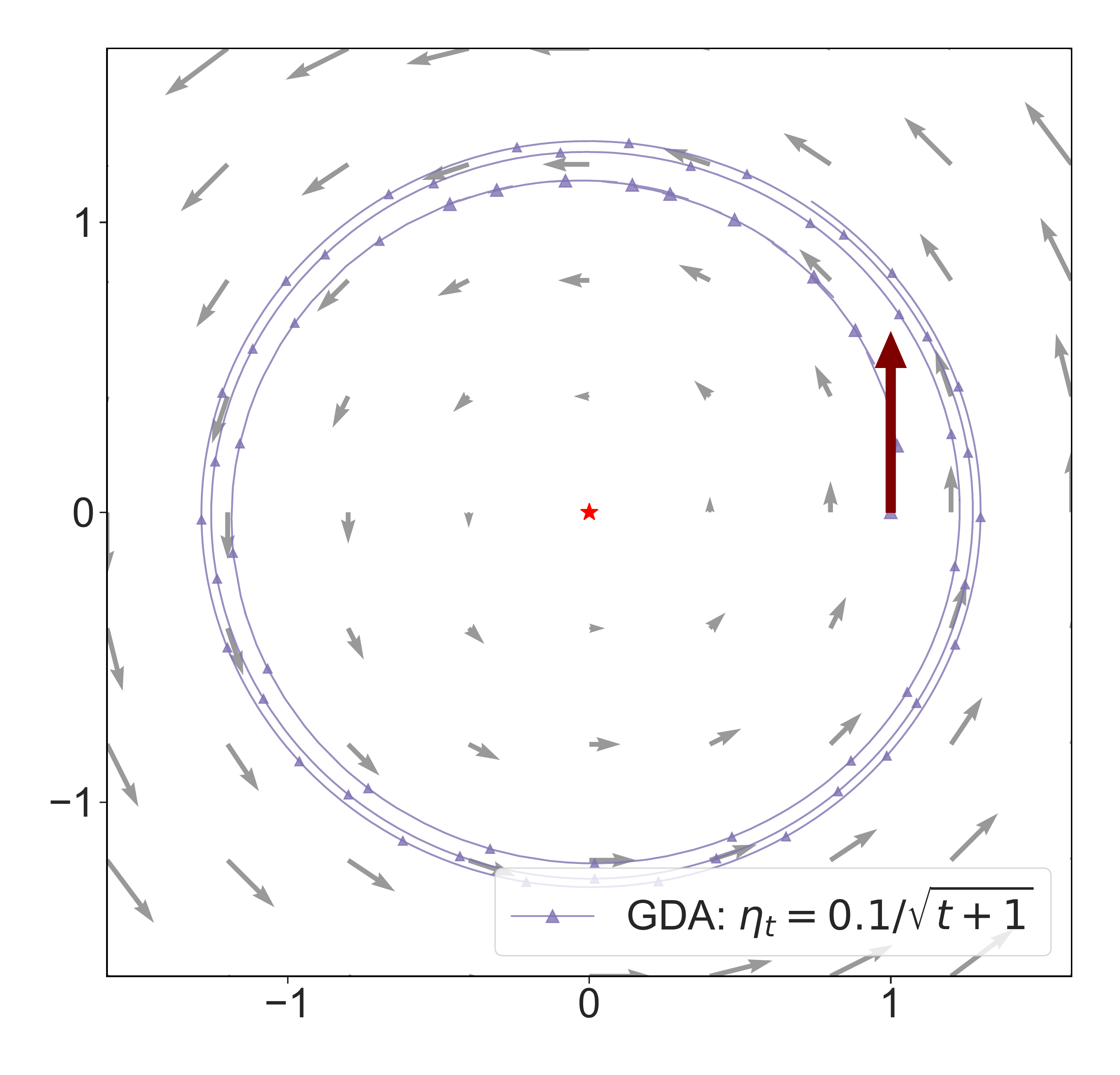}
        \hspace{0.02em}
        \includegraphics[width=0.235\linewidth]{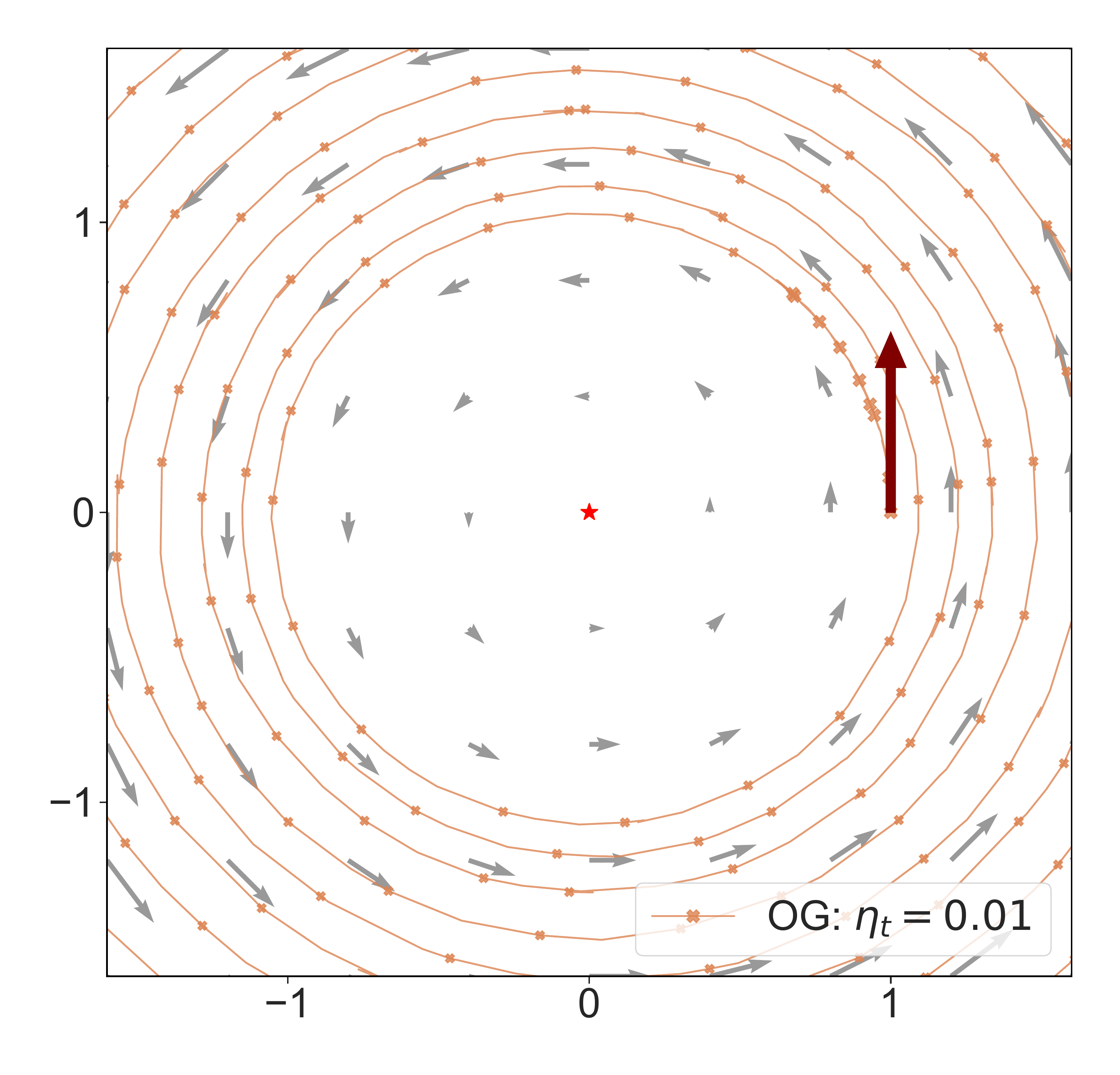}
        \hspace{0.02em}
        \includegraphics[width=0.235\linewidth]{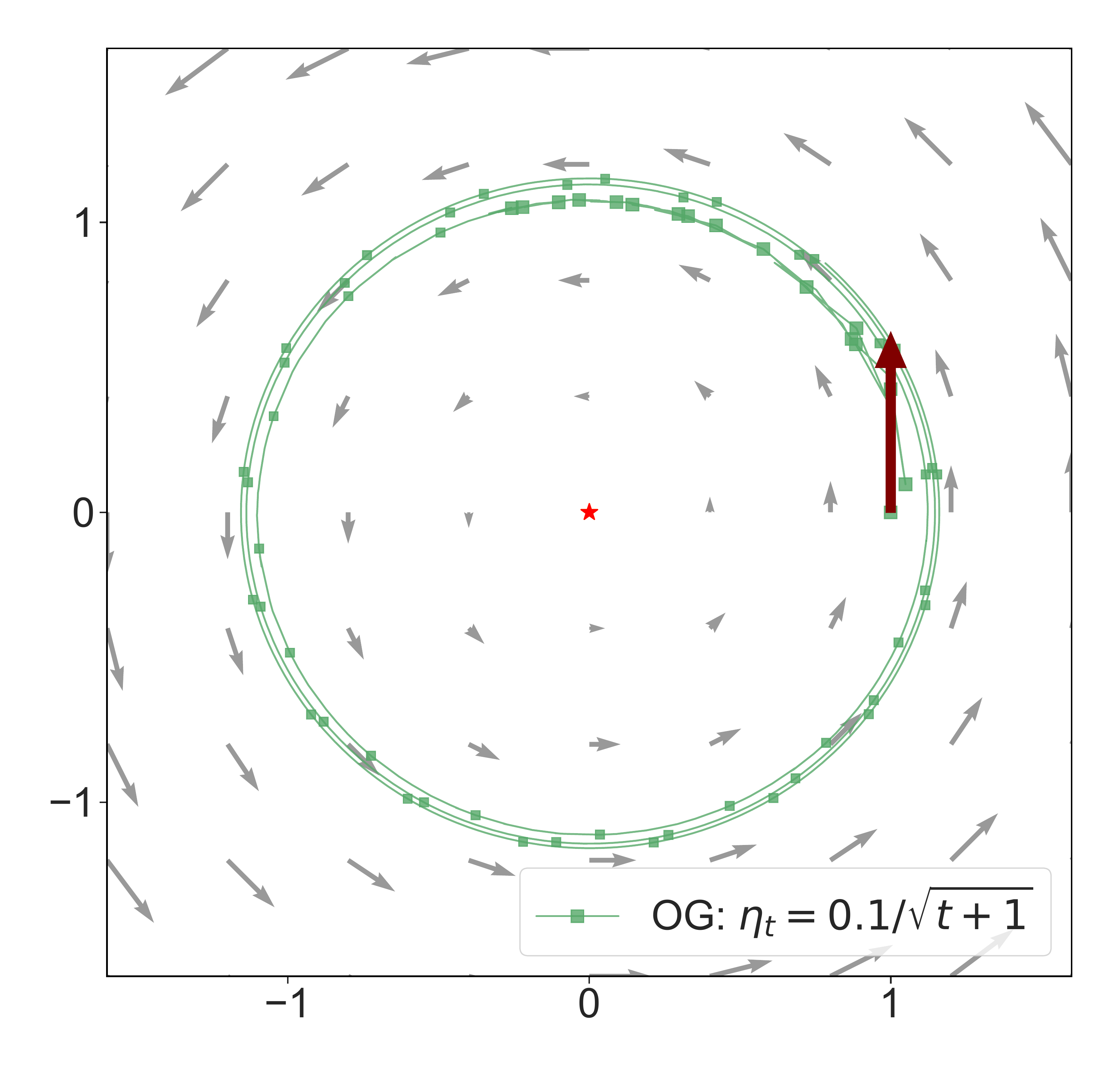}
        \hspace{0.02em}
        \includegraphics[width=0.235\linewidth]{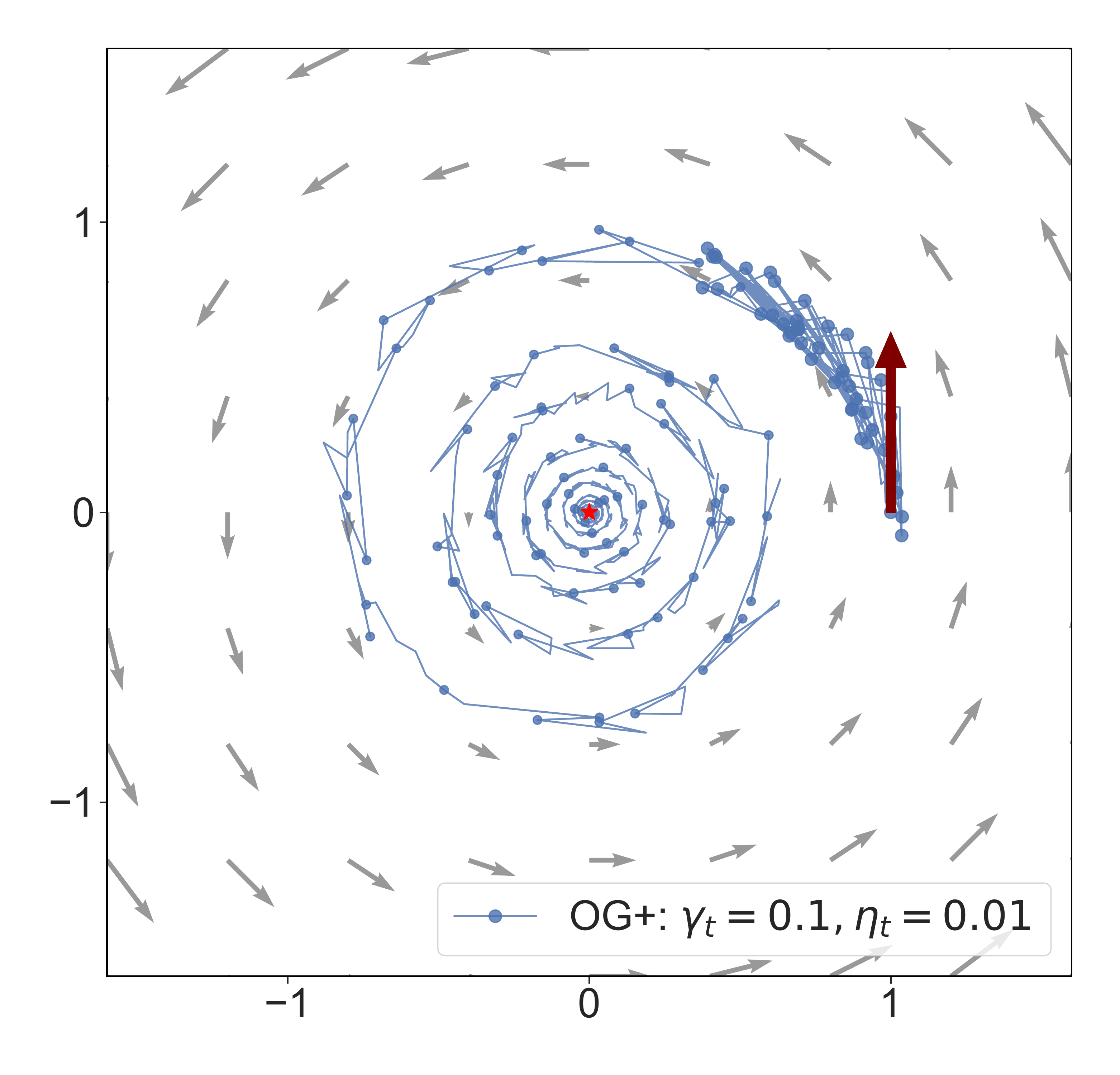}
        \caption{Trajectory of $\vt[\jaction]$ of different learning algorithms.}
    \end{subfigure}
    \\
    \begin{subfigure}{\linewidth}
        \centering
        \includegraphics[width=0.235\linewidth]{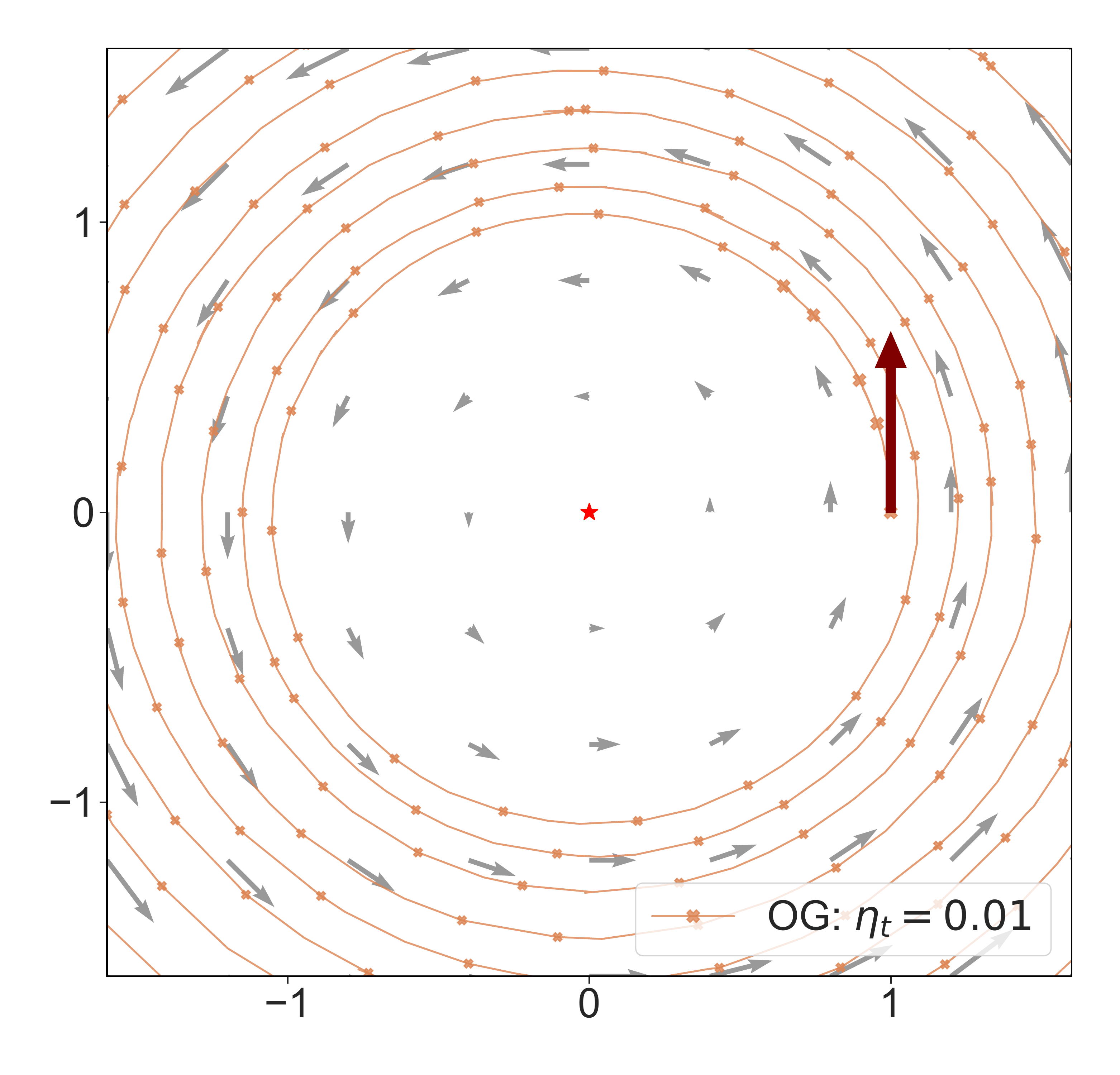}
        \hspace{0.02em}
        \includegraphics[width=0.235\linewidth]{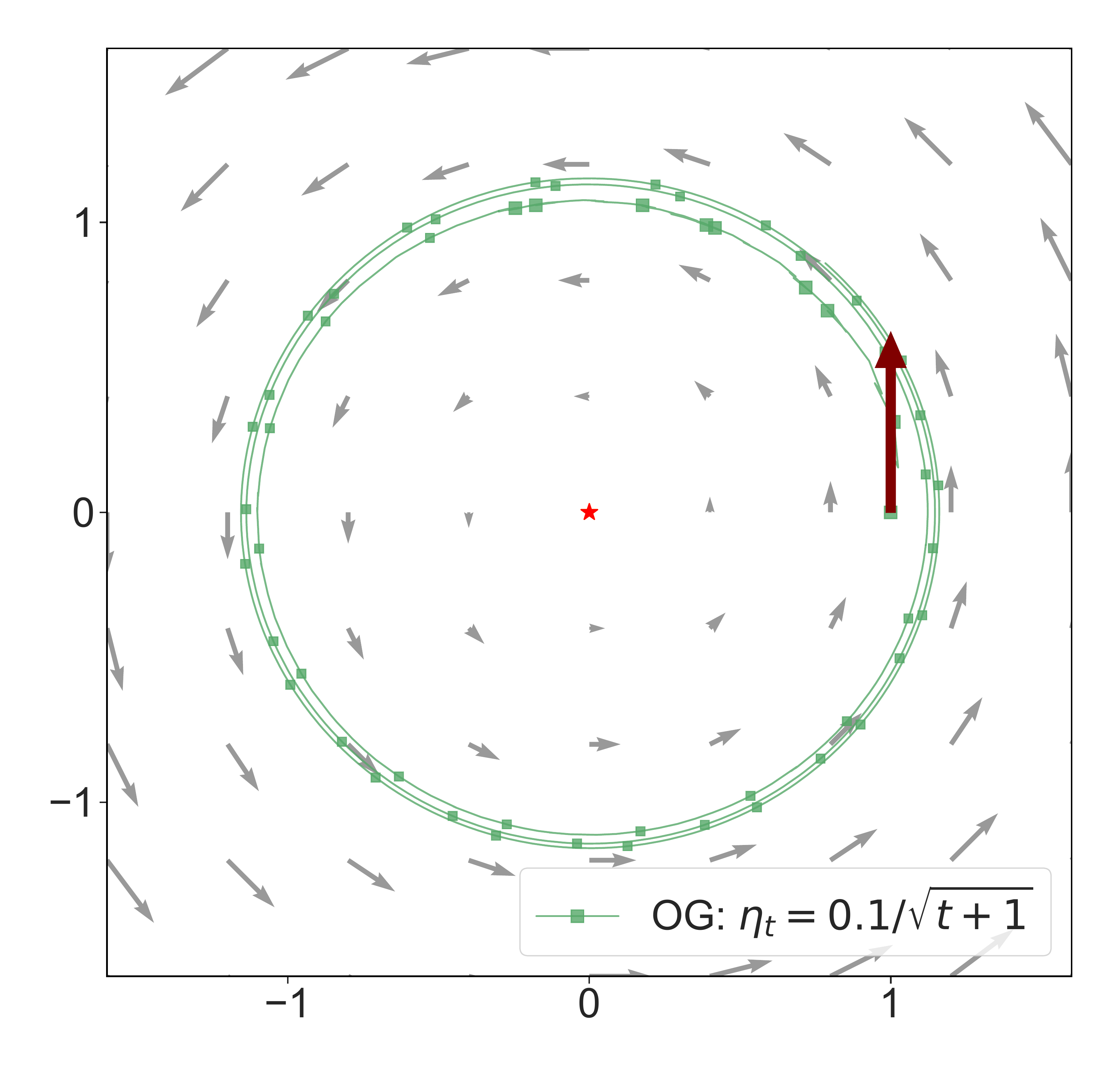}
        \hspace{0.02em}
        \includegraphics[width=0.235\linewidth]{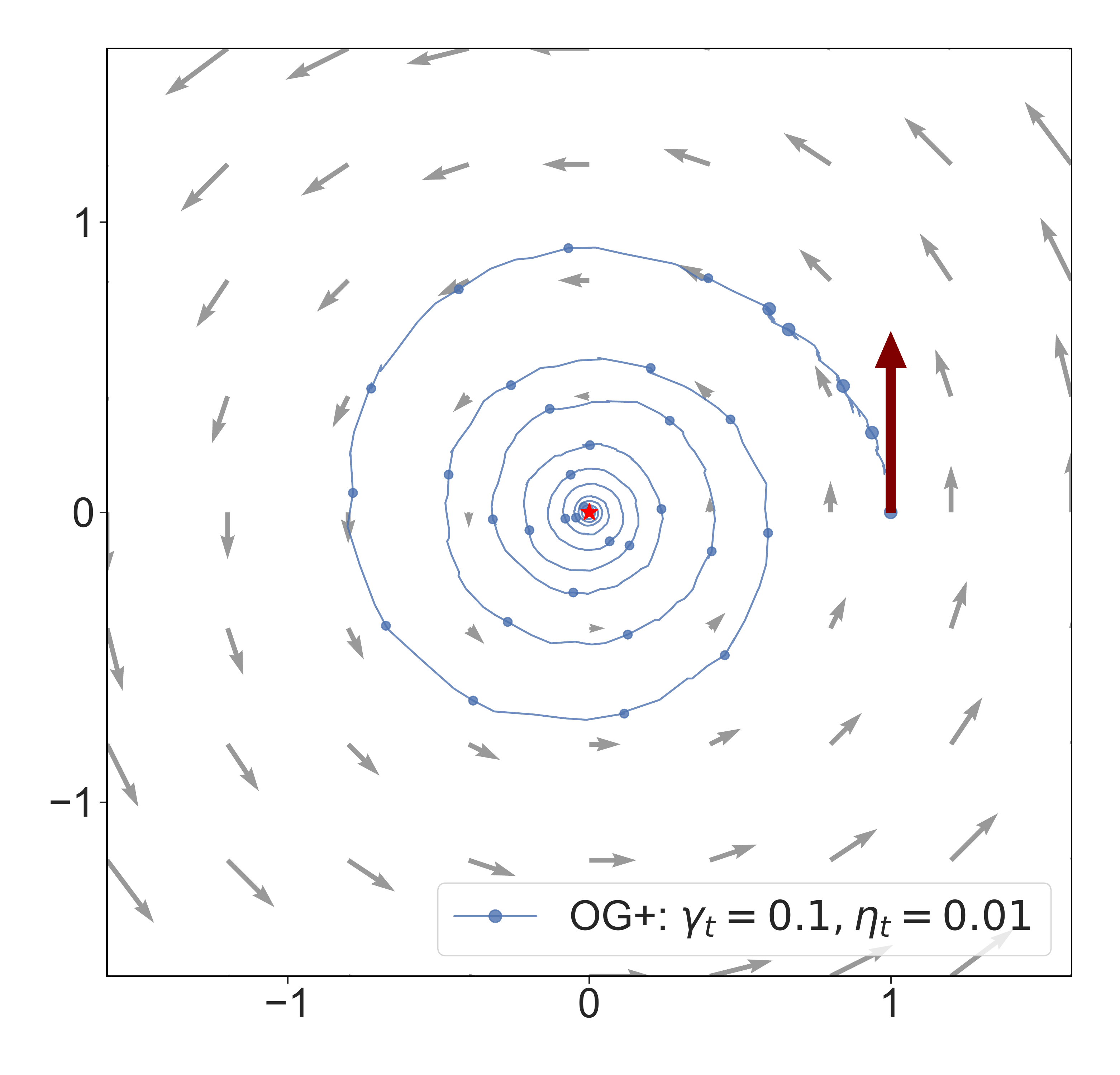}
        \caption{Trajectory of $\vt[\jstate]$ of different optimistic learning algorithms.}
    \end{subfigure}
    \\
    \begin{subfigure}{0.42\linewidth}
        \centering
        \includegraphics[width=\linewidth]{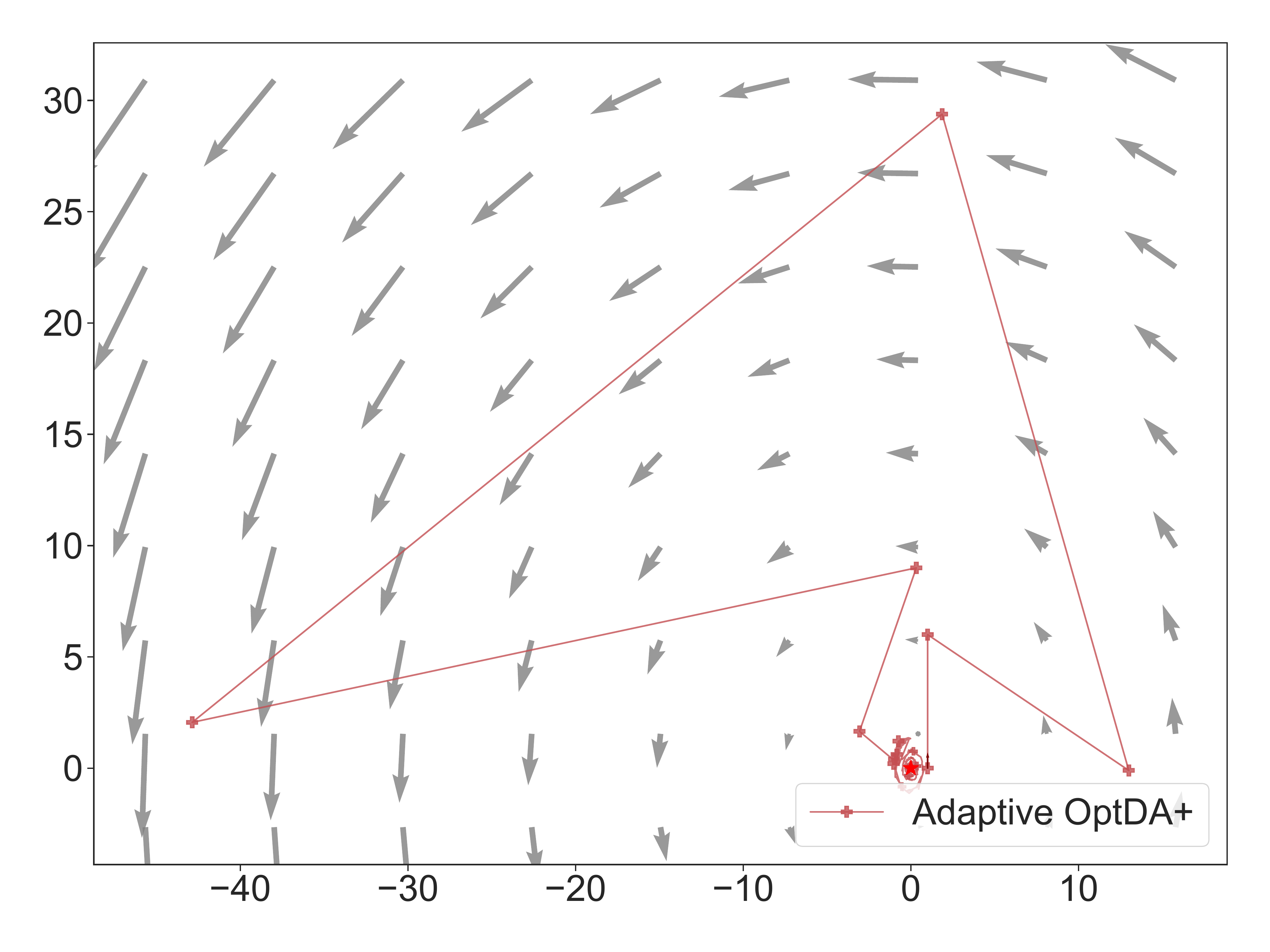}
        \caption{Trajectory of $\vt[\jaction]$ of adaptive OptDA+.}
    \end{subfigure}
    \hspace{1em}
    \begin{subfigure}{0.42\linewidth}
        \centering
        \includegraphics[width=\linewidth]{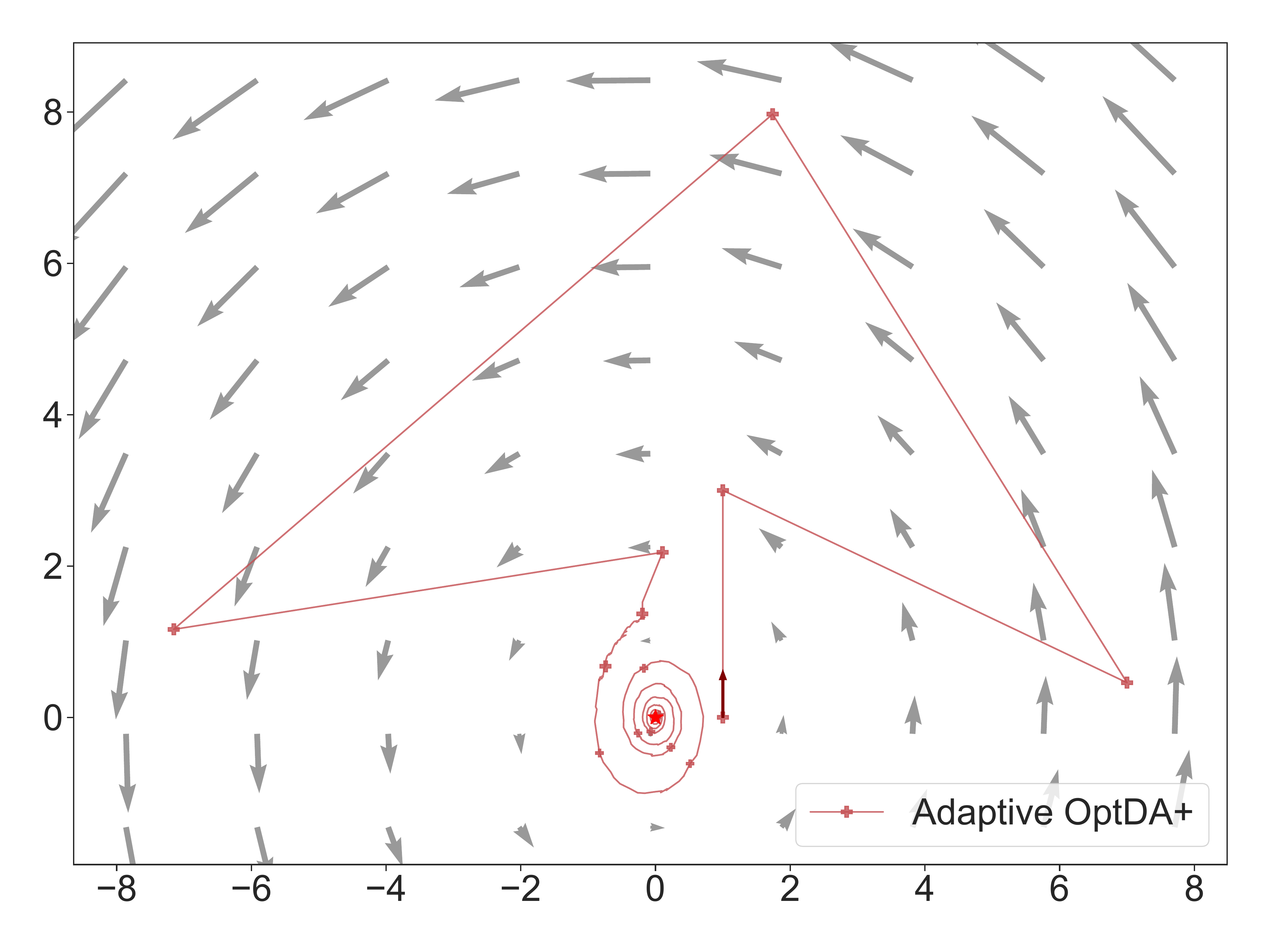}
        \caption{Trajectory of $\vt[\jstate]$ of adaptive OptDA+.}
    \end{subfigure}
    \caption{Trajectories induced by different learning algorithms on the model described in \cref{ex:bilinear}.
    We recall that for optimistic learning algorithms, the played point is $\vt[\jaction]=\inter[\jstate]$.
    We take $\exponent=1/4$ for adaptive OptDA+.
    }
    \label{fig:trajectories}
    \vspace{-4em}
\end{figure}

%% file: figures/performances.tex
\begin{figure}[ht]
    \centering
    \includegraphics[width=\textwidth]{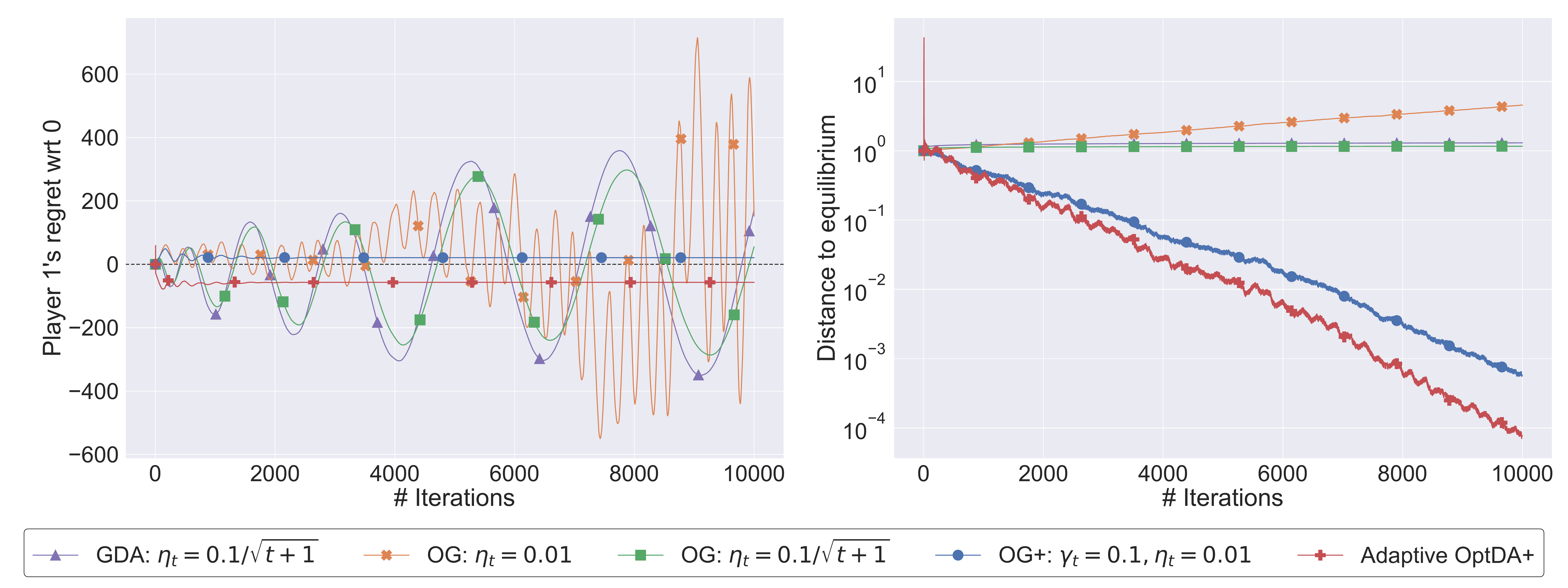}
    \caption{Player $1$'s regret and distance to equilibrium when both players follow a certain learning strategy in the model described in \cref{ex:bilinear}.
    We take $\exponent=1/4$ for adaptive OptDA+.}
    \label{fig:performances}
\end{figure}

%% file: appendices/notations.tex
In this section we introduce the necessary notations for our analysis and discuss some technical details omitted in the main text.

\paragraph{Noise, initialization, and measurability\afterhead}

Throughout our proof, to emphasize that $\vpt[\gvec]=\vp[\vecfield](\vpt[\action])+\vpt[\noise]$ is a stochastic estimate of $\vp[\vecfield](\vptinter[\state])$ in our algorithms, we use the notations $\vptinter[\svecfield]=\vpt[\gvec]$ and $\vptinter[\snoise]=\vpt[\noise]$.
For the update of $\vpt[\state][\play][3/2]$, we systematically take $\vpt[\gvec][\play][0]=\vpt[\svecfield][\play][\paststart]=0$.
We also write $\vpt[\noise][\play][1/2]=0$.

A part of our analysis will be built on the fact that $\vptpast[\noise]$ is $\vt[\filter]$-measurable.
There is however no a priori reason for this to be true \textendash\ as $\seqinf[\filter]$ is the natural filtration associated to $\seqinf[\jaction]$, a sequence that can for example be taken constant independent of the feedback.
To address this, we establish here that $\vptpast[\noise]$ is indeed $\vt[\filter]$-measurable when player $\play$ uses OG+ or OptDA+ with learning rates satisfying a certain measurability assumption.
To state it, we define $\vpt[\filter]$ as the $\sigma$-algebra generated by $\{(\vt[\jaction][\runalt])_{\runalt=1}^{\run},(\vpt[\noise][\play][\runalt])_{\runalt=1}^{\run-1}\}$.

\begin{assumption}
\label{asm:stepsize-measurable}
For all $\run\in\N$, the learning rates $\vptupdate[\stepalt]$ and $\vptupdate[\step]$ are $\vpt[\filter]$-measurable.
\end{assumption}

The following lemma shows that whenever \cref{asm:stepsize-measurable} holds, one can directly work with $\seqinf[\filter]$.

\begin{lemma}
\label{lem:measurable}
Let player $\play$ run \eqref{OG+} or \eqref{OptDA+} with learning rates satisfying \cref{asm:stepsize-measurable}. Then, for every $\run\in\N$, it holds
$\vpt[\filter]=\vt[\filter]$.
In other words, $\vptpast[\noise]$ is  $\vt[\filter]$-measurable.
\end{lemma}
\vspace*{-1em}
\begin{proof}
We prove the lemma by induction.
For $\run=1$, this is true by definition.
Now, fix $\run\ge2$ and assume that we have proven the statements for all $\runalt\le\run-1$.
To show that the statement is also true for $\run$, we note that for both OG+ and OptDA+, $\vpt[\action]=\vptinter[\state]$ is a linear combination of the vectors in $\{\vp[\vecfield](\vt[\jaction][\runalt])\}_{\runalt=1}^{\run-1}\union
\{\vptinter[\noise][\play][\runalt]\}_{\runalt=1}^{\run-1}
$
with coefficients in $\{\vpt[\step][\play][\runalt]\}_{\runalt=1}^{\run}\union\{\vpt[\stepalt]\}$.
All the involved quantities except for $\vptpast[\noise]$ is $\last[\filter]$-measurable by the induction hypothesis.
They are thus $\vt[\filter]$-measurable, and as $\vpt[\state]$ is  $\vt[\filter]$-measurable by the definition of $\vt[\filter]$ we concludes that $\vptpast[\noise]$ is also $\vt[\filter]$-measurable, which along with the induction hypothesis implies immediately $\vpt[\filter]=\vt[\filter]$.
\end{proof}

An immediate consequence of \cref{lem:measurable} is the following.

\begin{corollary}
\label{cor:measurable}
Let player $\play$ run \eqref{OG+} or \eqref{OptDA+} with learning rates satisfying \cref{asm:stepsize-measurable}. Then for every $\run\in\N$, $\vptupdate[\stepalt]$ and $\vptupdate[\step]$ are $\vt[\filter]$-measurable.
\end{corollary}

Throughout the sequel, both \cref{lem:measurable} and \cref{cor:measurable} will be used implicitly.
Our adaptive learning rates \eqref{adaptive-lr} apparently satisfy \cref{asm:stepsize-measurable}.
As for the non-adaptive case, for simplicity, we assume all their learning rates are predetermined, that is, they are $\vt[\filter][1]$-measurable;
for more details on this point see \cref{rem:lr-measurability}.
$\vt[\filter][0]$ denotes the trivial $\sigma$-algebra.

As another technical detail, in our proofs we assume deterministic $\vt[\jstate][1]$, but the entire analysis still goes through for random $\vt[\jstate][1]$ under the following conditions
\begin{enumerate}[leftmargin=*]
    \item For non-adaptive algorithms, we
    require $\ex[\norm{\vt[\jstate][1]}^2]<+\infty$.
    \item For adaptive OptDA+, we require existence of $\radius\in\R_+$ such that $\norm{\vt[\jstate][1]}\le\radius$ holds almost surely.
\end{enumerate}

\paragraph{Notations related to the learning rates.}
For any $\jaction=(\vp[\action])_{\allplayers}\in\points=\vecspace$
and $\weights=(\vp[\weights])_{\allplayers}\in\R_+^{\nPlayers}$, we write the weighted norm as $\norm{\jaction}_{\weights}
=\sqrt{
\sumplayers
\vp[\weight]\norm{\vp[\action]}^2}$.
%
%
The weights $\weights$ will be taken as a function of the learning rates.
It is thus convenient to write $\vt[\jstep]=(\vpt[\step])_{\allplayers}$ and
$\vt[\jstepalt]=(\vpt[\stepalt])_{\allplayers}$ for the joint learning rates.
The arithmetic manipulation and the comparisons of these vectors should be taken elementwisely.
For example, the element-wise division is $1/\vt[\jstep]=(1/\vpt[\step])_{\allplayers}$.
For ease of notation, we also write $\norm{\weights}_1=\sumplayers\vp[\weight]$
and $\infnorm{\weights}=\max_{\allplayers}\vp[\weight]$ respectively for the L1 norm and the L-infinity norm of an $\nPlayers$-dimensional vector $\weights$.

%% file: appendices/shared.tex
In this section, we lay out the basis for the analysis of OG+ and OptDA+.

\subsection{Generalized Schemes with Arbitrary Input Sequences}

As a starting point, we derive elementary energy inequalities for the following two generalized schemes run with arbitrary vector sequences $\seqinf[\gvec]$ and $(\inter[\gvec])_{\run\in\N}$.
\begin{equation*}
    \arraycolsep=4pt
    \begin{array}{llc}
    \text{\labelitemi}
    &
    \text{Generalized OG+}
    &
    \inter = \current - \vt[\stepalt]\vt[\gvec],~~
    \update = \current - \update[\step]\inter[\gvec]\\
    \text{\labelitemi}
    &
    \text{Generalized OptDA+}
    \hspace*{2em}
    &
    \inter = \current - \vt[\stepalt]\vt[\gvec],~~
    \update = \vt[\state][1]-\update[\step]\sum_{\runalt=1}^{\run}\inter[\gvec][\runalt]
    \end{array}
\end{equation*}
In fact, Generalized OG+ with $\inter[\gvec]=\grad\vt[\obj](\inter)$ is nothing but the unconstrained, double step-size variant of the \acl{OptMD} method proposed in \cite{RS13-NIPS}.
On the other hand, Generalized OptDA+ with single learning rate was introduced in \cite{APKMC21} under the name of \acl{GEG}.
These two methods coincide when the learning rates are taken constant.
In practice, $\inter[\gvec]$ is almost always an estimate of $\grad\vt[\obj](\inter)$ while $\vt[\gvec]$ is an approximation of $\inter[\gvec]$.
As a matter of fact, as we show in the following propositions, the dot product $\product{\inter[\gvec]}{\vt[\gvec]}$ appears with a negative sign in the energy inequalities, which results in a negative contribution when the two vectors are close.

We start with the energy inequality for Generalized OG+.

\begin{proposition}[Energy inequality for Generalized OG+]
\label{prop:OG+-develop}
Let $\seqinf[\state]$ and $\seqinfinter[\state]$ be generated by Generalized OG+. It holds for any $\arpoint\in\points$ and $\run\in\N$ that
\begin{equation}
    \notag
    \norm{\update-\arpoint}^2
    = \norm{\current-\arpoint}^2
    - 2 \update[\step]\product{\inter[\gvec]}{\inter-\arpoint}
    - 2 \vt[\stepalt]\update[\step]\product{\inter[\gvec]}{\vt[\gvec]}
    + (\update[\step])^2\norm{\inter[\gvec]}^2.
\end{equation}
\end{proposition}
\begin{proof}
We develop directly
\begin{align*}
    \norm{\update-\arpoint}^2
    &= \norm{\current-\update[\step]\inter[\gvec]-\arpoint}^2\\
    &= \norm{\current-\arpoint}^2
    - 2\product{\inter[\gvec]}{\current-\arpoint}
    + (\update[\step])^2\norm{\inter[\gvec]}^2\\
    &= \norm{\current-\arpoint}^2
    - 2 \update[\step]\product{\inter[\gvec]}{\inter-\arpoint}
    - 2 \vt[\stepalt]\update[\step]\product{\inter[\gvec]}{\vt[\gvec]}
    + (\update[\step])^2\norm{\inter[\gvec]}^2,
\end{align*}
where in the last equality we use the fact that $\current = \inter + \vt[\stepalt]\vt[\gvec]$.
\end{proof}

For Generalized OptDA+ we have almost the same inequality but for squared distance weighted by $1/\vt[\step]$, with the notation $\vt[\step][1]=\vt[\step][2]$.

\begin{proposition}[Energy inequality for Generalized OptDA+]
\label{prop:OptDA+-descent}
Let $\seqinf[\state]$ and $\seqinfinter[\state]$ be generated by Generalized OptDA+. It holds for any $\arpoint\in\points$ and $\run\in\N$ that
\begin{align*}
    \frac{\norm{\update-\arpoint}^2}{\update[\step]}
    &= \frac{\norm{\current-\arpoint}^2}{\vt[\step]}
    - \frac{\norm{\current-\update}^2}{\vt[\step]}
    \\
    &~~~
    + \left(\frac{1}{\update[\step]}-\frac{1}{\current[\step]}\right)\norm{\vt[\state][1]-\arpoint}^2
    - \left(\frac{1}{\update[\step]}-\frac{1}{\current[\step]}\right)\norm{\vt[\state][1]-\update[\state]}^2
    \\
    &~~~
    - 2 \product{\inter[\gvec]}{\inter-\arpoint}
    - 2 \vt[\stepalt]\product{\inter[\gvec]}{\vt[\gvec]}
    + \product{\inter[\gvec]}{\current-\update}.
\end{align*}
\end{proposition}
\begin{proof}
Using $\inter[\gvec]=(\current-\vt[\state][1])/\vt[\step]-(\update-\vt[\state][1])/\update[\step]$, we can write
\begin{align*}
    \product{\inter[\gvec]}{\update-\arpoint}
    &= \left\langle
    \frac{\current-\vt[\state][1]}{\vt[\step]}
    -\frac{\update-\vt[\state][1]}{\update[\step]},
    \update-\arpoint\right\rangle\\
    &= \frac{1}{\vt[\step]}\product{\current-\update}{\update-\arpoint}
    + \left(\frac{1}{\update[\step]}-\frac{1}{\current[\step]}\right)
    \product{\vt[\state][1]-\update}{\update-\arpoint}\\
    &= \frac{1}{2\vt[\step]}
    (\norm{\current-\arpoint}^2-\norm{\update-\arpoint}^2-\norm{\current-\update}^2)\\
    &~+ \left(\frac{1}{2\update[\step]}-\frac{1}{2\current[\step]}\right)
    (\norm{\vt[\state][1]-\arpoint}^2-\norm{\update-\arpoint}^2-\norm{\vt[\state][1]-\update}^2).
\end{align*}
Multiplying the equality by $2$ and rearranging, we get
\begin{align*}
    \frac{\norm{\update-\arpoint}^2}{\update[\step]}
    &= \frac{\norm{\current-\arpoint}^2}{\vt[\step]}
    - \frac{\norm{\current-\update}^2}{\vt[\step]}
    + \left(\frac{1}{\update[\step]}-\frac{1}{\current[\step]}\right)\norm{\vt[\state][1]-\arpoint}^2\\
    &~- \left(\frac{1}{\update[\step]}-\frac{1}{\current[\step]}\right)\norm{\vt[\state][1]-\update[\state]}^2
    - 2 \product{\inter[\gvec]}{\update-\arpoint}.
\end{align*}
We conclude with the equality
\begin{align*}
    \product{\inter[\gvec]}{\update-\arpoint}
    &= \product{\inter[\gvec]}{\update-\current}
    + \product{\inter[\gvec]}{\current-\inter}
    + \product{\inter[\gvec]}{\inter-\arpoint}\\
    &= \product{\inter[\gvec]}{\update-\current}
    + \vt[\stepalt]\product{\inter[\gvec]}{\vt[\gvec]}
    + \product{\inter[\gvec]}{\inter-\arpoint},
\end{align*}
where we have used $\current = \inter + \vt[\stepalt]\vt[\gvec]$.
\end{proof}

Throughout our work, we assume the learning rate sequences to be non-increasing.
This is essential for OptDA+, as it guarantees the following corollary.

\begin{corollary}
\label{cor:OptDA+-descent}
Let $\seqinf[\state]$ and $\seqinfinter[\state]$ be generated by Generalized OptDA+. For any $\arpoint\in\points$ and $\run\in\N$, if $\update[\step]\le\current[\step]$, it holds that
\begin{align*}
\frac{\norm{\update-\arpoint}^2}{\update[\step]}
    &\le \frac{\norm{\current-\arpoint}^2}{\vt[\step]}
    + \left(\frac{1}{\update[\step]}-\frac{1}{\current[\step]}\right)\norm{\vt[\state][1]-\arpoint}^2
    - 2 \product{\inter[\gvec]}{\inter-\arpoint}\\
    &~~- 2 \vt[\stepalt]\product{\inter[\gvec]}{\vt[\gvec]}
    + \vt[\step]^2\norm{\inter[\gvec]}^2
    +
    \min\left(\vt[\step]^2\norm{\inter[\gvec]}^2- \frac{\norm{\current-\update}^2}{2\vt[\step]},\, 0
    \right).
\end{align*}
\end{corollary}
\begin{proof}
This is immediate from \cref{prop:OptDA+-descent} by applying Young's inequality. More precisely, we use
$(1/\update[\step]-1/\current[\step])\norm{\vt[\state][1]-\update[\state]}^2 \ge 0$ and
\begin{gather*}
    2\product{\inter[\gvec]}{\inter-\arpoint}
    \le 
    \min\left(
    \vt[\step]^2\norm{\inter[\gvec]}^2 + \frac{\norm{\current-\update}^2}{\vt[\step]},\,
    2\vt[\step]^2\norm{\inter[\gvec]}^2 + \frac{\norm{\current-\update}^2}{2\vt[\step]}
    \right).
    \qedhere
\end{gather*}
\end{proof}


\subsection{Quasi-Descent Inequalities for OG+ and OptDA+}

We now turn back to \eqref{OG+} and \eqref{OptDA+} introduced in \cref{sec:regret}.
These are special cases of Generalized OG+ and Generalized OptDA+ with $\vt[\gvec]=\past[\gvec]=\vptpast[\svecfield]$.
The following lemma provides an upper bound on the conditional expectation of $\product{\vptinter[\svecfield]}{\vptpast[\svecfield]}$ 
when all the players follow one of the two strategies,
and is essential for establishing our quasi-descent inequities.

\begin{lemma}
\label{lem:feedback-product}
Let \cref{asm:noises,asm:lips} hold and all players run either \eqref{OG+} or \eqref{OptDA+} with learning rates satisfying \cref{asm:stepsize-measurable}.
Then, for all $\allplayers$ and $\run\ge2$, it holds
\begin{align*}
    -2\ex_{\run-1}[\product{\vptinter[\svecfield]}{\vptpast[\svecfield]}]
    \le\ex_{\run-1}\Bigg[&
    -\norm{\vp[\vecfield](\inter[\jstate])}^2
    - \norm{\vp[\vecfield](\past[\jstate])}^2
    \\
    &+ \norm{\vp[\vecfield](\inter[\jstate])-\vp[\vecfield](\past[\jstate])}^2
    \\
    &+ \lips\left(
    \vpt[\stepalt]\norm{\vptpast[\noise]}^2
    +\sumplayers[\playalt]\frac{(\vpt[\step][\playalt]+\vpt[\stepalt][\playalt])^2\norm{\vptpast[\noise][\playalt]}^2}{\vpt[\stepalt]}
    \right)\Bigg]
\end{align*}
\end{lemma}
\begin{proof}
Thanks to \cref{lem:measurable}, we can apply the law of total expectation  of the expectation to get
\begin{equation}
\label{eq:feedback-product-ex}
\begin{aligned}[b]
    \ex_{\run-1}[\product{\vptinter[\svecfield]}{\vptpast[\svecfield]}]
    &=\ex_{\run-1}[\product{\ex_{\run}[\vptinter[\svecfield]]}{\vptpast[\svecfield]}]\\
    &=\ex_{\run-1}[\product{\vp[\vecfield](\inter[\jstate])}{\vptpast[\svecfield]}]\\
    &=\ex_{\run-1}[\product{\vp[\vecfield](\inter[\jstate])}{\vp[\vecfield](\past[\jstate])}
    +\product{\vp[\vecfield](\inter[\jstate])}{\vptpast[\noise]}].
\end{aligned}
\end{equation}
We rewrite the first term as
\begin{equation}
    \label{eq:feedback-product-develop}
    2\product{\vp[\vecfield](\inter[\jstate])}{\vp[\vecfield](\past[\jstate])}
    = \norm{\vp[\vecfield](\inter[\jstate])}^2
    + \norm{\vp[\vecfield](\past[\jstate])}^2
    - \norm{\vp[\vecfield](\inter[\jstate])-\vp[\vecfield](\past[\jstate])}^2.
\end{equation}

As for the second term, for all $\allplayers[\playalt]$, we define
$\vptinter[\surr{\state}][\playalt]=\vptinter[\state][\playalt]+(\vpt[\step][\playalt]+\vpt[\stepalt][\playalt])\vptpast[\noise][\playalt]$ and
as a surrogate for $\vptinter[\state][\playalt]$ obtained by removing the noise of round $\run-1$.
For OG+ and OptDA+ we have respectively
\begin{align*}
    \vptinter[\surr{\state}][\playalt]
    &
    =\vptlast[\state][\playalt]-(\vpt[\step][\playalt]+\vpt[\stepalt][\playalt])\vp[\vecfield][\playalt](\past[\jstate])
    \\
    \vptinter[\surr{\state}][\playalt]
    &
    =
    \vpt[\state][\play][1]
    -\vpt[\step][\playalt]\sum_{\runalt=1}^{\run-2}
    \vptinter[\svecfield][\playalt][\runalt]
    -(\vpt[\step][\playalt]+\vpt[\stepalt][\playalt])\vp[\vecfield][\playalt](\past[\jstate]).
\end{align*}
With \cref{asm:stepsize-measurable} we then deduce that $\inter[\surr{\jstate}]$ is $\last[\filter]$-measurable and hence
\begin{equation*}
    \ex_{\run-1}[\product{\vp[\vecfield](\inter[\surr{\jstate}])}{\vptpast[\noise]}]
    = \product{\vp[\vecfield](\inter[\surr{\jstate}])}
    {\ex_{\run-1}[\vptpast[\noise]]} = 0.
\end{equation*}
Moreover, by definition of $\inter[\surr{\jstate}]$ we have
\begin{equation*}
    \norm{\inter[\jstate]-\inter[\surr{\jstate}]}^2
    =\sumplayers[\playalt]\norm{\vptinter[\state][\playalt]-\vptinter[\surr{\state}][\playalt]}^2
    =\sumplayers[\playalt](\vpt[\step][\playalt]+\vpt[\stepalt][\playalt])^2\norm{\vptpast[\noise][\playalt]}^2
\end{equation*}
It then follows from the Lipschitz continuity of $\vp[\vecfield]$ that
\begin{equation}
\label{eq:feedback-product-noisebound}
\begin{aligned}[b]
    \ex_{\run-1}[-\product{\vp[\vecfield](\inter[\jstate])}{\vptpast[\noise]}]
    &= \ex_{\run-1}[-\product{\vp[\vecfield](\inter[\jstate])
    -\vp[\vecfield](\inter[\surr{\jstate}])
    }{\vptpast[\noise]}]
    \\
    &~~~
    - \ex_{\run-1}[\product{\vp[\vecfield](\inter[\surr{\jstate}])}{\vptpast[\noise]}]\\
    &\le
    \ex_{\run-1}[\lips\norm{\inter[\jstate]-\inter[\surr{\jstate}]}\norm{\vptpast[\noise]}]\\
    &\le
    \ex_{\run-1}\left[
    \lips\left(
    \frac{\norm{\inter[\jstate]-\inter[\surr{\jstate}]}^2}{2\vpt[\stepalt]}+
    \frac{\vpt[\stepalt]\norm{\vptpast[\noise]}^2}{2}
    \right)\right]\\
    &=\ex_{\run-1}\left[
    \lips\left(
    \frac{\vpt[\stepalt]\norm{\vptpast[\noise]}^2}{2}
    +\sumplayers[\playalt]\frac{(\vpt[\step][\playalt]+\vpt[\stepalt][\playalt])^2\norm{\vptpast[\noise][\playalt]}^2}{2\vpt[\stepalt]}
    \right)\right].
\end{aligned}
\end{equation}
Putting \eqref{eq:feedback-product-ex}, \eqref{eq:feedback-product-develop}, and \eqref{eq:feedback-product-noisebound} together gives the desired inequality.
\end{proof}


\paragraph{Quasi-Descent Inequalities for OG+\afterhead}

Below we establish respectively the individual and the global quasi-descent inequalities for OG+.
In this part, all the players use the same learning rate sequences and we can thus drop the player index in the learning rates.

\begin{lemma}[Individual quasi-descent inequality for OG+]
\label{lem:OG+-quasi-descent-individual}
Let \cref{asm:noises,asm:lips} hold and all players run \eqref{OG+} with the same predetermined learning rate sequences.
Then, for all $\allplayers$, $\run\ge2$, and $\vp[\arpoint]\in\vp[\points]$, it holds
%
\begin{equation}
\label{eq:lem:OG+-quasi-descent-individual}
\begin{aligned}[b]
    \ex_{\run-1}[\norm{\vptupdate[\state]-\vp[\arpoint]}^2]
    \le 
    \ex_{\run-1}[&
    \norm{\vpt[\state]-\vp[\arpoint]}^2
    -2\update[\step]
    \product{\vp[\vecfield](\inter[\jstate])}{\vptinter[\state]-\vp[\arpoint]}
    \\
    &-\vt[\stepalt]\update[\step](\norm{\vp[\vecfield](\inter[\jstate])}^2
    +\norm{\vp[\vecfield](\past[\jstate])}^2)\\
    &+\vt[\stepalt]\update[\step]\norm{\vp[\vecfield](\inter[\jstate])-\vp[\vecfield](\past[\jstate])}^2
    +\vt[\stepalt]^2\update[\step]\lips\norm{\vptpast[\noise]}^2
    \\
    &
    +\update[\step](\vt[\step]+\vt[\stepalt])^2\lips
    \norm{\past[\jnoise]}^2
    +(\update[\step])^2\norm{\vptinter[\svecfield]}^2].
\end{aligned}
\end{equation}
\end{lemma}
\begin{proof}
We apply \cref{prop:OG+-develop} to player $\play$'s update and $\arpoint\subs\vp[\arpoint]$. Since the inequality holds for any realization we can take expectation with respect to $\last[\filter]$ to get
\begin{align*}
    \ex_{\run-1}[\norm{\vptupdate[\state]-\vp[\arpoint]}^2]
    =
    \ex_{\run-1}[
    &
    \norm{\vpt[\state]-\vp[\arpoint]}^2
    -2\update[\step]\product{\vptinter[\svecfield]}{\vptinter[\state]-\vp[\arpoint]}
    \\
    &
    -2\vt[\stepalt]\update[\step]\product{\vptinter[\svecfield]}{\vptpast[\svecfield]}
    +(\update[\step])^2\norm{\vptinter[\svecfield]}^2].
\end{align*}
The learning rates $\vt[\stepalt]$ and $\update[\step]$ being $\vt[\filter][1]$-measurable and in particular $\last[\filter]$-measurable, we conclude immediately with
\cref{lem:feedback-product}
and the equality
\[
\ex_{\run-1}[\update[\step]\product{\vptinter[\svecfield]}{\vptinter[\state]-\vp[\arpoint]}]
=\update[\step]\ex_{\run-1}[\product{\vp[\vecfield](\inter[\jstate])}{\vptinter[\state]-\vp[\arpoint]}].
\qedhere
\]
\end{proof}

\begin{remark}
\label{rem:lr-measurability}
From the proof of \cref{lem:OG+-quasi-descent-individual} we see that the exact requirement concerning the measurability of the learning rates here is that both $\vt[\stepalt]$ and $\update[\step]$ should be $\last[\filter]$-measurable.
For simplicity throughout our analysis for \ac{OG+} we simply say that all the learning rates are predetermined, \ie $\vt[\filter][1]$-measurable.
In contrast, for \ac{OptDA+} \cref{asm:stepsize-measurable} is indeed sufficient.
This is a technical detail that we have omitted in the main text.
\end{remark}

\begin{lemma}[Global quasi-descent inequality for OG+]
\label{lem:OG+-quasi-descent}
Let \cref{asm:noises,asm:lips,asm:VS} hold and all players run \eqref{OG+} with the same predetermined learning rate sequences.
Then, for all $\run\ge2$ and $\jsol\in\sols$, we have
\begin{equation}
\notag
\begin{aligned}[b]
    \ex_{\run-1}[\norm{\update[\jstate]-\jsol}^2]
    \le 
    \ex_{\run-1}[&
    \norm{\current[\jstate]-\jsol}^2
    -\vt[\stepalt]\update[\step]
    (\norm{\jvecfield(\inter[\jstate])}^2+\norm{\jvecfield(\past[\jstate])}^2)\\
    &+3\vt[\stepalt]\update[\step]\nPlayers\lips^2(
    (\vt[\step]^2+\vt[\stepalt]^2)\norm{\past[\jsvecfield]}^2
    +(\last[\stepalt])^2\norm{\ancient[\jsvecfield]}^2)\\
    &+(\vt[\stepalt]^2\update[\step]+\nPlayers\update[\step](\vt[\step]+\vt[\stepalt])^2)\lips\norm{\past[\jnoise]}^2
    +(\update[\step])^2\norm{\inter[\jsvecfield]}^2].
\end{aligned}
\end{equation}
\end{lemma}

\begin{proof}
We will apply \cref{lem:OG+-quasi-descent-individual} to $\vp[\sol]$.
We first bound the variation $\norm{\vp[\vecfield](\inter[\jstate])-\vp[\vecfield](\past[\jstate])}^2$ by
\begin{equation}
\begin{aligned}[b]
    \label{eq:OG+-diff-bound}
    \norm{\vp[\vecfield](\inter[\jstate])-\vp[\vecfield](\past[\jstate])}^2
    &\le 
    3 \norm{\vp[\vecfield](\inter[\jstate])-\vp[\vecfield](\vt[\jstate])}^2
    + 3 \norm{\vp[\vecfield](\vt[\jstate])-\vp[\vecfield](\last[\jstate])}^2\\
    &~+ 3 \norm{\vp[\vecfield](\last[\jstate])-\vp[\vecfield](\past[\jstate])}^2\\
    &\le
    3\vt[\stepalt]^2\lips^2\norm{\past[\jsvecfield]}^2
    + 3\vt[\step]^2\lips^2\norm{\past[\jsvecfield]}^2
    + 3(\last[\stepalt])^2\lips^2\norm{\ancient[\jsvecfield]}^2.
\end{aligned}
\end{equation}
In the second inequality, we have used the Lipschitz continuity of $\vp[\vecfield]$ and
$\inter[\jstate]=\vt[\jstate]-\vt[\stepalt]\past[\jsvecfield]$ to obtain
\begin{align*}
    \norm{\vp[\vecfield](\inter[\jstate])-\vp[\vecfield](\vt[\jstate])}^2
    &\le \lips^2 \norm{\inter[\jstate]-\vt[\jstate]}^2
    =3\vt[\stepalt]^2\lips^2\norm{\past[\jsvecfield]}^2.
\end{align*}
The terms $\norm{\vp[\vecfield](\vt[\jstate])-\vp[\vecfield](\last[\jstate])}^2$ and $\norm{\vp[\vecfield](\last[\jstate])-\vp[\vecfield](\past[\jstate])}^2$ were bounded in the same way.
Applying \cref{lem:OG+-quasi-descent-individual} with $\vp[\arpoint]\subs\vp[\sol]$, plugging \eqref{eq:OG+-diff-bound} into \eqref{eq:lem:OG+-quasi-descent-individual}, and summing from $\play=1$ to $\nPlayers$ then yields
\begin{align*}
    \ex_{\run-1}[\norm{\update[\jstate]-\jsol}^2]
    \le 
    \ex_{\run-1}[&
    \norm{\current[\jstate]-\jsol}^2
    -\update[\step]\product{\jvecfield(\inter[\jstate])}{\inter[\jstate]-\jsol}\\
    &-\vt[\stepalt]\update[\step]
    (\norm{\jvecfield(\inter[\jstate])}^2+\norm{\jvecfield(\past[\jstate])}^2)\\
    &+3\vt[\stepalt]\update[\step]\nPlayers\lips^2(
    (\vt[\step]^2+\vt[\stepalt]^2)\norm{\past[\jsvecfield]}^2
    +(\last[\stepalt])^2\norm{\ancient[\jsvecfield]}^2)\\
    &+(\vt[\stepalt]^2\update[\step]+\nPlayers\update[\step](\vt[\step]+\vt[\stepalt])^2)\lips\norm{\past[\jnoise]}^2
    +(\update[\step])^2\norm{\inter[\jsvecfield]}^2].
\end{align*}
To conclude, we drop $-\update[\step]\product{\jvecfield(\inter[\jstate])}{\inter[\jstate]-\jsol}$ which is non-positive by \cref{asm:VS}.
\end{proof}


\paragraph{Quasi-Descent Inequalities for OptDA+\afterhead}

Similarly, we establish quasi-descent inequalities for OptDA+ that will be used for both non-adaptive and adaptive analyses.

\begin{lemma}[Individual quasi-descent inequality for OptDA+]
\label{lem:OptDA+-quasi-descent-individual}
Let \cref{asm:noises,asm:lips} hold and all players run \eqref{OptDA+} with non-increasing learning rates satisfying \cref{asm:stepsize-measurable}.
Then, for all $\allplayers$, $\run\ge2$, and $\vp[\arpoint]\in\vp[\points]$, it holds
\begin{equation}
\label{eq:lem:OptDA+-quasi-descent-individual}
\begin{aligned}[b]
    \ex_{\run-1}\Bigg[
    \frac{\norm{\vptupdate[\state]-\vp[\arpoint]}^2}
    {\vptupdate[\step]}
    \Bigg]
    \le 
    \ex_{\run-1}\Bigg[&
    \frac{\norm{\vpt[\state]-\vp[\arpoint]}^2}{\vpt[\step]}
    +
    \left(
    \frac{1}{\vptupdate[\step]}-\frac{1}{\vpt[\step]}
    \right)\norm{\vpt[\state][\play][1]-\vp[\arpoint]}^2
    \\
    &-
    2\product{\vp[\vecfield](\inter[\jstate])}{\vptinter[\state]-\vp[\arpoint]}
    \\
    &-
    \vpt[\stepalt](\norm{\vp[\vecfield](\inter[\jstate])}^2
    +\norm{\vp[\vecfield](\past[\jstate])}^2)
    \\
    &
    +\vpt[\stepalt]\norm{\vp[\vecfield](\inter[\jstate])-\vp[\vecfield](\past[\jstate])}^2
    \\
    &+
    \min\left(
    -
    \frac{\norm{\vpt[\state]-\vptupdate[\state]}^2}{2\vpt[\step]
    }
    +\vpt[\step]\norm{\vptinter[\svecfield]}^2,
    \, 0
    \right)
    \\
    &
    +(\vpt[\stepalt])^2\lips\norm{\vptpast[\noise]}^2
    +\lips
    \norm{\past[\jnoise]}_{(\vt[\jstep]+\vt[\jstepalt])^2}^2
    +\vpt[\step]\norm{\vptinter[\svecfield]}^2
    \Bigg].
\end{aligned}
\end{equation}
\end{lemma}
\begin{proof}
This is an immediate by combining \cref{cor:OptDA+-descent} and \cref{lem:feedback-product}.
We just notice that as $\vpt[\stepalt]$ is $\last[\filter]$-measurable, we have
$\ex_{\run-1}[\vpt[\stepalt]\product{\vptinter[\svecfield]}{\vptpast[\svecfield]}]=\vpt[\stepalt]\ex_{\run-1}[\product{\vptinter[\svecfield]}{\vptpast[\svecfield]}]$.
\end{proof}

\begin{lemma}[Global quasi-descent inequality for OptDA+]
\label{lem:OptDA+-quasi-descent}
Let \cref{asm:noises,asm:lips,asm:VS} hold and all players run \eqref{OptDA+} with non-increasing learning rates satisfying \cref{asm:stepsize-measurable}.
Then, for all $\run\ge2$ and $\jsol\in\sols$, if
$\vt[\jstep]\le\vt[\jstepalt]$, 
we have
\begin{equation}
\label{eq:lem:OptDA+-quasi-descent}
\begin{aligned}[b]
    \ex_{\run-1}[\norm{\update[\jstate]-\jsol}_{1/\update[\jstep]}^2]
    \le 
    \ex_{\run-1}[&
    \norm{\current[\jstate]-\jsol}_{1/\vt[\jstep]}^2
    +\norm{\vt[\jstate][1]-\jsol}_{1/\update[\jstep]-1/\vt[\jstep]}^2
    \\
    &
    -
    \norm{\jvecfield(\inter[\jstate])}_{\vt[\jstepalt]}^2
    -\norm{\jvecfield(\past[\jstate])}_{\vt[\jstepalt]}^2
    \\
    &
    -\norm{\vt[\jstate]-\update[\jstate]}_{1/(2\vt[\jstep])}^2
    +3\norm{
    \jvecfield(\vt[\jstate])-\jvecfield(\last[\jstate])
    }_{\vt[\jstepalt]}^2
    \\
    &+3\lips^2
    (\norm{\vt[\jstepalt]}_1\norm{\past[\jsvecfield]}_{
    \vt[\jstepalt]^2}^2
    +
    \norm{\last[\jstepalt]}_1
    \norm{\ancient[\jsvecfield]}_{
    (\last[\jstepalt])^2}^2)
    \\
    &+(4
    \nPlayers+1)\lips\norm{\past[\jnoise]}
    _{\vt[\jstepalt]^2}^2
    +2\norm{\inter[\jsvecfield]}_{\vt[\jstep]}^2].
\end{aligned}
\end{equation}
\end{lemma}
\begin{proof}
The result is proved in the same way as \cref{lem:OG+-quasi-descent} but instead of \cref{lem:OG+-quasi-descent-individual} we make use of \cref{lem:OptDA+-quasi-descent-individual} with
\[
\min\left(
    -
    \frac{\norm{\vpt[\state]-\vptupdate[\state]}^2}{2\vpt[\step]
    }
    +\vpt[\step]\norm{\vptinter[\svecfield]}^2,
    \, 0
    \right)
    \le -
    \frac{\norm{\vpt[\state]-\vptupdate[\state]}^2}{2\vpt[\step]
    }
    +\vpt[\step]\norm{\vptinter[\svecfield]}^2.
\]
Moreover, as there is not a simple expression for $\norm{\vt[\jstate]-\update[\jstate]}$,
in the place of \eqref{eq:OG+-diff-bound} we use \begin{equation}
\begin{aligned}[b]
    \label{eq:OptDA+-diff-bound}
    \norm{\vp[\vecfield](\inter[\jstate])-\vp[\vecfield](\past[\jstate])}^2
    &\le
    3\lips^2\norm{\past[\jsvecfield]}^2_{\vt[\jstepalt]^2}
    + 3\lips^2\norm{\ancient[\jsvecfield]}^2_{(\last[\jstepalt])^2}
    + 3 \norm{\vp[\vecfield](\vt[\jstate])-\vp[\vecfield](\last[\jstate])}^2.
\end{aligned}
\end{equation}
To obtain \eqref{eq:lem:OptDA+-quasi-descent}, we further use $\vt[\jstep]\le\vt[\jstepalt]$ and $\norm{\vt[\jstepalt]}_1\le\norm{\last[\jstepalt]}_1$.
\end{proof}
\begin{remark}
The players can take different learning rates in OptDA+ because in the quasi-descent inequality \eqref{eq:lem:OptDA+-quasi-descent-individual}, there is no learning rate in front of $\product{\vp[\vecfield](\inter[\jstate])}{\vptinter[\state]-\vp[\arpoint]}$.
Take $\vp[\arpoint]\subs\vp[\sol]$ and summing from $\play=1$ to $\nPlayers$ we get directly 
$\product{\jvecfield(\inter[\jstate])}{\inter[\jstate]-\jsol}$
which is non-negative according to \cref{asm:VS}.
While it is also possible to put \eqref{eq:lem:OG+-quasi-descent-individual} in the form of \cref{main:lem:quasi-descent-individual}, we are not able to control the sum of     $(1/\vptupdate[\step]-1/\vpt[\step])\norm{\vpt[\state]-\vp[\arpoint]}^2$ as explained in \cref{subsec:quasi-descent}.
\end{remark}

%% file: appendices/regret.tex
In this section, we tackle the regret analysis of OG+ and OptDA+ run with non-adaptive learning rates.
We prove bounds on the pseudo-regret
$\max_{\vp[\arpoint]\in\vp[\cpt]}\ex[\vpt[\reg][\play][\nRuns](\vp[\arpoint])]$
and on the sum of the expected magnitude of the noiseless feedback
$\sum_{\run=1}^{\nRuns}\ex[\norm{\jvecfield(\inter[\jstate])}^2]$.
In fact, in our analysis, building bounds on $\sum_{\run=1}^{\nRuns}\ex[\norm{\jvecfield(\inter[\jstate])}^2]$ 
is a crucial step for deriving bounds on the pseudo-regret.

Moreover, as the loss functions are convex in their respective player's action parameter, a player's regret can be bounded by its linearized counterpart, as stated in the following lemma.
\begin{lemma}
\label{lem:linearized-regret}
Let \cref{asm:lips} hold. Then, for all $\allplayers$, any sequence of actions $(\vt[\jaction])_{\run\in\N}$, and all reference point $\vp[\arpoint]\in\vp[\points]$, we have
\begin{equation*}
    \vpt[\reg][\play][\nRuns](\vp[\arpoint])
    \le
    \sum_{\run=1}^{\nRuns}
    \product{\vp[\vecfield](\vt[\jaction])}
    {\vpt[\action]-\vp[\arpoint]}
\end{equation*}
\end{lemma}
We therefore focus exclusively on bounding the linearized regret in the sequel.

\subsection{Bounds for OG+}
\input{appendices/OG+-bounds}


\subsection{Bounds for OptDA+}
\label{apx:OptDA+-bound-nonadapt}
\input{appendices/OptDA+-bounds}

%% file: appendices/OG+-bounds.tex
In this part we will simply assume the learning rates to be $\vt[\filter][1]$-measurable, a technical detailed that we ignored in the main text.
The global quasi-descent inequality of OG+ introduced in \cref{lem:OG+-quasi-descent} indeed allows us to bound several important quantities, as shown below.

\begin{proposition}[Bound on sum of squared norms]
\label{lem:OG+-sum-bound}
Let \cref{asm:noises,asm:lips,asm:VS} hold and all players run \eqref{OG+} 
with learning rates described in \cref{thm:OG+-regret}.
Then, for all $\nRuns\in\N$ and $\jsol\in\sols$, we have
\begin{equation}
    \notag
    \begin{aligned}
    &\ex[\norm{\update[\jstate]-\jsol}^2]
    +
    \frac{1}{2}\sum_{\run=1}^{\nRuns}
    \vt[\stepalt]\update[\step]
    \ex[\norm{\jvecfield(\inter[\jstate])}^2]\\
    &~~~~
    \le
    \norm{\vt[\jstate][1]-\jsol}^2
    +\vt[\stepalt][1]\vt[\step][2]
    \norm{\vecfield(\vt[\jstate][1])}^2
    +\sum_{\run=1}^{\nRuns}
    \left(
    9\vt[\stepalt]^3\update[\step]\nPlayers\lips^2
    +\vt[\stepalt]^2\update[\step](4\nPlayers+1)\lips
    +(\update[\step])^2
    \right)
    \nPlayers\noisevar.
    \end{aligned}
\end{equation}
Accordingly, $\sum_{\run=1}^{\infty}\vt[\stepalt]\update[\step]\ex[\norm{\jvecfield(\inter[\jstate])}^2]<\infty$.
\end{proposition}
\begin{proof}
Since $\vt[\jsvecfield][1/2]=0$, we have  $\vt[\jstate][3/2]=\vt[\jstate][1]$ and with $\vt[\jstate][2]=\vt[\jstate][1]-\vt[\step][2]\vt[\jsvecfield][3/2]$ we obtain
\begin{equation}
    \notag
    \norm{\vt[\jstate][2]-\jsol}^2
    = \norm{\vt[\jstate][1]-\jsol}^2
    -2\vt[\step][2]\product{\vt[\jsvecfield][3/2]}{\vt[\jstate][3/2]-\jsol}
    +\vt[\step][2]^2\norm{\vt[\jsvecfield][3/2]}^2.
\end{equation}
Taking expectation then gives
\begin{equation}
    \label{eq:OG+-quasi-descent-1}
    \begin{aligned}[b]
    \ex[\norm{\vt[\jstate][2]-\jsol}^2]
    &= \ex[\norm{\vt[\jstate][1]-\jsol}^2
    -2\vt[\step][2]\product{\jvecfield(\vt[\jstate][3/2])}{\vt[\jstate][3/2]-\jsol}
    +\vt[\step][2]^2\norm{\vt[\jsvecfield][3/2]}^2]\\
    &\le
    \ex[\norm{\vt[\jstate][1]-\jsol}^2
    +\vt[\step][2]^2\norm{\vt[\jsvecfield][3/2]}^2],
    \end{aligned}
\end{equation}
where we have used \cref{asm:VS} to deduce that $\product{\jvecfield(\vt[\jstate][3/2])}{\vt[\jstate][3/2]-\jsol}\ge0$.
Taking total expectation of the inequality of 
\cref{lem:OG+-quasi-descent},
summing from $\run=2$ to $\nRuns$, and further adding \eqref{eq:OG+-quasi-descent-1} gives
\begin{equation}
\notag
\begin{aligned}[b]
    &\underbrace{
    \ex\Bigg[
    \norm{\update[\jstate][\nRuns]-\jsol}^2+
    \sum_{\run=2}^{\nRuns}
    \vt[\stepalt]\update[\step]
    (\norm{\jvecfield(\inter[\jstate])}^2+\norm{\jvecfield(\past[\jstate])}^2)\Bigg]}_{\displaystyle (A)}\\
    &~~~\le
    \begin{aligned}[t]
    \ex\Bigg[&
    \norm{\vt[\jstate][1]-\jsol}^2
    +\sum_{\run=1}^{\nRuns}(\update[\step])^2\norm{\inter[\jsvecfield]}^2
    +\sum_{\run=2}^{\nRuns}
    3\vt[\stepalt]\update[\step](\vt[\step]^2+\vt[\stepalt]^2)\nPlayers\lips^2
    \norm{\past[\jsvecfield]}^2\\
    & +\sum_{\run=2}^{\nRuns}
    3\vt[\stepalt]\update[\step](\last[\stepalt])^2\nPlayers\lips^2\norm{\ancient[\jsvecfield]}^2
    +\sum_{\run=2}^{\nRuns}
    \left(
    \vt[\stepalt]^2\update[\step]+\nPlayers\update[\step](\vt[\step]+\vt[\stepalt])^2
    \right)
    \lips\norm{\past[\jnoise]}^2
    \Bigg].
    \end{aligned}
\end{aligned}
\end{equation}
We use \cref{asm:noises} to bound the noise terms. For example, we have
\begin{equation}
    \ex[\norm{\inter[\jnoise]}^2]
    = \sumplayers\ex[\norm{\vptinter[\noise]}^2]
    \le \sumplayers (\noisevar+\noisecontrol\norm{\vp[\vecfield](\vt[\jaction])}^2)
    = \noisecontrol\norm{\jvecfield(\inter[\jstate])}^2
    + \nPlayers\noisevar.
\end{equation}
Subsequently,
\begin{align*}
    \ex[(\update[\step])^2\norm{\inter[\jsvecfield]}^2]
    &= (\update[\step])^2\ex[\norm{\jvecfield(\inter[\jstate])}^2+\norm{\inter[\jnoise]}^2]
    \\
    &\le (\update[\step])^2
    \left(
    \ex[(1+\noisecontrol)\norm{\jvecfield(\inter[\jstate])}^2]+\nPlayers\noisevar
    \right).
\end{align*}
Along with the fact that the learning rates are non-increasing and $\vt[\step][\runalt]\le\vt[\stepalt][\runalt]$ for all $
\runalt\in\N$, we get
\begin{equation}
\label{eq:OG+-sum-expectation-inequality}
\begin{aligned}[b]
    (A)
    \le
    \ex\Bigg[&
    \norm{\vt[\jstate][1]-\jsol}^2
    +\sum_{\run=1}^{\nRuns}(\update[\step])^2
    \left(
    (1+\noisecontrol)\norm{\jvecfield(\inter[\jstate])}^2
    +\nPlayers\noisevar
    \right)
    \\
    &+\sum_{\run=2}^{\nRuns}
    \left(
    6\vt[\stepalt]^3\update[\step]\nPlayers\lips^2(1+\noisecontrol)
    + \vt[\stepalt]^2\update[\step]
    (4\nPlayers+1)\lips\noisecontrol
    \right)
    \norm{\jvecfield(\past[\jstate])}^2
    \\
    &+\sum_{\run=2}^{\nRuns}
    \left(
    6\vt[\stepalt]^3\update[\step]\nPlayers\lips^2
    + \vt[\stepalt]^2\update[\step]
    (4\nPlayers+1)\lips
    \right)
    \nPlayers\noisevar
    \\
    &+\sum_{\run=3}^{\nRuns}
    3(\last[\stepalt])^3\vt[\step]\nPlayers\lips^2
    \left(
    (1+\noisecontrol)\norm{\jvecfield(\ancient[\jstate])}^2
    +\nPlayers\noisevar
    \right)
    \Bigg].
\end{aligned}
\end{equation}
Re-indexing the summations and adding positive terms to the \ac{RHS} of the inequality, we deduce
\begin{equation}
\notag
\begin{aligned}[b]
    (A)
    \le
    \ex\Bigg[&
    \norm{\vt[\jstate][1]-\jsol}^2
    +\sum_{\run=2}^{\nRuns}
    \left(
    9\vt[\stepalt]^3\update[\step]\nPlayers\lips^2(1+\noisecontrol)
    +\vt[\stepalt]^2\update[\step](4\nPlayers+1)\lips\noisecontrol
    \right)
    \norm{\jvecfield(\past[\jstate])}^2
    \\
    &
    +\sum_{\run=1}^{\nRuns}(\update[\step])^2
    (1+\noisecontrol)\norm{\jvecfield(\inter[\jstate])}^2
    \\
    &
    +\sum_{\run=1}^{\nRuns}
    \left(
    (\update[\step])^2
    +9\vt[\stepalt]^3\update[\step]\nPlayers\lips^2
    +\vt[\stepalt]^2\update[\step](4\nPlayers+1)\lips
    \right)
    \nPlayers\noisevar
    \Bigg].
\end{aligned}
\end{equation}
On the other hand, we have
\begin{align*}
    (A)
    &=
    \norm{\update[\jstate][\nRuns]-\jsol}^2
    -\vt[\stepalt][1]\vt[\step][2]
    \norm{\jvecfield(\vt[\jstate][3/2])}^2
    \\
    &~~~+
    \sum_{\run=1}^{\nRuns}
    \vt[\stepalt]\update[\step]
    \ex[\norm{\jvecfield(\inter[\jstate])}^2]
    +
    \sum_{\run=2}^{\nRuns}
    \vt[\stepalt]\update[\step]
    \ex[\norm{\jvecfield(\past[\jstate])}^2].
\end{align*}
Combining the above two (in)equalities, rearranging, and using $\vt[\jstate][3/2]=\vt[\jstate][1]$ leads to
\begin{equation}
    \notag
    \begin{aligned}
    &
    \ex[\norm{\update[\jstate][\nRuns]-\jsol}^2]+
    \sum_{\run=1}^{\nRuns}
    \vt[\stepalt]\update[\step]\left(1-
    \frac{(1+\noisecontrol)\update[\step]}{\vt[\stepalt]}\right)
    \ex[\norm{\jvecfield(\inter[\jstate])}^2]\\
    &~~~~+
    \sum_{\run=2}^{\nRuns}
    \vt[\stepalt]\update[\step](1-
    \vt[\csta](1+\noisecontrol)
    -\vt[\cstb]\noisecontrol)
    \ex[\norm{\jvecfield(\past[\jstate])}^2]
    \\
    &~~\le
    \norm{\vt[\jstate][1]-\jsol}^2
    +\vt[\stepalt][1]\vt[\step][2]
    \norm{\vecfield(\vt[\jstate][1])}^2
    +\sum_{\run=1}^{\nRuns}
    \vt[\stepalt]\update[\step]\left(
    \frac{\update[\step]}{\vt[\stepalt]}+\vt[\csta]+\vt[\cstb]\right)\nPlayers\noisevar,
    \end{aligned}
\end{equation}
where $\vt[\csta]=9\vt[\stepalt]^2\nPlayers\lips^2$ and $\vt[\cstb]=\vt[\stepalt](4\nPlayers+1)\lips$.
To conclude, we notice that with the learning rate choices of \cref{thm:OG+-regret}, it always holds $1-(1+\noisecontrol)(\update[\step]/\vt[\stepalt])\ge1/2$,
$1-\vt[\csta](1+\noisecontrol)-\vt[\cstb]\noisecontrol\ge0$,
and
$\sumtoinf\vt[\stepalt]\update[\step](\update[\step]/\vt[\stepalt]+\vt[\csta]+\vt[\cstb])\nPlayers\noisevar<+\infty$.
\end{proof}


From \cref{lem:OG+-sum-bound} we obtain immediately the bounds on $\sum_{\run=1}^{\nRuns}\ex[\norm{\jvecfield(\inter[\state])}^2]$ of OG+ as claimed in \cref{sec:trajectory}.

\begin{theorem}
\label{thm:OG+-bound-V2}
Let \cref{asm:noises,asm:lips,asm:VS} hold and all players run \eqref{OG+} with non-increasing learning rate sequences $(\vt[\stepalt])_{\run\in\N}$ and $(\vt[\step])_{\run\in\N}$ satisfying \eqref{OG+-lr}.
We have
\begin{enumerate}[(a), leftmargin=*]
    \item
    \label{thm:OG+-bound-V2-add}
    If there exists $\exponent\in[0,1/4]$ such that
    $\vt[\stepalt]=\bigoh(1/(\run^{\frac{1}{4}}\sqrt{\log\run}))$,
    $\vt[\stepalt]=\Omega(1/\run^{\frac{1}{2}-\exponent})$,
    and $\vt[\step]=\Theta(1/(\sqrt{\run}\log\run))$, then
    \begin{equation}
    \notag
    \sum_{\run=1}^{\nRuns}\ex[\norm{\jvecfield(\inter[\state])}^2]=\tbigoh\left(
    \nRuns^{1-\exponent}
    \right)
    \end{equation}
    \item
    \label{thm:OG+-bound-V2-mul}
    If  the noise is multiplicative (\ie $\noisedev=0$) and the learning rates are constant $\vt[\stepalt]\equiv\stepalt$, $\vt[\step]\equiv\step$, then
    \begin{equation}
    \notag
    \begin{aligned}
    \sum_{\run=1}^{\nRuns}
    \ex[\norm{\jvecfield(\inter[\jstate])}^2]
    \le
    \frac{2\dist(\vt[\jstate][1],\sols)^2}
    {\stepalt\step}
    +2\norm{\jvecfield(\vt[\jstate][1])}^2.
    \end{aligned}
    \end{equation}
    In particular, if the equalities hold in \eqref{OG+-lr}, then the above is in $\bigoh(\nPlayers^3\lips^2(1+\noisecontrol)^3)$.
\end{enumerate}
\end{theorem}
\begin{proof}
Let $\jsol=\proj_{\sols}(\vt[\jstate][1])$.
By the choice of our learning rates, the constant
\begin{equation}
    \notag
    \Cst\defeq
    \norm{\vt[\jstate][1]-\jsol}^2
    +\vt[\stepalt][1]\vt[\step][2]
    \norm{\vecfield(\vt[\jstate][1])}^2
    +\sum_{\run=1}^{+\infty}
    (9\vt[\stepalt]^3\update[\step]\nPlayers\lips^2
    +\vt[\stepalt]^2\update[\step](4\nPlayers+1)\lips
    +(\update[\step])^2)
    \nPlayers\noisevar.
\end{equation}
is finite.
In addition, from \cref{lem:OG+-sum-bound} we know tat
\[
\sum_{\run=1}^{\nRuns}\vt[\stepalt]\update[\step]\ex[\norm{\jvecfield(\inter[\jstate])}^2]\le2\Cst.
\]
On the other hand since the learning rates are non-increasing, it holds
\[
\sum_{\run=1}^{\nRuns}\vt[\stepalt]\update[\step]\ex[\norm{\jvecfield(\inter[\jstate])}^2]
\ge
\update[\stepalt][\nRuns]\update[\step][\nRuns]\sum_{\run=1}^{\nRuns}\ex[\norm{\jvecfield(\inter[\jstate])}^2].
\]
As a consequence,
\begin{equation}
    \label{eq:OG+-V2-final-bound}
    \sum_{\run=1}^{\nRuns}\ex[\norm{\jvecfield(\inter[\jstate])}^2]
    \le \frac{2\Cst}{\update[\stepalt][\nRuns]\update[\step][\nRuns]}.
\end{equation}
The results are then immediate from our choice of learning rates.
\end{proof}

\begin{remark}
In the estimation of (b) we use $\dist(\vt[\jstate][1],\sols)^2=\bigoh(\nPlayers)$
and $1/\stepalt=\bigoh(\nPlayers\lips(1+\noisecontrol))$.
We can get improved dependence on $\nPlayers$ if the noises of the players are supposed to be mutually independent conditioned on the past. In fact, in this case we only require $\stepalt \le \min
    \left(\frac{1}{3\lips\sqrt{2\nPlayers(1+\noisecontrol)}}
    ,\frac{1}{8\lips\noisecontrol}\right)$.
\end{remark}

\paragraph{Bounding Linearized Regret.}

We proceed to bound the linearized regret.
The following lemma is a direct consequence of the individual quasi-descent inequality of \cref{lem:OG+-quasi-descent-individual}.

\begin{lemma}[Bound on linearized regret]
\label{lem:OG+-linreg-bound}
Let \cref{asm:noises,asm:lips,asm:VS} hold and all players run \eqref{OG+} 
with learning rates described in \cref{thm:OG+-regret}.
Then, for all $\allplayers$, $\nRuns\in\N$, and $\vp[\arpoint]\in\vp[\points]$, we have
\begin{equation}
    \notag
    \begin{aligned}
    \sum_{\run=1}^{\nRuns}
    \ex[
    \product{\vp[\vecfield](\inter[\jstate])}
    {\vptinter[\state]-\vp[\arpoint]}]
    \le 
    \ex\Bigg[&
    \frac{\norm{\vpt[\state][\play][1]-\vp[\arpoint]}^2}{2\vt[\step][2]}
    +
    \sum_{\run=2}^{\nRuns}\left(
    \frac{1}{2\update[\step]}-\frac{1}{2\vt[\step]}
    \right)
    \norm{\vpt[\state]-\vp[\arpoint]}^2\\
    &+
    \sum_{\run=1}^{\nRuns}
    \left(
    \frac{3\vt[\stepalt]}{4}
    \norm{\jvecfield(\inter[\jstate])}^2
    +\frac{\vt[\csta]\noisevar}{2}
    \right)
    \Bigg],
    \end{aligned}
\end{equation}
where
$\vt[\csta]=9\vt[\stepalt]^3\lips^2\nPlayers
+\vt[\stepalt]^2(4\nPlayers+1)\lips
+\update[\step]$.
\end{lemma}
\begin{proof}
Applying \cref{lem:OG+-quasi-descent-individual},
dividing both sides of \eqref{eq:lem:OG+-quasi-descent-individual} by $\update[\step]$, rearranging, taking total expectation, and using \cref{asm:noises}, we get
\begin{align*}
    \ex[2\langle\vp[\vecfield]&(\inter[\jstate]),\vptinter[\state]-\vp[\arpoint]\rangle]
    \\
    \le 
    \ex\Bigg[&
    \frac{\norm{\vpt[\state]-\vp[\arpoint]}^2}{\update[\step]}
    -
    \frac{\norm{\vptupdate[\state]-\vp[\arpoint]}^2}{\update[\step]}
    \\
    &
    -\vt[\stepalt](\norm{\vp[\vecfield](\inter[\jstate])}^2
    +\norm{\vp[\vecfield](\past[\jstate])}^2)
    +\vt[\stepalt]\norm{\vp[\vecfield](\inter[\jstate])
    -\vp[\vecfield](\past[\jstate])}^2
    \\
    &
    +\vt[\stepalt]^2\lips\noisecontrol\norm{\vp[\vecfield](\past[\jstate])}^2
    +(\vt[\step]+\vt[\stepalt])^2\lips\noisecontrol
    \norm{\jvecfield(\past[\jstate])}^2
    +\update[\step](1+\noisecontrol)\norm{\vp[\vecfield](\inter[\jstate])}^2\\
    &
    +\vt[\stepalt]^2\lips\noisevar
    +(\vt[\step]+\vt[\stepalt])^2\lips
    \nPlayers\noisevar
    +\update[\step]\noisevar
    \Bigg]\\
    \le
    \ex\Bigg[&
    \frac{\norm{\vpt[\state]-\vp[\arpoint]}^2}{\update[\step]}
    -
    \frac{\norm{\vptupdate[\state]-\vp[\arpoint]}^2}{\update[\step]}\\
    &
    + 3\vt[\stepalt]^3\lips^2\norm{\past[\jsvecfield]}^2
    + 3\vt[\stepalt]\vt[\step]^2\lips^2\norm{\past[\jsvecfield]}^2
    + 3\vt[\stepalt](\last[\stepalt])^2\lips^2\norm{\ancient[\jsvecfield]}^2\\
    &
    +5\vt[\stepalt]^2\lips\noisecontrol\norm{\jvecfield(\past[\jstate])}^2
    +\update[\step](1+\noisecontrol)\norm{\jvecfield(\inter[\jstate])}^2
    +\vt[\stepalt]^2(4\nPlayers+1)\lips\noisevar
    + \update[\step]\noisevar
    \Bigg]\\
    \le
    \ex\Bigg[&
    \frac{\norm{\vpt[\state]-\vp[\arpoint]}^2}{\vt[\step]}
    +
    \left(
    \frac{1}{\update[\step]}-\frac{1}{\vt[\step]}
    \right)
    \norm{\vpt[\state]-\vp[\arpoint]}^2
    -
    \frac{\norm{\vptupdate[\state]-\vp[\arpoint]}^2}{\update[\step]}\\
    &
    +\frac{\vt[\stepalt]}{2}
    \norm{\jvecfield(\inter[\jstate])}^2
    +\frac{5\vt[\stepalt]}{6}\norm{\jvecfield(\past[\jstate])}^2
    + 3(\last[\stepalt])^3\lips^2\norm{\ancient[\jsvecfield]}^2\\
    &
    +6\vt[\stepalt]^3\lips^2\nPlayers\noisevar
    +\vt[\stepalt]^2(4\nPlayers+1)\lips\noisevar
    +\update[\step]\noisevar
    \Bigg].
\end{align*}
In the last inequality we have used
$\vt[\step]^2\le
\vt[\stepalt]^2\le
1/(18\lips^2(1+\noisecontrol))$,
$5\vt[\stepalt]\lips\noisecontrol\le\vt[\stepalt][1](4\nPlayers+1)\lips\noisecontrol\le1/2$ and
$\update[\step](1+\noisecontrol)\le\vt[\step](1+\noisecontrol)\le\vt[\stepalt]/2$.
As for the $\norm{\ancient[\jsvecfield]}^2$ term, we recall that $\vt[\jsvecfield][1/2]=0$ and otherwise its expectation can again be bounded using \cref{asm:noises}.
Summing the above inequality from $\run=2$ to $\nRuns$ and dividing both sides by $2$, we then obtain
\begin{equation}
    \label{eq:OG+-sum-regret-absnoise-2toT}
    \begin{aligned}[b]
    \sum_{\run=2}^{\nRuns}
    \ex[
    \product{\vp[\vecfield](\inter[\jstate])}{\vptinter[\state]-\vp[\arpoint]}]
    \le
    \ex\Bigg[&
    \frac{\norm{\vpt[\state][\play][2]-\vp[\arpoint]}^2}{2\vt[\step][2]}
    +
    \sum_{\run=2}^{\nRuns}\left(
    \frac{1}{2\update[\step]}-\frac{1}{2\vt[\step]}
    \right)
    \norm{\vpt[\state]-\vp[\arpoint]}^2\\
    &+
    \sum_{\run=2}^{\nRuns}
    \frac{1}{4}
    \Bigg(
    \vt[\stepalt]
    \norm{\jvecfield(\inter[\jstate])}^2
    +2\vt[\stepalt]\norm{\jvecfield(\past[\jstate])}^2\\
    &~~~~~~~~~~
    + 18\vt[\stepalt]^3\lips^2\nPlayers\noisevar
    + \vt[\stepalt]^2(8\nPlayers+2)\lips\noisevar
    + 2\update[\step]\noisevar
    \Bigg)
    \Bigg].
    \end{aligned}
\end{equation}
For $\run=1$, since $\vpt[\state][\play][3/2]=\vpt[\state][\play][1]$ and $\vpt[\state][\play][2]=\vpt[\state][\play][1]-\vt[\step][2]\vpt[\svecfield][\play][3/2]$, we have
\begin{equation}
    \notag
    \norm{\vpt[\state][\play][2]-\vp[\arpoint]}^2
    = \norm{\vpt[\state][1]-\vp[\arpoint]}^2
    -2\vt[\step][2]\product{\vpt[\svecfield][\play][3/2]}{\vpt[\state][\play][3/2]-\vp[\arpoint]}
    +\vt[\step][2]^2\norm{\vpt[\svecfield][\play][3/2]}^2.
\end{equation}
Taking expectation then gives
\begin{equation}
    \label{eq:OG+-sum-regret-absnoise-1}
    \ex[\norm{\vpt[\state][\play][2]-\vp[\arpoint]}^2]
    \le \ex[\norm{\vpt[\state][\play][1]-\vp[\arpoint]}^2
    -2\vt[\step][2]\product{\vp[\vecfield](\vt[\jstate][3/2])}{\vpt[\state][\play][3/2]-\vp[\arpoint]}
    +\vt[\step][2]^2(1+\noisecontrol)
    \norm{\vp[\vecfield](\vt[\jstate][3/2])}^2
    +\vt[\step][2]^2\noisevar].
\end{equation}
Combining \eqref{eq:OG+-sum-regret-absnoise-2toT} and \eqref{eq:OG+-sum-regret-absnoise-1} and bounding $\vt[\step][2](1+\noisecontrol)\norm{\vp[\vecfield](\vt[\jstate][3/2])}^2\le(\vt[\stepalt]/2)\norm{\jvecfield(\vt[\jstate][3/2])}^2$, we get the desired inequality.
\end{proof}

With \cref{lem:OG+-linreg-bound} and \cref{lem:OG+-sum-bound}, we are now ready to prove our result concerning the regret of OG+.
The main difficulty here consists in controlling the sum of $(1/(2\update[\step])-1/(2\vt[\step]))\ex[\norm{\vpt[\state]-\vp[\arpoint]}^2]$ when the learning rates are not constant.

\begin{theorem}
\label{thm:OG+-regret-apx}
Let \cref{asm:noises,asm:lips,asm:VS} hold and all players run \eqref{OG+} with non-increasing learning rate sequences $(\vt[\stepalt])_{\run\in\N}$ and $(\vt[\step])_{\run\in\N}$ satisfying \eqref{OG+-lr}.
For any $\play\in\players$ and bounded set $\vp[\cpt]\subset\vp[\points]$ with $\radius\ge\sup_{\vp[\arpoint]}\norm{\vpt[\state][\play][1]-\vp[\arpoint]}$, we have:
\begin{enumerate}[(a), leftmargin=*]
    \item
    \label{thm:OG+-regret-add}
    If $\vt[\stepalt]=\bigoh(1/(\run^{\frac{1}{4}}\sqrt{\log\run}))$ and $\vt[\step]=\Theta(1/(\sqrt{\run}\log\run))$, then
    \begin{equation}
    \notag
    \max_{\vp[\arpoint]\in\vp[\cpt]}
    \ex\left[
    \sum_{\run=1}^{\nRuns}
    \product{\vp[\vecfield](\inter[\jstate])}
    {\vptinter[\state]-\vp[\arpoint]}
    \right]
    =\tbigoh\left(\sqrt{\nRuns}\right).
    \end{equation}
    \item
    \label{thm:OG+-regret-mul}
    If  the noise is multiplicative (\ie $\noisedev=0$) and the learning rates are constant $\vt[\stepalt]\equiv\stepalt$, $\vt[\step]\equiv\step$, then
    \begin{equation}
    \notag
    \max_{\vp[\arpoint]\in\vp[\cpt]}
    \ex\left[
    \sum_{\run=1}^{\nRuns}
    \product{\vp[\vecfield](\inter[\jstate])}
    {\vptinter[\state]-\vp[\arpoint]}
    \right]
    \le 
    \frac{\radius^2}{2\step}
    +
    \frac{2}{\step}
    (\dist(\vt[\jstate][1],\sols)^2
    +\stepalt\step\norm{\jvecfield(\vt[\jstate][1])}^2).
    \end{equation}
    In particular, if the equalities hold in \eqref{OG+-lr}, the above is in
    $\bigoh(\nPlayers^2\lips(1+\noisecontrol)^2)$.
\end{enumerate}
\end{theorem}
\begin{proof}
Let $\jsol=\proj_{\sols}(\vt[\jstate][1])$ be the projection of $\vt[\jstate][1]$ onto the solution set. For any $\vp[\arpoint]\in\vp[\cpt]$, it holds
\begin{align}
\notag
&\sum_{\run=2}^{\nRuns}\left(
    \frac{1}{2\update[\step]}-\frac{1}{2\vt[\step]}
    \right)
    \norm{\vpt[\state]-\vp[\arpoint]}^2
\\
\notag
&~~~\le
\sum_{\run=2}^{\nRuns}\left(
    \frac{1}{\update[\step]}-\frac{1}{\vt[\step]}
    \right)
    \left(\norm{\vpt[\state]-\vp[\sol]}^2
    +
    \norm{\vp[\sol]-\vp[\arpoint]}^2\right)
    \\
    \notag
    &~~~\le
    \sum_{\run=2}^{\nRuns}\left(
    \frac{1}{\update[\step]}-\frac{1}{\vt[\step]}
    \right)
    \left(
    \norm{\vt[\jstate]-\jsol}^2
    +\norm{\vp[\sol]
    -\vpt[\state][\play][1]
    +\vpt[\state][\play][1]
    -\vp[\arpoint]}^2
    \right)
    \\
    \label{eq:OG+-sum-regret-diff-lr}
    &~~~\le
    \left(
    \frac{1}{\update[\step][\nRuns]}-\frac{1}{\vt[\step][2]}\right)
    \left(
    2\norm{\vpt[\state][\play][1]-\vp[\sol]}^2
    +2\radius^2
    \right)
    +
    \sum_{\run=2}^{\nRuns}\left(
    \frac{1}{\update[\step]}-\frac{1}{\vt[\step]}
    \right)
    \norm{\vt[\jstate]-\jsol}^2.
\end{align}
To proceed, with \cref{lem:OG+-sum-bound}, we know that 
for $\Cst$ defined in 
the proof of \cref{thm:OG+-bound-V2}, we have for all $\run\in\N$
\begin{equation}
    \label{eq:OG+-bound-sum-exp-re}
    \begin{aligned}
    \ex[\norm{\update[\jstate]-\jsol}^2]
    +
    \sum_{\runalt=1}^{\run}
    \frac{\vt[\stepalt][\runalt]\update[\step][\runalt]}{2}
    \ex[\norm{\jvecfield(\inter[\jstate][\runalt])}^2]
    &\le \Cst
    \end{aligned}
\end{equation}
%
We can therefore write
\begin{equation}
    \label{eq:OG+-regret-bound-cst1}
    \sum_{\run=2}^{\nRuns}
    \left(
    \frac{1}{\update[\step]}-\frac{1}{\vt[\step]}
    \right)
    \ex[\norm{\vt[\jstate]-\jsol}^2]
    \le  \sum_{\run=2}^{\nRuns}
    \left(
    \frac{1}{\update[\step]}-\frac{1}{\vt[\step]}
    \right)\Cst
    \le \frac{\Cst}{\update[\step][\nRuns]}.
\end{equation}
Since $\update[\step][\nRuns]\le\update[\step]$ for all $\run\le\nRuns$.
From \eqref{eq:OG+-bound-sum-exp-re} we also deduce
\begin{equation}
    \label{eq:OG+-regret-bound-cst2}
    \sum_{\run=1}^{\nRuns}
    \vt[\stepalt]\ex[\norm{\jvecfield(\inter[\jstate])}^2]
    \le
    \frac{1}{\update[\step][\nRuns]}
    \sum_{\run=1}^{\nRuns}
    \vt[\stepalt]\update[\step]
    \ex
    [\norm{\jvecfield(\inter[\jstate])}^2]
    \le \frac{2\Cst}{\update[\step][\nRuns]}.
\end{equation}
Plugging \eqref{eq:OG+-sum-regret-diff-lr}, \eqref{eq:OG+-regret-bound-cst1}, and  \eqref{eq:OG+-regret-bound-cst2} into \cref{lem:OG+-linreg-bound},
we obtain
\begin{equation}
    \notag
    \begin{aligned}[b]
    \sum_{\run=1}^{\nRuns}
    \ex[
    \product{\vp[\vecfield](\inter[\jstate])}
    {\vptinter[\state]-\vp[\arpoint]}]
    \le 
    \ex\Bigg[&
    \frac{2\dist(\vt[\jstate][1],\sols)^2
    +2\radius^2
    }{\vt[\step][\nRuns+1]}
    +
    \frac{3\Cst}{\update[\step][\nRuns]}
    +
    \sum_{\run=1}^{\nRuns}
    \frac{
    \vt[\csta]\noisevar}{2}
    \Bigg].
    \end{aligned}
\end{equation}
The result is now immediate from $\vt[\stepalt]=\bigoh(1/(\run^{\frac{1}{4}}\sqrt{\log\run}))$ and $\vt[\step]=\Theta(1/(\sqrt{\run}\log\run))$.

\vspace{0.4em}
\noindent
(b) Let $\vp[\arpoint]\in\vp[\cpt]$.
With $\noisevar=0$, constant learning rates, and $\norm{\vpt[\state][\play][1]-\vp[\arpoint]}\le\radius^2$, \cref{lem:OG+-linreg-bound} gives
\begin{equation}
    \notag
    \begin{aligned}[b]
    \sum_{\run=1}^{\nRuns}
    \ex[
    \product{\vp[\vecfield](\inter[\jstate])}
    {\vptinter[\state]-\vp[\arpoint]}]
    \le 
    \ex\Bigg[&
    \frac{\radius^2}{2\step}
    +
    \sum_{\run=1}^{\nRuns}
    \frac{3\stepalt}{4}
    \norm{\jvecfield(\inter[\jstate])}^2
    \Bigg],
    \end{aligned}
\end{equation}
We conclude immediately with the help of \cref{thm:OG+-bound-V2}\ref{thm:OG+-bound-V2-mul}.
\end{proof}

%% file: appendices/OptDA+-bounds.tex
For the analysis of OptDA+, we first establish two preliminary bounds respectively for the linearized regret and for the sum of the squared operator norms.
These bounds are used later for deriving more refined bounds in the non-adaptive and the adaptive case.
We use the notation $\vpt[\step][\play][1]=\vpt[\step][\play][2]$.

\begin{lemma}[Bound on linearized regret]
\label{lem:OptDA+-lin-regret-bound}
Let \cref{asm:noises,asm:lips} hold and all players run \eqref{OptDA+} with non-increasing learning rates satisfying \cref{asm:stepsize-measurable}
and $\vt[\jstep]\le\vt[\jstepalt]$ for all $\run\in\N$.
Then, for all $\allplayers$, $\nRuns\in\N$, and $\vp[\arpoint]\in\vp[\points]$, we have
\begin{equation}
    \begin{aligned}[b]
    \ex\left[
    \sum_{\run=1}^{\nRuns}
    \product{
    \vp[\vecfield](\inter[\jstate])
    }
    {\vptinter[\state]-\vp[\arpoint]}
    \right]
    \le
    \ex\Bigg[&
    \frac{
    \norm{\vpt[\state][\play][1]-\vp[\arpoint]}^2}{2\vptupdate[\step][\play][\nRuns]}
    +
    \frac{1}{2}\sum_{\run=1}^{\nRuns}
    \vpt[\step]\norm{\vptinter[\svecfield]}^2
    \\
    &+
    \sum_{\run=2}^{\nRuns}
    \vpt[\stepalt]\lips^2
    \left(
    3\norm{\past[\jsvecfield]}^2_{\vt[\jstepalt]^2}
    +\frac{3}{2}\norm{\vt[\jstate]-\last[\jstate]}^2
    \right)
    \\
    &+
    \frac{1}{2}\sum_{\run=2}^{\nRuns}
    ((\vpt[\stepalt])^2\lips\norm{\vptpast[\noise]}^2
    +4\lips
    \norm{\past[\jnoise]}_{\vt[\jstepalt]^2}^2)
    \Bigg].
    \end{aligned}
\end{equation}
\end{lemma}
\begin{proof}
Applying \cref{lem:OptDA+-quasi-descent-individual},
dropping non-positive terms on the \ac{RHS} of \eqref{eq:lem:OptDA+-quasi-descent-individual}, using
\[
\min\left(
    -
    \frac{\norm{\vpt[\state]-\vptupdate[\state]}^2}{2\vpt[\step]
    }
    +\vpt[\step]\norm{\vptinter[\svecfield]}^2,
    \, 0
    \right)
\le 0
\]
and taking total expectation gives
\begin{equation}
\label{eq:OptDA+-t>=2}
\begin{aligned}[b]
    \ex\Bigg[
    \frac{\norm{\vptupdate[\state]-\vp[\arpoint]}^2}
    {\vptupdate[\step]}
    \Bigg]
    \le 
    \ex\Bigg[&
    \frac{\norm{\vpt[\state]-\vp[\arpoint]}^2}{\vpt[\step]}
    +
    \left(
    \frac{1}{\vptupdate[\step]}-\frac{1}{\vpt[\step]}
    \right)\norm{\vpt[\state][\play][1]-\vp[\arpoint]}^2
    \\
    &
    -
    2\product{\vp[\vecfield](\inter[\jstate])}{\vptinter[\state]-\vp[\arpoint]}
    +\vpt[\stepalt]\norm{\vp[\vecfield](\inter[\jstate])-\vp[\vecfield](\past[\jstate])}^2
    \\
    &
    +(\vpt[\stepalt])^2\lips\norm{\vptpast[\noise]}^2
    +\lips
    \norm{\past[\jnoise]}_{(\vt[\jstep]+\vt[\jstepalt])^2}^2
    +\vpt[\step]\norm{\vptinter[\svecfield]}^2
    \Bigg].
\end{aligned}
\end{equation}
The above inequality holds for $\run\ge2$.
As for $\run=1$, we notice that with $\vpt[\state][\play][2]=\vpt[\state][\play][1]-\vpt[\step][\play][2]\vpt[\svecfield][\play][3/2]$, we have in fact
\begin{equation}
    \notag
    \norm{\vpt[\state][\play][2]-\vp[\arpoint]}^2
    =\norm{\vpt[\state][\play][1]-\vp[\arpoint]}^2
    -2\vpt[\step][\play][2]
    \product{\vpt[\svecfield][\play][3/2]
    }{\vpt[\state][\play][1]-\vp[\arpoint]}
    +(\vpt[\step][\play][2])^2\norm{\vpt[\svecfield][\play][3/2]}^2.
\end{equation}
As $\vpt[\state][\play][3/2]=\vpt[\state][\play][1]=0$ and $\vpt[\step][\play][1]=\vpt[\step][\play][2]$, the above implies
\begin{equation}
    \label{eq:OptDA+-t=1-div}
    \begin{aligned}[b]
    \ex
    \left[\product{\vp[\vecfield](\vpt[\state][\play][3/2])
    }{\vpt[\state][\play][3/2]-\vp[\arpoint]}\right]
    &=
    \ex\Bigg[
    \frac{\norm{\vpt[\state][\play][1]-\vp[\arpoint]}^2}{2\vpt[\step][\play][2]}
    -\frac{\norm{\vpt[\state][\play][2]-\vp[\arpoint]}^2}{2\vpt[\step][\play][2]}
    +
    \frac{\vpt[\step][\play][1]\norm{\vpt[\svecfield][\play][3/2]}^2}{2}
    \Bigg]
    .
    \end{aligned}
\end{equation}
Summing \eqref{eq:OptDA+-t>=2} from $\run=2$ to $\nRuns$, dividing by $2$, adding \eqref{eq:OptDA+-t=1-div}, and using $\vt[\jstep]\le\vt[\jstepalt]$ leads to
\begin{equation}
    \notag
    \begin{aligned}[b]
    \sum_{\run=1}^{\nRuns}
    \ex[
    \product{\vp[\vecfield](\inter[\jstate])}{\vptinter[\state]-\vp[\arpoint]}]
    \le
    \frac{1}{2}\ex\Bigg[&
    \frac{
    \norm{\vpt[\state][\play][1]-\vp[\arpoint]}^2}{\vptupdate[\step][\play][\nRuns]}
    +
    \sum_{\run=1}^{\nRuns}
    \vpt[\step]\norm{\vptinter[\svecfield]}^2
    \\
    &
    +
    \sum_{\run=2}^{\nRuns}
    \vpt[\stepalt]\norm{\vp[\vecfield](\inter[\jstate])-\vp[\vecfield](\past[\jstate])}^2
    \\
    &+
    \sum_{\run=2}^{\nRuns}
    ((\vpt[\stepalt])^2\lips\norm{\vptpast[\noise]}^2
    +4\lips
    \norm{\past[\jnoise]}_{\vt[\jstepalt]^2}^2)
    \Bigg].
    \end{aligned}
\end{equation}
Similar to \eqref{eq:OG+-diff-bound},
we can bound the difference term by 
\begin{equation}
\notag
\begin{aligned}[b]
    \norm{\vp[\vecfield](\inter[\jstate])-\vp[\vecfield](\past[\jstate])}^2
    &\le
    3\lips^2\norm{\past[\jsvecfield]}^2_{\vt[\jstepalt]^2}
    + 3\lips^2\norm{\ancient[\jsvecfield]}^2_{(\last[\jstepalt])^2}
    + 3 \lips^2\norm{\vt[\jstate]-\last[\jstate]}^2.
\end{aligned}
\end{equation}
Combining the above two inequalities and using $\vt[\jsvecfield][1/2]=0$ gives the desired inequality.
\end{proof}

\begin{lemma}[Bound on sum of squared norms]
\label{lem:OptDA+-sum-bound}
Let \cref{asm:noises,asm:lips,asm:VS} hold and all players run \eqref{OptDA+} with non-increasing learning rates satisfying \cref{asm:stepsize-measurable} and $\vt[\jstep]\le\vt[\jstepalt]$ for all $\run\in\N$.
Then, for all $\nRuns\in\N$ and $\jsol\in\sols$, 
we have
\begin{equation}
\label{eq:lem:OptDA+-sum-bound}
\begin{aligned}[b]
    \sum_{\run=2}^{\nRuns}
    \ex[\norm{&\jvecfield(\inter[\jstate])}_{\vt[\jstepalt]}^2
    +\norm{\jvecfield(\past[\jstate])}_{\vt[\jstepalt]}^2]
    \\
    \le
    \ex\Bigg[&
    \norm{\vt[\jstate][1]-\jsol}_{1/\vt[\jstep][\nRuns+1]}^2
    + \sum_{\run=1}^{\nRuns}
    \left(
    3\norm{
    \jvecfield(\vt[\jstate])-\jvecfield(\update[\jstate])
    }_{\vt[\jstepalt]}^2
    -\norm{\vt[\jstate]-\update[\jstate]}_{1/(2\vt[\jstep])}^2\right)
    \\
    &
    +
    \sum_{\run=2}^{\nRuns}
    6\norm{\vt[\jstepalt]}_1\lips^2\norm{\past[\jsvecfield]}^2_{
    \vt[\jstepalt]^2}
    +\sum_{\run=2}^{\nRuns}
    (4\nPlayers+1)\lips\norm{\past[\jnoise]}_{\vt[\jstepalt]^2}^2
    +\sum_{\run=1}^{\nRuns}
    2\norm{\inter[\jsvecfield]}_{\vt[\jstep]}^2\Bigg].
\end{aligned}
\end{equation}
\end{lemma}
\begin{proof}
This is a direct consequence of \cref{lem:OptDA+-quasi-descent}.
In fact, taking total expectation of \eqref{eq:lem:OptDA+-quasi-descent} and summing from $\run=2$ to $\nRuns$ gives already
\begin{equation}
\label{eq:OptDA+-sum-exp-2}
\begin{aligned}[b]
    \sum_{\run=2}^{\nRuns}
    \ex[\norm{&\jvecfield(\inter[\jstate])}_{\vt[\jstepalt]}^2
    +\norm{\jvecfield(\past[\jstate])}_{\vt[\jstepalt]}^2]
    \\
    \le
    \ex\Bigg[&
    \norm{\vt[\jstate][2]-\jsol}_{1/\vt[\jstep][2]}^2
    +
    \norm{\vt[\jstate][1]-\jsol}_{1/\vt[\jstep][\nRuns+1]-1/\vt[\jstep][2]}^2
    \\
    &
    + \sum_{\run=2}^{\nRuns}
    (3\norm{
    \jvecfield(\vt[\jstate])-\jvecfield(\last[\jstate])
    }_{\vt[\jstepalt]}^2
    -\norm{\vt[\jstate]-\update[\jstate]}_{1/(2\vt[\jstep])}^2)
    \\
    &
    +
    \sum_{\run=2}^{\nRuns}
    6\norm{\vt[\jstepalt]}_1\lips^2\norm{\past[\jsvecfield]}_{
    \vt[\jstepalt]^2}^2
    +\sum_{\run=2}^{\nRuns}
    (4\nPlayers+1)\lips\norm{\past[\jnoise]}_{\vt[\jstepalt]^2}^2
    +\sum_{\run=2}^{\nRuns}
    2\norm{\inter[\jsvecfield]}_{\vt[\jstep]}^2
    \Bigg].
\end{aligned}
\end{equation}
We have in particular used $\vt[\jsvecfield][1/2]=0$ to bound
\begin{equation}
    \notag
    \begin{aligned}
    &\sum_{\run=2}^{\nRuns}
    3\lips^2(
    \norm{\vt[\jstepalt]}_1
    \norm{\past[\jsvecfield]}_{
    \vt[\jstepalt]^2}^2
    +\norm{\last[\jstepalt]}_1\norm{\ancient[\jsvecfield]}_{
    (\last[\jstepalt])^2}^2)
    \\
    &~~=
    \sum_{\run=2}^{\nRuns}
    3\norm{\vt[\jstepalt]}_1\lips^2\norm{\past[\jsvecfield]}_{
    \vt[\jstepalt]^2}^2
    +
    \sum_{\run=3}^{\nRuns}
    \norm{\vt[\jstepalt]}_1
    3\nPlayers\lips^2\norm{\past[\jsvecfield]}_{
    \vt[\jstepalt]^2}^2\\
    &~~\le 
    \sum_{\run=2}^{\nRuns}
    6\norm{\vt[\jstepalt]}_1\lips^2\norm{\past[\jsvecfield]}_{
    \vt[\jstepalt]^2}^2.
    \end{aligned}
\end{equation}
To obtain \eqref{eq:lem:OptDA+-sum-bound}, we further bound
\begin{equation}
    \label{eq:OptDA+-sum-exp-add-term}
    \sum_{\run=2}^{\nRuns}
    3\norm{
    \jvecfield(\vt[\jstate])-\jvecfield(\last[\jstate])
    }_{\vt[\jstepalt]}^2
    =
    \sum_{\run=1}^{\nRuns-1}
    3\norm{
    \jvecfield(\vt[\jstate])-\jvecfield(\update[\jstate])
    }_{\update[\jstepalt]}^2
    \le
    \sum_{\run=1}^{\nRuns}
    3\norm{
    \jvecfield(\vt[\jstate])-\jvecfield(\update[\jstate])
    }_{\vt[\jstepalt]}^2
\end{equation}
%
%
For $\run=1$, we use \eqref{eq:OptDA+-t=1-div} with $\vp[\arpoint]\subs\vp[\sol]$; that is
\begin{equation}
    \notag
    \frac{\norm{\vpt[\state][\play][2]-\vp[\sol]}^2}{\vpt[\step][\play][2]}
    =\frac{\norm{\vpt[\state][\play][1]-\vp[\sol]}^2}{\vpt[\step][\play][2]}
    -2
    \product{\vp[\vecfield](\vt[\jstate][3/2])+\vpt[\noise][\play][3/2]
    }{\vpt[\state][\play][1]-\vp[\sol]}
    +\vpt[\step][\play][1]\norm{\vpt[\svecfield][\play][3/2]}^2.
\end{equation}
%
Since $\vt[\jstate][3/2]=\vt[\jstate][1]$,
summing the above inequality from $\play=1$ to $\nPlayers$ leads to
\begin{equation}
    \label{eq:OptDA+-t=1}
    \norm{\vt[\jstate][2]-\jsol}_{1/\vt[\jstep][2]}^2
    =\norm{\vt[\jstate][1]-\jsol}_{1/\vt[\jstep][2]}^2
    -2\product{\jvecfield(\vt[\jstate][3/2])+\vt[\jnoise][3/2]}{\vt[\jstate][3/2]-\jsol}
    +\norm{\vt[\jsvecfield][3/2]}_{\vt[\jstep][1]}^2.
\end{equation}
\cref{asm:noises,asm:VS} together ensure
\begin{equation}
    \notag
\ex[\product{\jvecfield(\vt[\jstate][3/2])+\vt[\jnoise][3/2]}{\vt[\jstate][3/2]-\jsol}]
=\product{\jvecfield(\vt[\jstate][3/2])}{\vt[\jstate][3/2]-\jsol}\ge0.
\end{equation}
Subsequently,
\begin{equation}
    \label{eq:OptDA+-sum-exp-1}
    \begin{aligned}[b]
    \ex[\norm{\vt[\jstate][2]-\jsol}_{1/\vt[\jstep][2]}^2]
    &
    \le\ex[\norm{\vt[\jstate][1]-\jsol}_{1/\vt[\jstep][2]}^2
    +\norm{\vpt[\svecfield][\play][3/2]}_{\vt[\jstep][2]}^2]
    \\
    &
    \le
    \ex[\norm{\vt[\jstate][1]-\jsol}_{1/\vt[\jstep][2]}^2
    +2\norm{\vpt[\svecfield][\play][3/2]}_{\vt[\jstep][1]}^2
    -\norm{\vt[\jstate][1]
    -\vt[\jstate][2]}_{(1/2\vt[\jstep][1])}^2
    ].
    \end{aligned}
\end{equation}
Combining \eqref{eq:OptDA+-sum-exp-2}, \eqref{eq:OptDA+-sum-exp-add-term}, and \eqref{eq:OptDA+-sum-exp-1} gives exactly \eqref{eq:lem:OptDA+-sum-bound}.
\end{proof}

\subsubsection{Dedicated Analysis for Non-Adaptive OptDA+}

In this part, we show how the non-adaptive learning rates suggested in \cref{thm:OptDA+-regret} helps to achieve small regret and lead to fast convergence of the norms of the payoff gradients.

\begin{proposition}[Bound on sum of squared norms]
\label{lem:OptDA+-sum-bound-nonadapt}
Let \cref{asm:noises,asm:lips,asm:VS} hold and all players run \eqref{OptDA+} with learning rates described in \cref{thm:OptDA+-regret}.
Then, for all $\nRuns\in\N$ and $\jsol\in\sols$, 
we have
\begin{equation}
    \notag
    \begin{aligned}
    &\frac{1}{2}\sum_{\runalt=1}^{\nRuns}
    \ex[\norm{\jvecfield(\inter[\jstate])}^2_{\vt[\jstepalt]}]
    +
    \sum_{\run=1}^{\nRuns}
    21\infnorm{\vt[\jstepalt][1]}\nPlayers\lips^2
    \ex[\norm{\vt[\jstate]-\update[\jstate]}^2]\\
    &~~\le
    \norm{\vt[\jstate][1]-\jsol}_{1/\vt[\jstep][\nRuns+1]}^2
    +
    \norm{\jvecfield(\vt[\jstate][1])}^2_{\vt[\jstepalt][1]}
    +\sum_{\run=1}^{\nRuns}
    \left(
    6\infnorm{\vt[\jstepalt]}^3
    \nPlayers\lips^2
    +\infnorm{\vt[\jstepalt]}^2(4\nPlayers+1)\lips
    +2\infnorm{\vt[\jstep]}
    \right)
    \nPlayers\noisevar
    \end{aligned}
\end{equation}
\end{proposition}
\begin{proof}
We first apply \cref{lem:OptDA+-sum-bound} to obtain \eqref{eq:lem:OptDA+-sum-bound}.
We bound the expectations of the following three terms separately.
\begin{align*}
    \vt[A]
    &= 3\norm{
    \jvecfield(\vt[\jstate])-\jvecfield(\update[\jstate])
    }_{\vt[\jstepalt]}^2
    -\norm{\vt[\jstate]-\update[\jstate]}_{1/(2\vt[\jstep])}^2,\\
    \vt[B]
    &= 6\onenorm{\vt[\jstepalt]}\lips^2\norm{\past[\jsvecfield]}_{
    \vt[\jstepalt]^2}^2
    +(4\nPlayers+1)\lips\norm{\past[\jnoise]}_{\vt[\jstepalt]^2}^2,~~~~
    \vt[C]
    = 2\norm{\inter[\jsvecfield]}_{\vt[\jstep]}^2.
\end{align*}

To bound $\vt[A]$, we first use 
$\vt[\jstep]
\le\vt[\jstepalt]/(4(1+\noisecontrol))
\le\norm{\vt[\jstepalt][1]}_{\infty}/(4(1+\noisecontrol))$ to get 
\begin{equation*}
    \norm{\vt[\jstate]-\update[\jstate]}_{1/(2\vt[\jstep])}^2
    \ge \frac{2(1+\noisecontrol)}{\infnorm{\vt[\jstepalt][1]}}\norm{\vt[\jstate]-\update[\jstate]}^2.
\end{equation*}
Moreover, with $\infnorm{\vt[\jstepalt][1]}^2
\le 1/(12\nPlayers\lips^2(1+\noisecontrol))$ we indeed have
\begin{equation}
    \notag
    \frac{2(1+\noisecontrol)}{\infnorm{\vt[\jstepalt][1]}}
    \ge 24\nPlayers\lips^2(1+\noisecontrol)^2\infnorm{\vt[\jstepalt][1]}
    \ge 24\nPlayers\lips^2\infnorm{\vt[\jstepalt][1]}.
\end{equation}
On the other hand, with the Lipschitz continuity of $(\vp[\vecfield])_{\allplayers}$ it holds
\begin{equation}
    \notag
    \begin{aligned}
    3\norm{
    \jvecfield(\vt[\jstate])-\jvecfield(\update[\jstate])
    }_{\vt[\jstepalt]}^2
    \le 
    \sumplayers
    3\vpt[\stepalt]\lips^2\norm{\vt[\jstate]-\update[\jstate]}^2
    \le 3\infnorm{\vt[\jstepalt][1]}\nPlayers\lips^2
    \norm{\vt[\jstate]-\update[\jstate]}^2.
    \end{aligned}
\end{equation}
Combining the above inequalities we deduce that $\vt[A]\le-21\infnorm{\vt[\jstepalt][1]}\nPlayers\lips^2
\norm{\vt[\jstate]-\update[\jstate]}^2$ and accordingly %
\begin{equation}
\label{eq:OptDA+-nonadapt-cancel-bound}
\ex[\vt[A]]\le\ex[-21\infnorm{\vt[\jstepalt][1]}\nPlayers\lips^2
\norm{\vt[\jstate]-\update[\jstate]}^2].
\end{equation}
%

We proceed to bound $\ex[\vt[B]]$.
The exploration learning rates $\vt[\jstepalt]$ being
$\last[\filter]$-measurable, using \cref{asm:noises} and the law of total expectation, we get
\begin{equation}
\label{eq:OptDA+-sum-bound-nonadapt-past-noise}
\begin{aligned}[b]
    \ex[\vt[B]]
    &=
    \ex\,[\,\ex_{\run-1}[6\onenorm{\vt[\jstepalt]}\lips^2\norm{\past[\jsvecfield]}_{
    \vt[\jstepalt]^2}
    +(4\nPlayers+1)\lips\norm{\past[\jnoise]}_{\vt[\jstepalt]}^2]]\\
    &=
    \ex\Bigg[
    \sumplayers
    \left(
    6\onenorm{\vt[\jstepalt]}(\vpt[\stepalt])^2\lips^2
    \ex_{\run-1}[\norm{\vptpast[\svecfield]}^2]
    +(\vpt[\stepalt])^2
    (4\nPlayers+1)\lips
    \ex_{\run-1}[\norm{\vptpast[\noise]}^2]
    \right)
    \Bigg]
    \\
    &\le
    \ex\Bigg[
    6\infnorm{\vt[\jstepalt]}^2\nPlayers\lips^2(1+\noisecontrol)\norm{\jvecfield(\past[\jstate])}^2_{\vt[\jstepalt]}
    +\infnorm{\vt[\jstepalt]}(4\nPlayers+1)\lips\noisecontrol\norm{\jvecfield(\past[\jstate])}_{\vt[\jstepalt]}^2\\
    & ~~~~~~~~+
    (6\infnorm{\vt[\jstepalt]}^3\nPlayers\lips^2
    +\infnorm{\vt[\jstepalt]}^2
    (4\nPlayers+1)\lips)\nPlayers\noisevar
    \Bigg].
\end{aligned}
\end{equation}
Similarly, $\update[\jstep]$ being
deterministic and in particular 
$\vt[\filter]$-measurable, we have
\begin{equation}
    \label{eq:OptDA+-sum-bound-nonadapt-current-noise}
    \ex[\vt[C]]
    =\ex[\ex_{\run}[2\norm{\inter[\jsvecfield]}_{\vt[\jstep]}^2]]
    \le \ex\Bigg[2(1+\noisecontrol)\norm{\jvecfield(\inter[\jstate])}^2_{\vt[\jstep]}
    + 
    2\infnorm{\vt[\jstep]}\nPlayers\noisevar
    \Bigg].
\end{equation}
Putting together \eqref{eq:lem:OptDA+-sum-bound}, \eqref{eq:OptDA+-nonadapt-cancel-bound}, \eqref{eq:OptDA+-sum-bound-nonadapt-past-noise}, and \eqref{eq:OptDA+-sum-bound-nonadapt-current-noise}, we get
\begin{equation}
\notag
\begin{aligned}[b]
    &\sum_{\run=2}^{\nRuns}
    \ex[\norm{\jvecfield(\inter[\jstate])}_{
    \vt[\jstepalt]-2(1+\noisecontrol)\vt[\jstep]
    }^2
    +
    (1-\vt[\csta](1+\noisecontrol)-\vt[\cstb]\noisecontrol)
    \norm{\jvecfield(\past[\jstate])}_{\vt[\jstepalt]}^2
    ]\\
    &~~~\le
    \begin{aligned}[t]
    \ex\Bigg[&
    \norm{\vt[\jstate][1]-\jsol}_{1/\vt[\jstep][\nRuns+1]}^2
    +
    2(1+\noisecontrol)\norm{\jvecfield(\vt[\jstate][3/2])}^2_{\vt[\jstep][1]}
    -\sum_{\run=1}^{\nRuns}
    21\infnorm{\vt[\jstepalt][1]}\nPlayers\lips^2
    \norm{\vt[\jstate]-\update[\jstate]}^2\\
    &+\sum_{\run=1}^{\nRuns}
    (\vt[\csta]+\vt[\cstb]
    +2\infnorm{\vt[\jstep]})
    \nPlayers\noisevar\Bigg],
    \end{aligned}
\end{aligned}
\end{equation}
where $\vt[\csta]=6\infnorm{\vt[\jstepalt]}^3\nPlayers\lips^2$
and $\vt[\cstb]=\infnorm{\vt[\jstepalt]}^2(4\nPlayers+1)\lips$.
We conclude by using $\vt[\jstate][3/2]=\vt[\jstate][1]$ and noticing that under our learning rate requirement it is always true that
$1-6\infnorm{\vt[\jstepalt]}^2\nPlayers\lips^2(1+\noisecontrol)
-\infnorm{\vt[\jstepalt]}(4\nPlayers+1)\lips\noisecontrol\ge0$
and $\vt[\jstepalt]-2(1+\noisecontrol)\vt[\jstep]\ge\vt[\jstepalt]/2$.
\end{proof}
\begin{remark}
We notice that in the analysis, we can replace the common Lipschitz constant by the ones that are proper to each player (\ie $\vp[\vecfield]$ is $\vp[\lips]$-Lipschitz continuous) when bounding $\vt[B]$.
This is however hot the case for our bound on $\vt[A]$, unless we bound directly $\vpt[\stepalt](\vp[\lips])^2$ by a constant.
\end{remark}


Again, from \cref{lem:OptDA+-sum-bound-nonadapt} we obtain immediately the bounds on $\sum_{\run=1}^{\nRuns}\ex[\norm{\jvecfield(\inter[\state])}^2]$ of non-adaptive OptDA+ as claimed in \cref{sec:trajectory}.

\begin{theorem}
\label{thm:OptDA+-bound-V2}
Let \cref{asm:noises,asm:lips,asm:VS} hold and all players run \eqref{OptDA+} with non-increasing learning rate sequences $(\vpt[\stepalt])_{\run\in\N}$ and $(\vpt[\step])_{\run\in\N}$ satisfying \eqref{OptDA+-lr}.
We have
\begin{enumerate}[(a), leftmargin=*]
    \item
    \label{thm:OptDA+-bound-V2-add}
    If there exists $\exponent\in[0,1/4]$ such that
    $\vpt[\stepalt][\playalt] = \bigoh(1/\run^{\frac{1}{4}})$,
    $\vpt[\stepalt][\playalt] = \Omega(1/\run^{\frac{1}{2}-\exponent})$,
    and $\vpt[\step][\playalt] = \Theta(1/\sqrt{\run})$ for all $\allplayers[\playalt]$, then
    \begin{equation}
    \notag
    \sum_{\run=1}^{\nRuns}\ex[\norm{\jvecfield(\inter[\state])}^2]=\bigoh\left(\nRuns^{1-\exponent}\right)
    \end{equation}
    \item
    \label{thm:OptDA+-bound-V2-mul}
    If  the noise is multiplicative (\ie $\noisedev=0$) and the learning rates are constant $\vt[\stepalt]\equiv\stepalt$, $\vt[\step]\equiv\step$, then
    \begin{equation}
    \notag
    \begin{aligned}
    \sum_{\run=1}^{\nRuns}
    \ex[\norm{\jvecfield(\inter[\jstate])}^2]
    \le
    \frac{2}{\min_{\allplayers}\vp[\stepalt][\play]}\left(
    \dist_{1/\jstep}(\vt[\jstate][1],\sols)^2
    +
    \norm{\jvecfield(\vt[\jstate][1])}^2_{\vt[\jstepalt]}\right)
    \end{aligned}
    \end{equation}
    In particular, if the equalities hold in \eqref{OptDA+-lr}, then the above is in
    $\bigoh(\nPlayers^3\lips^2(1+\noisecontrol)^3)$.
\end{enumerate}
\end{theorem}
\begin{proof}
Let us define $\vt[\csta]=
6\infnorm{\vt[\jstepalt]}^3
\nPlayers\lips^2
+\infnorm{\vt[\jstepalt]}^2(4\nPlayers+1)\lips
+2\infnorm{\vt[\jstep]}$.
From \cref{lem:OptDA+-sum-bound-nonadapt} we know that for all $\jsol\in\sols$, it holds
\begin{equation}
    \notag
    \sum_{\runalt=1}^{\nRuns}
    \ex[\norm{\jvecfield(\inter[\jstate])}^2_{\vt[\jstepalt]/2}]]
    \le
    \norm{\vt[\jstate][1]-\jsol}_{1/\vt[\jstep][\nRuns+1]}^2
    +
    \norm{\jvecfield(\vt[\jstate][1])}^2_{\vt[\jstepalt][1]}
    +\sum_{\run=1}^{\nRuns}
    \vt[\csta]
    \nPlayers\noisevar,
\end{equation}
Since the learning rates are decreasing, we can lower bound $\vt[\jstepalt]$ by $\vt[\jstepalt]\ge\vt[\jstepalt][\nRuns]\ge\min_{\allplayers}\vpt[\stepalt][\play][\nRuns]$.
Accordingly,
\begin{equation}
    \label{eq:OptDA+-sum-normV2-nonadapt}
    \sum_{\runalt=1}^{\nRuns}
    \ex[\norm{\jvecfield(\inter[\jstate])}^2]]
    \le
    \frac{2}{\min_{\allplayers}\vpt[\stepalt][\play][\nRuns]}\left(
    \norm{\vt[\jstate][1]-\jsol}_{1/\vt[\jstep][\nRuns+1]}^2
    +
    \norm{\jvecfield(\vt[\jstate][1])}^2_{\vt[\jstepalt][1]}
    +\sum_{\run=1}^{\nRuns}
    \vt[\csta]
    \nPlayers\noisevar\right),
\end{equation}
The result then follows immediately from our learning rate choices.
For (a), we observe that with $\infnorm{\vt[\jstepalt]}=\bigoh(1/\run^{\frac{1}{4}})$ and $\infnorm{\vt[\jstep]}=\bigoh(1/\sqrt{\run})$,
we have $\sum_{\run=1}^{\nRuns}\vt[\csta]=\bigoh(\sqrt{\nRuns})$, while
$\vpt[\stepalt][\playalt] = \Omega(1/\run^{\frac{1}{2}-\exponent})$,
and $\vpt[\step][\playalt] = \Omega(1/\sqrt{\run})$ guarantees
$1/\min_{\allplayers}\vpt[\stepalt][\nRuns]=\bigoh(\nRuns^{\frac{1}{2}-\exponent})$ and
$1/\min_{\allplayers}\vptupdate[\step][\nRuns]=\bigoh(\sqrt{\nRuns})$.
For (b), we take $\jsol=\argmin_{\jaction\in\sols}\norm{\vt[\jstate][1]-\jaction}_{1/\jstep}$.
\end{proof}

\paragraph{Bounding Linearized Regret\afterhead}

To bound the linearized regret, we refine \cref{lem:OptDA+-lin-regret-bound} as follows.

\begin{lemma}[Bound on linearized regret]
\label{lem:OptDA+-lin-regret-bound-nonadapt}
Let \cref{asm:noises,asm:lips,asm:VS} hold and all players run \eqref{OptDA+} with learning rates described in \cref{thm:OptDA+-regret}.
Then, for all $\play\in\players$, $\nRuns\in\N$, and $\vp[\arpoint]\in\vp[\points]$, we have
\begin{equation}
    \notag
    \begin{aligned}
    \ex\left[
    \sum_{\run=1}^{\nRuns}
    \product{
    \vp[\vecfield](\inter[\jstate])
    }
    {\vptinter[\state]-\vp[\arpoint]}
    \right]
    \le
    \ex\Bigg[&
    \frac{
    \norm{\vpt[\state][\play][1]-\vp[\arpoint]}^2}
    {2\vptupdate[\step][\play][\nRuns]}
    +
    \sum_{\run=1}^{\nRuns}
    \frac{5}{8}
    \norm{\jvecfield(\inter[\jstate])}^2_{\vt[\jstepalt]}
    \\
    &
    +
    \sum_{\run=1}^{\nRuns-1}
    \frac{3\infnorm{\vt[\jstepalt][1]}\lips^2}{2}\norm{\vt[\jstate]-\update[\jstate]}^2
    \\
    &+
    \frac{1}{2}\sum_{\run=1}^{\nRuns}
    \left(
    6\infnorm{\vt[\jstepalt]}^3\nPlayers\lips^2
    +\infnorm{\vt[\jstepalt]}^2(4\nPlayers+1)\lips
    +\vpt[\step]
    \right)\noisevar
    \Bigg].
    \end{aligned}
\end{equation}
\end{lemma}
\begin{proof}
Thanks to \cref{lem:OptDA+-lin-regret-bound} and \cref{asm:noises}, we can bound
\begin{equation}
    \notag
    \begin{aligned}
    &\left[
    \sum_{\run=1}^{\nRuns}
    \product{
    \vp[\vecfield](\inter[\jstate])
    }
    {\vptinter[\state]-\vp[\arpoint]}
    \right]\\
    &~~~
    \begin{aligned}
    \le
    \ex\Bigg[&
    \frac{
    \norm{\vpt[\state][\play][1]-\vp[\arpoint]}^2}{2\vptupdate[\step][\play][\nRuns]}
    +
    \sum_{\run=2}^{\nRuns}
    \vpt[\stepalt]\lips^2
    \left(3(1+\noisecontrol)
    \norm{\jvecfield(\past[\jstate])}^2_{\vt[\jstepalt]^2}
    +3\onenorm{\vt[\jstepalt]^2}\noisevar
    +\frac{3}{2}\norm{\vt[\jstate]-\last[\jstate]}^2
    \right)\\
    &+
    \frac{1}{2}\sum_{\run=2}^{\nRuns}
    \left(
    (\vpt[\stepalt])^2\lips
    (\noisecontrol\norm{\vp[\vecfield](\past[\jstate])
    }^2+\noisevar)
    +4\lips
    (\noisecontrol\norm{\jvecfield(\past[\jstate])}^2_{\vt[\jstepalt]^2}
    +\onenorm{\vt[\jstepalt]^2}\noisevar
    )
    \right)\\
    &+
    \frac{1}{2}\sum_{\run=1}^{\nRuns}
    \vpt[\step]
    \left(
    (1+\noisecontrol)
    \norm{\vp[\vecfield](\inter[\jstate])}^2
    +\noisevar\right)
    \Bigg].
    \end{aligned}
    \end{aligned}
\end{equation}
In the following, we further bound the above inequality using
\begin{enumerate*}[\itshape i\upshape)]
\item $\vpt[\step]\le\vpt[\stepalt]/(4(1+\noisecontrol))$,
\item $\update[\jstepalt]\le\vt[\jstepalt]$,
\item $\vpt[\weight]\norm{\vp[\vecfield](\jaction)}^2\le
\norm{\jvecfield(\jaction)}^2_{\weights}$ for any $\weights\in\R_+^{\nPlayers}$ and $\jaction\in\points$, and
\item $\infnorm{\weights}=\max_{\allplayers}\vp[\weight]$ and in particular $\onenorm{\weights^2}\le\nPlayers\infnorm{\weights}^2$ for  $\weights\in\R_+^{\nPlayers}$.
\end{enumerate*}
\begin{align*}
    &\left[
    \sum_{\run=1}^{\nRuns}
    \product{
    \vp[\vecfield](\inter[\jstate])
    }
    {\vptinter[\state]-\vp[\arpoint]}
    \right]
    \\
    &~~~\begin{aligned}
    \le
    \ex\Bigg[&
    \frac{
    \norm{\vpt[\state][\play][1]-\vp[\arpoint]}^2}
    {2\vptupdate[\step][\play][\nRuns]}
    +
    \sum_{\run=2}^{\nRuns}
    3\infnorm{\vt[\jstepalt]}^2\lips^2
    \left(
    (1+\noisecontrol)
    \norm{\jvecfield(\past[\jstate])}^2_{\vt[\jstepalt]}
    +\infnorm{\vt[\jstepalt]}\nPlayers\noisevar
    \right)
    \\
    &+
    \frac{1}{2}\sum_{\run=2}^{\nRuns}
    \left(
    \infnorm{\vt[\jstepalt]}\lips\noisecontrol
    (\vpt[\stepalt]\norm{\vp[\vecfield](\past[\jstate])}^2
    +4\norm{\jvecfield(\past[\jstate])}^2_{\vt[\jstepalt]})
    +\infnorm{\vt[\jstepalt]}^2(4\nPlayers+1)\lips\noisevar
    \right)
    \\
    &
    +
    \sum_{\run=2}^{\nRuns}
    \frac{3\infnorm{\vt[\jstepalt]}\lips^2}{2}\norm{\vt[\jstate]-\last[\jstate]}^2
    +
    \frac{1}{2}\sum_{\run=1}^{\nRuns}
    \left(
    \frac{\vpt[\stepalt]}{4}
    \norm{\vp[\vecfield](\inter[\jstate])}^2
    +\vpt[\step]\noisevar
    \right)
    \Bigg]
    \end{aligned}
    \\
    &~~~\begin{aligned}
    \le
    \ex\Bigg[&
    \frac{
    \norm{\vpt[\state][\play][1]-\vp[\arpoint]}^2}
    {2\vptupdate[\step][\play][\nRuns]}
    +
    \sum_{\run=1}^{\nRuns}
    \left(3\infnorm{\vt[\jstepalt]}^2\lips^2(1+\noisecontrol)
    +\frac{5\infnorm{\vt[\jstepalt]}\lips\noisecontrol}{2}
    +\frac{1}{8}
    \right)
    \norm{\jvecfield(\inter[\jstate])}^2_{\vt[\jstepalt]}
    \\
    &
    +\sum_{\run=1}^{\nRuns-1}
    \frac{3\infnorm{\vt[\jstepalt][1]}\lips^2}{2}\norm{\vt[\jstate]-\update[\jstate]}^2
    \\
    &+
    \frac{1}{2}\sum_{\run=1}^{\nRuns}
    \left(
    6\infnorm{\vt[\jstepalt]}^3\nPlayers\lips^2\noisevar
    +\infnorm{\vt[\jstepalt]}^2(4\nPlayers+1)\lips\noisevar
    +\vpt[\step]\noisevar
    \right)
    \Bigg].
    \end{aligned}
\end{align*}
To conclude, we notice that under that our learning rate requirements it holds that
$3\infnorm{\vt[\jstepalt]}^2\lips^2(1+\noisecontrol)
+5\infnorm{\vt[\jstepalt]}\lips\noisecontrol/2\le1/2$.
\end{proof}

Our main regret guarantees of non-adaptive OptDA+ follows from the combination of \cref{lem:OptDA+-lin-regret-bound-nonadapt} and \cref{lem:OptDA+-sum-bound-nonadapt}.

\begin{theorem}
\label{thm:OptDA+-regret-apx}
Let \cref{asm:noises,asm:lips,asm:VS} hold and all players run \eqref{OptDA+} with non-increasing learning rate sequences $(\vpt[\stepalt])_{\run\in\N}$ and $(\vpt[\step])_{\run\in\N}$ satisfying \eqref{OptDA+-lr}.
For any $\play\in\players$ and bounded set $\vp[\cpt]\subset\vp[\points]$ with $\radius\ge\sup_{\vp[\arpoint]}\norm{\vpt[\state][\play][1]-\vp[\arpoint]}$, we have:
\begin{enumerate}[(a), leftmargin=*]
    \item If $\vpt[\stepalt][\playalt] = \bigoh(1/\run^{\frac{1}{4}})$ and $\vpt[\step][\playalt] = \Theta(1/\sqrt{\run})$ for all $\allplayers[\playalt]$, then
    \begin{equation}
    \notag
    \max_{\vp[\arpoint]\in\vp[\cpt]}
    \ex\left[
    \sum_{\run=1}^{\nRuns}
    \product{\vp[\vecfield](\inter[\jstate])}
    {\vptinter[\state]-\vp[\arpoint]}
    \right]
    =\bigoh\left(\sqrt{\nRuns}\right).
    \end{equation}
    \item If the noise is multiplicative (\ie $\noisedev=0$) and the learning rates are constant $\vt[\stepalt]\equiv\stepalt$, $\vt[\step]\equiv\step$, then
    \begin{equation}
    \notag
    \max_{\vp[\arpoint]\in\vp[\cpt]}
    \ex\left[
    \sum_{\run=1}^{\nRuns}
    \product{\vp[\vecfield](\inter[\jstate])}
    {\vptinter[\state]-\vp[\arpoint]}
    \right]
    \le 
    \frac{\radius^2}{2\vp[\step]}
    +
    \frac{5}{4}
    \left(
    \dist_{1/\jstep}(\vt[\jstate][1],\sols)^2
    +\norm{\jvecfield(\vt[\jstate][1])}^2_{\jstepalt}
    \right).
    \end{equation}
    In particular, if the equalities hold in \eqref{OptDA+-lr}, the above is in
    $\bigoh(\nPlayers^2\lips(1+\noisecontrol)^2)$.
\end{enumerate}
\end{theorem}
\begin{proof}
Let $\vp[\arpoint]\in\vp[\cpt]$ and $\jsol=\argmin_{\jaction\in\sols}\norm{\vt[\jstate][1]-\jaction}_{1/\jstep}$.
We define
$\vt[\csta]=
6\infnorm{\vt[\jstepalt]}^3
\nPlayers\lips^2
+\infnorm{\vt[\jstepalt]}^2(4\nPlayers+1)\lips
+2\infnorm{\vt[\jstep]}$
as in the proof of \cref{thm:OptDA+-bound-V2}.
Combining \cref{lem:OptDA+-sum-bound-nonadapt} and \cref{lem:OptDA+-lin-regret-bound-nonadapt}, we know that
\begin{equation}
    \notag
    \begin{aligned}[b]
    &\ex\left[
    \sum_{\run=1}^{\nRuns}
    \product{
    \vp[\vecfield](\inter[\jstate])
    }
    {\vptinter[\state]-\vp[\arpoint]}
    \right]
    \\
    &~~~\le
    \ex\Bigg[
    \frac{\radius^2}{2\vptupdate[\step][\play][\nRuns]}
    +
    \frac{1}{2}\sum_{\run=1}^{\nRuns}
    \vt[\csta]\noisevar
    +
    \frac{5}{4}
    \left(
    \norm{\vt[\jstate][1]-\jsol}_{1/\vt[\jstep][\nRuns+1]}^2
    +
    \norm{\jvecfield(\vt[\jstate][1])}^2_{\vt[\jstepalt][1]}
    +\sum_{\run=1}^{\nRuns}\vt[\csta]\nPlayers\noisevar
    \right)
    \Bigg].
    \end{aligned}
\end{equation}
The claims of the theorem follow immediately.
\end{proof}

To close this section, we bound the regret of non-adaptive OptDA+ when played against arbitrary opponents.

\begin{proposition}
\label{prop:OptDA+-regret-adversarial-apx}
Let \cref{asm:noises} hold and player $\play$ run \eqref{OptDA+} with non-increasing learning rates $\vpt[\stepalt] = \Theta(1/\run^{\frac{1}{2}-\exponent})$ and
$\vpt[\step] = \Theta(1/\sqrt{\run})$
for some $\exponent\in[0,1/4]$.
Then, if there exists $\gbound\in\R_+$ such that
$\sup_{\vp[\action]\in\vp[\points]}\norm{\vp[\vecfield](\vp[\action])}\le\gbound$, it holds for any bounded set $\vp[\cpt]$
with $\radius\ge\sup_{\vp[\arpoint]\in\vp[\cpt]}\norm{\vpt[\state][\play][1]-\vp[\arpoint]}$ that
\begin{equation}
    \notag
    \max_{\vp[\arpoint]\in\vp[\cpt]}
    \ex\left[
    \sum_{\run=1}^{\nRuns}
    \product{\vp[\vecfield](\inter[\jstate])}
    {\vptinter[\state]-\vp[\arpoint]}
    \right]
    =
    \bigoh\left(
    \radius^2\sqrt{\nRuns}
    +((1+\noisecontrol)\gbound^2+\noisevar)\nRuns^{\frac{1}{2}+\exponent}
    \right).
\end{equation}
\end{proposition}
\begin{proof}
Let $\vp[\arpoint]\in\vp[\cpt]$.
From \cref{cor:OptDA+-descent} and Young's inequality we get
\begin{equation}
    \notag
    \begin{aligned}[b]
    \product{\vptinter[\svecfield]}{\vptinter[\state]-\vp[\arpoint]}
    &\le
    \frac{\norm{\vpt[\state]-\vp[\arpoint]}^2}{2\vpt[\step]}
    -
    \frac{\norm{\vptupdate[\state]-\vp[\arpoint]}^2}{2\vptupdate[\step]}
    -
    \frac{\norm{\vpt[\state]-\vptupdate[\state]}^2}{2\vpt[\step]}\\
    &~~+
    \left(\frac{1}{2\vptupdate[\step]}
    -\frac{1}{2\vpt[\step]}\right)\norm{\vpt[\state][\play][1]-\vp[\arpoint]}^2
    - \vpt[\stepalt] \product{\vptinter[\svecfield]}{\vptpast[\svecfield]}
    + \vpt[\step]\norm{\vptinter[\svecfield]}^2\\
    &\le
    \frac{\radius^2}{2\vpt[\step]}
    -
    \frac{\norm{\vptupdate[\state]-\vp[\arpoint]}^2}{2\vptupdate[\step]}
    -
    \frac{\norm{\vpt[\state]-\vptupdate[\state]}^2}{2\vpt[\step]}\\
    &~~+
    \left(\frac{1}{2\vptupdate[\step]}
    -\frac{1}{2\vpt[\step]}\right)\norm{\vpt[\state][\play][1]-\vp[\arpoint]}^2
    + \frac{\vpt[\stepalt]}{2}
    (\norm{\vptinter[\svecfield]}^2
    +\norm{\vptpast[\svecfield]}^2)
    +\vpt[\step]\norm{\vptinter[\svecfield]}^2
    \end{aligned}
\end{equation}
As $\vpt[\jvecfield][\play][1/2]=0$ and $\vpt[\step][\play][1]=\vpt[\step][\play][2]$, summing the above from $\run=1$ to $\nRuns$ gives
\begin{equation}
    \label{eq:OptDA+-adapt-adv-sumrounds}
    \begin{aligned}[b]
    \sum_{\run=1}^{\nRuns}
    \product{\vptinter[\svecfield]}{\vptinter[\state]-\vp[\arpoint]}
    &\le
    \frac{
    \norm{\vpt[\state][\play][1]
    -\vp[\arpoint]}^2}{2\vptupdate[\step][\play][\nRuns]}
    -
    \sum_{\run=1}^{\nRuns}
    \frac{\norm{\vpt[\state]-\vptupdate[\state]}^2}{2\vpt[\step]}
    +
    \sum_{\run=1}^{\nRuns}
    (\vpt[\stepalt]+\vpt[\step])
    \norm{\vptinter[\svecfield]}^2.
    \end{aligned}
\end{equation}
Dropping the negative term and taking expectation leads to
\begin{equation}
    \notag
    \begin{aligned}
    \ex\left[
    \sum_{\run=1}^{\nRuns}
    \product{\vp[\vecfield](\inter[\jstate])}
    {\vptinter[\state]-\vp[\arpoint]}
    \right]
    &\le 
    \ex
    \left[
    \frac{\radius^2
    }{2\vptupdate[\step][\play][\nRuns]}
    +
    \sum_{\run=1}^{\nRuns}
    (\vpt[\stepalt]+\vpt[\step])
    ((1+\noisecontrol)\norm{\vp[\vecfield](\inter[\jstate])}^2+\noisevar)
    \right]
    \\
    &
    \le
    \frac{\radius^2
    }{2\vptupdate[\step][\play][\nRuns]}
    +
    \sum_{\run=1}^{\nRuns}
    (\vpt[\stepalt]+\vpt[\step])
    ((1+\noisecontrol)\gbound^2+\noisevar)
    \end{aligned}
\end{equation}
The claim then follows immediately from the choice of the learning rates.
\end{proof}

%% file: appendices/regret-adaptive.tex
In this section, we tackle the regret analysis of OptDA+ run with adaptive learning rates.
For ease of notation, we introduce the following quantities\footnote{For $\run\le0$, $\vpt[\regpar]=\vpt[\regparalt]=0$.}
\begin{equation}
    \notag
    \vpt[\regpar]
    =
    \sum_{\runalt=1}^{\run}
    \norm{\vptinter[\svecfield][\play][\runalt]}^2,
    ~~~~
    \vpt[\regparalt]
    =
    \sum_{\runalt=1}^{\run}
    \norm{\vpt[\state][\play][\runalt]
    -\vptupdate[\state][\play][\runalt]}^2.
\end{equation}
Clearly, our adaptive learning rates \eqref{adaptive-lr} correspond to $\vpt[\stepalt]=1/(1+\vpt[\regpar][\play][\run-2])^{\frac{1}{2}-\exponent}$ and $\vpt[\step]=1/\sqrt{1+\vpt[\regpar][\play][\run-2]+\vpt[\regparalt][\play][\run-2]}$.
As \cref{asm:boundedness} assumes the noise to be bounded \emph{almost surely}, whenever this assumption is used, the stated inequalities only hold almost surely. To avoid repetition, we will not mention this explicitly in the following.
Finally, \cref{asm:stepsize-measurable} is obviously satisfied by the learning rates given by \eqref{adaptive-lr}; therefore,
\cref{lem:OptDA+-lin-regret-bound,lem:OptDA+-sum-bound} can be effectively applied.

\subsection{Preliminary Lemmas}
\label{apx:regret-adapt-prelim}

We start by establishing several basic lemmas that will be used repeatedly in the analysis.
We first state the apparent fact that $\vpt[\regpar]$ grows at most linearly under \cref{asm:boundedness}.

\begin{lemma}
\label{lem:regpar-bound}
Let \cref{asm:boundedness} hold.
Then, for all $\allplayers$ and $\nRuns\in\N$, we have
\begin{equation}
    \notag
    \vpt[\regpar][\play][\nRuns]
    \le
    2(\gbound^2+\noisebound^2)\nRuns.
\end{equation}
\end{lemma}
\begin{proof}
Using \cref{asm:boundedness}, we deduce that
\begin{equation}
    \notag
    \norm{\vptinter[\svecfield]}^2
    \le
    2\norm{\vp[\vecfield](\inter[\jstate])}^2
    + 2\norm{\vptinter[\noise]}^2
    \le 2\gbound^2+2\noisebound^2,
\end{equation}
The claimed inequality is then immediate from the definition of $\vpt[\regpar][\play][\nRuns]$.
\end{proof}

The next lemma is a slight generalization of the AdaGrad lemma \cite[Lemma 3.5]{ACG02}.

\begin{lemma}
\label{lem:adagrad-lemma}
Let $\nRuns\in\N$, $\varepsilon>0$, and $\exponent\in[0,1)$.
For any sequence of non-negative real numbers $\csta_1,\ldots,\csta_{\nRuns}$, it holds
\begin{equation}
    \label{eq:lem:adagrad}
    \sum_{\run=1}^{\nRuns}
    \frac{\vt[\csta]}{
    \left(
    \varepsilon
    +\sum_{\runalt=1}^{\run}\vt[\csta][\runalt]
    \right)^{\exponent}
    }
    \le
    \frac{1}{1-\exponent}
    \left(
    \sum_{\run=1}^{\nRuns}\vt[\csta]
    \right)^{1-\exponent}.
\end{equation}
\end{lemma}
\begin{proof}
The function $y\in\R_+\mapsto y^{1-\exponent}$ is concave and has derivative $y\mapsto (1-\exponent)/y^\exponent$.
Therefore, it holds for every $y,z>0$ that
\begin{equation}
    \notag
    z^{1-\exponent}
    \le
    y^{1-\exponent}
    +\frac{1-\exponent}{y^{\exponent}}(z-y).
\end{equation}
For $\alt{\varepsilon}\in(0,\varepsilon)$, we apply the above inequality to $y=\alt{\varepsilon}+\sum_{\runalt=1}^{\run}\vt[\csta][\runalt]$
and $z=\alt{\varepsilon}+\sum_{\runalt=1}^{\run-1}\vt[\csta][\runalt]$. This gives
\begin{equation}
    \label{eq:adagrad-lemma>=2}
    \begin{aligned}[b]
    \frac{1}{1-\exponent}
    \left(
    \alt{\varepsilon}
    +
    \sum_{\runalt=1}^{\run-1}\vt[\csta][\runalt]
    \right)^{1-\exponent}
    &\le
    \frac{1}{1-\exponent}
    \left(
    \alt{\varepsilon}
    +
    \sum_{\runalt=1}^{\run}\vt[\csta][\runalt]
    \right)^{1-\exponent}
    -
    \frac{\vt[\csta]}{
    \left(
    \alt{\varepsilon}+\sum_{\runalt=1}^{\run}\vt[\csta][\runalt]
    \right)^\exponent
    }
    \\
    &\le
    \frac{1}{1-\exponent}
    \left(
    \alt{\varepsilon}
    +
    \sum_{\runalt=1}^{\run}\vt[\csta][\runalt]
    \right)^{1-\exponent}
    -
    \frac{\vt[\csta]}{
    \left(
    \varepsilon+\sum_{\runalt=1}^{\run}\vt[\csta][\runalt]
    \right)^\exponent
    }.
    \end{aligned}
\end{equation}
Moreover, at $\run=1$ we have
\begin{equation}
    \label{eq:adagrad-lemma=1}
    \frac{\vt[\csta][1]}{(\varepsilon+\vt[\csta][1])^\exponent}
    \le
    (\alt{\varepsilon}+\vt[\csta][1])^{1-\exponent}
    \le
    \frac{1}{1-\exponent}
    (\alt{\varepsilon}+\vt[\csta][1])^{1-\exponent}.
\end{equation}
Summing \eqref{eq:adagrad-lemma>=2} from $\run=2$ to $\nRuns$, adding \eqref{eq:adagrad-lemma=1}, and rearranging leads to
\begin{equation}
    \notag
    \sum_{\run=1}^{\nRuns}
    \frac{\vt[\csta]}{
    \left(
    \varepsilon
    +\sum_{\runalt=1}^{\run}\vt[\csta][\runalt]
    \right)^{\exponent}
    }
    \le
    \frac{1}{1-\exponent}
    \left(
    \alt{\epsilon}
    +
    \sum_{\run=1}^{\nRuns}\vt[\csta]
    \right)^{1-\exponent}.
\end{equation}
Provided that the above inequality holds for any $\alt{\varepsilon}\in(0,\varepsilon)$, we obtain \eqref{eq:lem:adagrad} by taking $\alt{\varepsilon}\rightarrow0$.
\end{proof}

The above two lemmas together provide us with the following bound on the sum of the weighted squared norms of feedback.

\begin{lemma}
\label{lem:adagrad-lemma-with-lr}
Let \cref{asm:boundedness} hold, $\runalt\in\N_{0}$, and $\exponentalt\in[0,1)$.
Then, for all $\allplayers$ and $\nRuns\in\N$, we have
\begin{equation}
    \notag
    \sum_{\run=1}^{\nRuns}
    \frac{\norm{\vptinter[\svecfield]}^2}{(1+\vpt[\regpar][\play][\run-\runalt])^{\exponentalt}}
    \le
    \frac{(\vpt[\regpar][\play][\nRuns])^{1-\exponentalt}}{1-\exponentalt}
    +2\runalt(\gbound^2+\noisebound^2).
\end{equation}
\end{lemma}
\begin{proof}
Since $1/(1+\vpt[\regpar][\play][\run])^{\exponentalt}\le1/(1+\vpt[\regpar][\play][\run-\runalt])^\exponentalt$
and
$\norm{\vptinter[\svecfield]}^2
\le 2\gbound^2+2\noisebound^2$, we have
\begin{equation*}
    \left(
    \frac{1}{(1+\vpt[\regpar][\play][\run-\runalt])^{\exponentalt}}
    -
    \frac{1}{{(1+\vpt[\regpar][\play][\run])^{\exponentalt}}}
    \right)
    \norm{\vptinter[\svecfield]}^2
    \le
    \left(
    \frac{1}{(1+\vpt[\regpar][\play][\run-\runalt])^{\exponentalt}}
    -
    \frac{1}{{(1+\vpt[\regpar][\play][\run])^{\exponentalt}}}
    \right)
    2(\gbound^2+\noisebound^2).
\end{equation*}
Subsequently, it follows from 
\cref{lem:adagrad-lemma}
that
\begin{equation*}
    \begin{aligned}[b]
    \sum_{\run=1}^{\nRuns}
    \frac{\norm{\vptinter[\svecfield]}^2}{(1+\vpt[\regpar][\play][\run-\runalt])^{\exponentalt}}
    &=
    \sum_{\run=1}^{\nRuns}
    \left(
    \frac{\norm{\vptinter[\svecfield]}^2}
    {(1+\vpt[\regpar][\play][\run])^{\exponentalt}}
    +
    \left(
    \frac{1}{(1+\vpt[\regpar][\play][\run-\runalt])^{\exponentalt}}
    -
    \frac{1}{{(1+\vpt[\regpar][\play][\run])^{\exponentalt}}}
    \right)
    \norm{\vptinter[\svecfield]}^2
    \right)\\
    &\le
    \sum_{\run=1}^{\nRuns}
    \frac{\norm{\vptinter[\svecfield]}^2}
    {(1+\vpt[\regpar][\play][\run])^{\exponentalt}}
    +
    \sum_{\run=1}^{\nRuns}
    \left(
    \frac{1}{(1+\vpt[\regpar][\play][\run-\runalt])^{\exponentalt}}
    -
    \frac{1}{{(\vpt[\regpar][\play][\run])^{\exponentalt}}}
    \right)
    2(\gbound^2+\noisebound^2)\\
    &\le
    \frac{(\vpt[\regpar][\play][\nRuns])^{1-\exponentalt}}{1-\exponentalt}
    +\sum_{\run=-\runalt+1}^{0}
    \frac{2(\gbound^2+\noisebound^2)}{(1+\vpt[\regpar])^{\exponentalt}}\\
    &=
    \frac{(\vpt[\regpar][\play][\nRuns])^{1-\exponentalt}}{1-\exponentalt}
    +2\runalt(\gbound^2+\noisebound^2).
    \end{aligned}
\qedhere
\end{equation*}
\end{proof}

We also state a variant of the above result that takes into account the feedback of all players.

\begin{lemma}
\label{cor:adagrad-lemma-with-lr}
Let \cref{asm:boundedness} hold, $\runalt\in\N_{0}$, $\exponentalt\in[0,1)$,
and $\seqinf[\weights]$ be a sequence of non-negative $\nPlayers$-dimensional vectors such that $\vpt[\weight]\le1/(1+\vpt[\regpar][\play][\run-\runalt])^{\exponentalt}$.
Then, for all $\nRuns\in\N$, we have
\begin{equation}
    \notag
    \sum_{\run=1}^{\nRuns}
    \norm{\inter[\jsvecfield]}^2_{\vt[\weights]}
    \le
    2\nPlayers\runalt(\gbound^2+\noisebound^2)
    +\sumplayers
    \frac{(\vpt[\regpar][\play][\nRuns])^{1-\exponentalt}}{1-\exponentalt}.
\end{equation}
\end{lemma}
\begin{proof}
This is immediate from \cref{lem:adagrad-lemma-with-lr}.
\end{proof}

Both \cref{lem:adagrad-lemma-with-lr} and \cref{cor:adagrad-lemma-with-lr} are essential for our analysis as they allow us to express the sums appearing in our analysis as a power of $\vpt[\regpar]$ plus a constant.
We end up with a technical lemma for bounding the inverse of $\vpt[\step]$.


\begin{lemma}
\label{lem:adaptive-lr-cancel-out}
Let the learning rates be defined as in \eqref{adaptive-lr}.
For any $\allplayers$, $\nRuns\in\N$, and $\csta, \cstb \in \R_+$, we have
\begin{equation}
    \notag
    \frac{\csta}{\vptupdate[\step][\play][\nRuns]}
    -\cstb\sum_{\run=1}^{\nRuns}
    \frac{\norm{\vpt[\state]-\vptupdate[\state]}^2}
    {\vpt[\step]}
    \le
    \csta\sqrt{1+\vpt[\regpar][\play][\nRuns-1]}
    +\frac{\csta^2}{4\cstb}.
\end{equation}
\end{lemma}
\begin{proof}
On one hand, we have
%
\begin{equation}
    \notag
    \frac{\csta}{\vptupdate[\step][\play][\nRuns]}
    =
    \csta\sqrt{1+\vpt[\regpar][\play][\nRuns-1]
    +\vpt[\regparalt][\play][\nRuns-1]}
    \le
    \csta\sqrt{1+\vpt[\regpar][\play][\nRuns-1]}
    +\csta\sqrt{\vpt[\regparalt][\play][\nRuns-1]}.
\end{equation}
On the other hand, with $\vpt[\step]\le1$, it holds
\begin{equation}
    \notag
    \cstb\sum_{\run=1}^{\nRuns}
    \frac{\norm{\vpt[\state]-\vptupdate[\state]}^2}
    {\vpt[\step]}
    \ge
    \cstb\sum_{\run=1}^{\nRuns}
    \norm{\vpt[\state]-\vptupdate[\state]}^2
    \ge \cstb \vpt[\regparalt][\play][\nRuns-1].
\end{equation}
Let us define the function $\func\from\scalar\in\R\mapsto-\cstb\scalar^2+\csta\scalar$.
Then
\begin{equation}
    \notag
    \csta\sqrt{\vpt[\regparalt][\play][\nRuns-1]}
    -\cstb \vpt[\regparalt][\play][\nRuns-1]
    \le
    \max_{\scalar\in\R} \func(\scalar)
    =\frac{\csta^2}{4\cstb}.
\end{equation}
Combining the above inequalities gives the desired result.
\end{proof}


\subsection{Robustness Against Adversarial Opponents}

In this part, we derive regret bounds for adaptive OptDA+ when played against adversarial opponents.

\begin{proposition}
\label{prop:OptDA+-adapt-adversarial-apx}
Let \cref{asm:boundedness} hold and player $\play$ run \eqref{OptDA+} with learning rates \eqref{adaptive-lr}.
Then, for any bounded set $\vp[\cpt]$
with $\radius\ge\sup_{\vp[\arpoint]\in\vp[\cpt]}\norm{\vpt[\state][\play][1]-\vp[\arpoint]}$, it holds
\begin{align*}
    \notag
    &\max_{\vp[\arpoint]\in\vp[\cpt]}
    \ex\left[
    \sum_{\run=1}^{\nRuns}
    \product{\vp[\vecfield](\inter[\jstate])}
    {\vptinter[\state]-\vp[\arpoint]}
    \right]
    \\
    &~~~=
    \bigoh\left(
    ((\gbound^2+\noisebound^2)\nRuns)^{\frac{1}{2}+\exponent}
    +\radius^2(\gbound+\noisebound)\sqrt{\nRuns}
    +\radius^4+\gbound^2+\noisebound^2
    \right).
\end{align*}
\end{proposition}
\begin{proof}
To begin, we notice that inequality \eqref{eq:OptDA+-adapt-adv-sumrounds} that we established in the proof of \cref{prop:OptDA+-regret-adversarial-apx} still holds here for any $\vp[\arpoint]\in\vp[\cpt]$.
Furthermore, applying \cref{lem:adaptive-lr-cancel-out} with $\csta\subs\radius^2/2$, $\cstb\subs1/2$ leads to
\begin{equation}
    \notag
    \frac{\radius^2}{2\vptupdate[\step][\play][\nRuns]}
    -
    \sum_{\run=1}^{\nRuns}
    \frac{\norm{\vpt[\state]-\vptupdate[\state]}^2}{2\vpt[\step]}
    \le
    \frac{\radius^2\sqrt{1+\vpt[\regpar][\play][\nRuns-1]}
    }{2}
    +\frac{\radius^4}{8}.
\end{equation}
On the other hand, invoking \cref{lem:adagrad-lemma-with-lr}
with either $\exponentalt\subs1/4+\exponent$ or $\exponentalt\subs1/2$ guarantees that
\begin{equation}
    \notag
    \sum_{\run=1}^{\nRuns}
    (\vpt[\stepalt]+\vpt[\step])
    \norm{\vptinter[\svecfield]}^2
    \le
    \frac{4(\vpt[\regpar][\play][\nRuns])^{3/4-\exponent}}{3-4\exponent}
    +2\sqrt{\vpt[\regpar][\play][\nRuns]}
    +8(\gbound^2+\noisebound^2)
    .
\end{equation}
Putting the above inequalities together, we obtain
\begin{equation}
    \notag
    \begin{aligned}
    \max_{\vp[\arpoint]\in\vp[\cpt]}
    \ex\left[
    \sum_{\run=1}^{\nRuns}
    \product{\vp[\vecfield](\inter[\jstate])}
    {\vptinter[\state]-\vp[\arpoint]}
    \right]
    &\le 
    \ex
    \left[
    \frac{\radius^2\sqrt{1+\vpt[\regpar][\play][\nRuns-1]}
    }{2}
    +\frac{4(\vpt[\regpar][\play][\nRuns])^{3/4-\exponent}}{3-4\exponent}
    +2\sqrt{\vpt[\regpar][\play][\nRuns]}
    \right]
    \\
    &~~
    +\frac{\radius^4}{8}
    +8(\gbound^2+\noisebound^2).
    \end{aligned}
\end{equation}
We conclude with the help of \cref{lem:regpar-bound}.
\end{proof}


\input{appendices/regret-adaptive-part2}

%% file: appendices/regret-adaptive-part2.tex
\subsection{Smaller Regret Against Opponents with Same Learning Algorithm}

We now address the more challenging part of the analysis: fast regret minimization when all players adopt adaptive OptDA+.
For this, we need to control the different terms appearing in \cref{lem:OptDA+-lin-regret-bound,lem:OptDA+-sum-bound}.
For the latter we build the following lemma to control the sum of some differences.
As argued in \cref{sec:adaptive}, this is the reason that we include $\norm{\vpt[\state][\play][\runalt]-\vptupdate[\state][\play][\runalt]}^2$ in the definition of $\vpt[\step]$.

\begin{lemma}
\label{lem:OptDA+-adapt-sum-cancel-terms}
Let \cref{asm:lips,asm:boundedness} hold and the learning rates be defined as in \eqref{adaptive-lr}, then for all $\nRuns\in\N$, we have
\begin{equation}
    \notag
    \sum_{\run=1}^{\nRuns}
    \left(
    3\norm{
    \jvecfield(\vt[\jstate])-\jvecfield(\update[\jstate])
    }_{\vt[\jstepalt]}^2
    -\norm{\vt[\jstate]-\update[\jstate]}_{1/(4\vt[\jstep])}^2\right)
    \le
    432\nPlayers^3\lips^6 + 24\nPlayers^2\gbound^2.
\end{equation}
\end{lemma}
\begin{proof}
For all $\allplayers$,
let us define
\begin{equation}
    \notag
    \vp[\runsep]
    \defeq
    \max\left\{\runalt\in\intinterval{0}{\nRuns}:
    \vpt[\step]\ge\frac{1}{12\nPlayers\lips^2}\right\},
\end{equation}
where we set $\vpt[\step][\play][0]=1/(12\nPlayers\lips^2)$ to ensure that $\vp[\runsep]$ is always well-defined.
By the definition of $\vpt[\step]$,
the inequality $\vpt[\step][\play][\vp[\runsep]]\ge1/(12\nPlayers\lips^2)$ implies
$\vpt[\regparalt][\play][\vp[\runsep]-2]\le144\nPlayers^2\lips^2$.
We next define the sets
\begin{equation}
    \notag
    \runs
    \defeq
    \bigcup_{\allplayers}
    \{
    \vp[\runsep]-1,
    \vp[\runsep]
    \}
    \intersect
    \intinterval{1}{\nRuns}
\end{equation}
Clearly, $\card(\runs)\le2\nPlayers$.
With $\vt[\jstepalt]\le1$, \cref{asm:lips,asm:boundedness}, we obtain
\begin{equation}
    \label{eq:OptDA+-adapt-cancel-VVdiff}
    \begin{aligned}[b]
    &\hspace{-1em}
    \sum_{\run=1}^{\nRuns}
    3\norm{
    \jvecfield(\vt[\jstate])-\jvecfield(\update[\jstate])
    }_{\vt[\jstepalt]}^2
    \\
    &
    \le
    \sum_{\run=1}^{\nRuns}
    3\norm{
    \jvecfield(\vt[\jstate])-\jvecfield(\update[\jstate])
    }^2
    \\
    &=
    \sum_{\run\in\setexclude{\oneto{\nRuns}}{\runs}}
    3\norm{
    \jvecfield(\vt[\jstate])-\jvecfield(\update[\jstate])
    }^2
    +
    \sum_{\run\in\runs}
    3\norm{
    \jvecfield(\vt[\jstate])-\jvecfield(\update[\jstate])
    }^2
    \\
    &\le
    \sum_{\run\in\setexclude{\oneto{\nRuns}}{\runs}}
    3\nPlayers\lips^2\norm{
    \vt[\jstate]-\update[\jstate]}^2
    +
    \sum_{\run\in\runs}
    6
    \left(
    \norm{\jvecfield(\vt[\jstate])}^2
    +\norm{\jvecfield(\update[\jstate])}^2
    \right)
    \\
    &\le
    \sumplayers
    \sum_{\run\in\setexclude{\oneto{\nRuns}}{\runs}}
    3\nPlayers\lips^2\norm{\vpt[\state]-\vptupdate[\state]}^2
    +
    \sum_{\run\in\runs}12\nPlayers\gbound^2
    \\
    &\le
    \sumplayers
    3\nPlayers\lips^2
    \Bigg(
    \underbrace{\sum_{\run=1}^{\vp[\runsep]-2}
    \norm{\vpt[\state]-\vptupdate[\state]}^2}_{
    \vpt[\regparalt][\play][\vp[\runsep]-2]
    \le
    144\nPlayers^2\lips^2}
    +
    \sum_{\run=\vp[\runsep]+1}^{\nRuns}
    \norm{\vpt[\state]-\vptupdate[\state]}^2
    \Bigg)
    +24\nPlayers^2\gbound^2
    \end{aligned}
\end{equation}
On the other hand, by the choice of $\vp[\runsep]$ we know that $1/\vpt[\step]\ge12\nPlayers\lips^2$ for all $\run\ge\vp[\runsep]+1$; hence
\begin{equation}
    \label{eq:OptDA+-adapt-cancel-diff}
    \begin{aligned}[b]
    \sum_{\run=1}^{\nRuns}
    \norm{\vt[\jstate]-\update[\jstate]}_{1/(4\vt[\jstep])}^2
    &= \sumplayers
    \sum_{\run=1}^{\nRuns}
    \frac{\norm{\vpt[\state]-\vptupdate[\state]}^2}
    {4\vpt[\step]}
    \\
    &\ge
    \sumplayers
    \sum_{\run=\vp[\runsep]+1}^{\nRuns}
    \frac{\norm{\vpt[\state]-\vptupdate[\state]}^2}
    {4\vpt[\step]}
    \\
    &\ge
    \sumplayers
    \sum_{\run=\vp[\runsep]+1}^{\nRuns}
    3\nPlayers\lips^2
    \norm{\vpt[\state]-\vptupdate[\state]}^2.
    \end{aligned}
\end{equation}
Combining \eqref{eq:OptDA+-adapt-cancel-VVdiff} and \eqref{eq:OptDA+-adapt-cancel-diff} gives the desired result.
\end{proof}


With \cref{lem:OptDA+-adapt-sum-cancel-terms} and the lemmas introduced in \cref{apx:regret-adapt-prelim}, we are in a position to provide a bound on the expectation of the sum of the 
weighted squared operators norms
plus the second-order path length.
The next lemma is a fundamental building block for showing faster rates of adaptive OptDA+.

\begin{lemma}[Bound on sum of squared norms]
\label{lem:OptDA+-sum-bound-adapt}
Let \cref{asm:lips,asm:VS,asm:boundedness} hold and all players run OptDA+ with adaptive learning rates \eqref{adaptive-lr}.
Then, for all $\nRuns\in\N$
we have
\begin{equation}
\notag
\begin{aligned}[b]
    \sum_{\run=1}^{\nRuns}
    \ex[\norm{\jvecfield(\inter[\jstate])}_{\vt[\jstepalt]}^2]
    +
    \frac{1}{8}\sum_{\run=1}^{\nRuns}
    \ex[\norm{\vt[\jstate]-\update[\jstate]}^2]
    \le
    \cst_1
    \sumplayers
    \ex\left[
    \sqrt{\vpt[\regpar][\play][\nRuns]}
    \right]
    +\cst_2,
\end{aligned}
\end{equation}
where
\begin{equation}
    \notag
    \begin{aligned}
    \solnorm
    &=
    \min_{\jsol\in\sols}\max_{\allplayers}\norm{\vpt[\state][\play][1]-\vp[\sol]},\\
    \cst_1
    &= 12\nPlayers\lips^2
    +8\nPlayers\lips
    +2\lips+\solnorm^2+4,
    \\
    \cst_2
    &=
    432\nPlayers^3\lips^6 + 24\nPlayers^2\gbound^2
    +(12\nPlayers\lips^2
    +8\nPlayers\lips+2\lips+8)
    (\nPlayers\gbound^2+\nPlayers\noisebound^2)
    + \nPlayers\solnorm^2
    + 2\nPlayers\solnorm^4.
    \end{aligned}
\end{equation}
\end{lemma}
\begin{proof}
As in the proof of \cref{lem:OptDA+-sum-bound-nonadapt}, 
and proceed to bound in expectation the sum of the following quantities
\begin{align}
    \notag
    \vt[A]
    &= 3\norm{
    \jvecfield(\vt[\jstate])-\jvecfield(\update[\jstate])
    }_{\vt[\jstepalt]}^2
    -\norm{\vt[\jstate]-\update[\jstate]}_{1/(2\vt[\jstep])}^2,\\
    \vt[B]
    &= 6\onenorm{\vt[\jstepalt]}\lips^2\norm{\past[\jsvecfield]}_{
    \vt[\jstepalt]^2}^2
    +(4\nPlayers+1)\lips\norm{\past[\jnoise]}_{\vt[\jstepalt]^2}^2,~~~~
    \vt[C]
    = 2\norm{\inter[\jsvecfield]}_{\vt[\jstep]}^2.
    \label{eq:OptDA+-adapt-sum-bound-At}
\end{align}
Thanks to \cref{lem:OptDA+-adapt-sum-cancel-terms}, we know that the sum of $\vt[A]$ can be bounded directly without taking expectation by
\begin{equation}
    \notag
    \begin{aligned}
    \sum_{\run=1}^{\nRuns}\vt[A]
    &=
    \sum_{\run=1}^{\nRuns}
    \left(
    3\norm{
    \jvecfield(\vt[\jstate])-\jvecfield(\update[\jstate])
    }_{\vt[\jstepalt]}^2
    -\norm{\vt[\jstate]-\update[\jstate]}_{1/(4\vt[\jstep])}^2
    -\norm{\vt[\jstate]-\update[\jstate]}_{1/(4\vt[\jstep])}^2\right)\\
    &
    \le
    432\nPlayers^3\lips^6 + 24\nPlayers^2\gbound^2
    -
    \sum_{\run=1}^{\nRuns}
    \norm{\vt[\jstate]-\update[\jstate]}_{1/(4\vt[\jstep])}^2.
    \end{aligned}
\end{equation}
To obtain the above inequality we have also used $\vt[\jstep]\le1$.
To bound $\ex[\vt[B]]$, we use $\ex[\norm{\past[\jnoise]}_{\vt[\jstepalt]^2}^2]\le\ex[
\norm{\past[\jsvecfield]}_{\vt[\jstepalt]^2}^2]$ shown in \eqref{eq:bound-noise-by-Vhat}, 
$\onenorm{\vt[\jstepalt]}\le\nPlayers$,
and
%
\cref{cor:adagrad-lemma-with-lr}
(as $(\vptupdate[\stepalt])^2\le1/\sqrt{1+\vpt[\regpar][\play][\run-1]}$)
to obtain
\begin{equation}
    \label{eq:OptDA+-adapt-sum-bound-Bt}
    \begin{aligned}[b]
    \sum_{\run=2}^{\nRuns}
    \ex[\vt[B]]
    &\le
    \ex\left[
    \sum_{\run=2}^{\nRuns}
    (6\nPlayers\lips^2
    +(4\nPlayers+1)\lips)
    \norm{\past[\jsvecfield]}_{
    \vt[\jstepalt]^2}^2
    \right]    
    \\
    &=
    \ex\left[
    \sum_{\run=1}^{\nRuns-1}
    (6\nPlayers\lips^2
    +(4\nPlayers+1)\lips)
    \norm{\inter[\jsvecfield]}_{
    (\update[\jstepalt])^2}^2
    \right]
    \\
    &\le
    (6\nPlayers\lips^2
    +(4\nPlayers+1)\lips)
    \left
    (
    2\nPlayers(\gbound^2+\noisebound^2)
    +
    \sumplayers
    2
    \ex\left[
    \sqrt{\vpt[\regpar][\play][\nRuns-1]}
    \right]
    \right).
    \end{aligned}
\end{equation}
Similarly, the sum of $\vt[C]$ can be bounded in expectation by
\begin{equation}
    \label{eq:OptDA+-adapt-sum-bound-Ct}
    \sum_{\run=1}^{\nRuns}
    \ex[\vt[C]]
    \le
    8\nPlayers(\gbound^2+\noisebound^2)
    +
    \sumplayers
    4\ex
    \left[
    \sqrt{\vpt[\regpar][\play][\nRuns]}
    \right].
\end{equation}
Let us choose $\jsol=\argmin_{\jaction\in\sols}\max_{\allplayers}\norm{\vpt[\state][\play][1]-\vp[\action]}$ so that $\solnorm=
\max_{\play\in\players}\norm{\vpt[\state][\play][1]-\vp[\sol]}$.
Plugging \eqref{eq:OptDA+-adapt-sum-bound-At}, \eqref{eq:OptDA+-adapt-sum-bound-Bt}, and \eqref{eq:OptDA+-adapt-sum-bound-Ct} into \eqref{eq:lem:OptDA+-sum-bound} of \cref{lem:OptDA+-sum-bound}, we get readily
\begin{equation}
    \label{eq:OptDA+-adapt-sum-proof-main}
    \begin{aligned}[b]
    &\sum_{\run=2}^{\nRuns}
    \ex[\norm{\jvecfield(\inter[\jstate])}_{\vt[\jstepalt]}^2
    +\norm{\jvecfield(\past[\jstate])}_{\vt[\jstepalt]}^2]
    +
    \sum_{\run=1}^{\nRuns}
    \ex[\norm{\vt[\jstate]-\update[\jstate]}_{1/(8\vt[\jstep])}^2]\\
    &~~~
    \le
    \ex[\norm{\vt[\jstate][1]-\jsol}^2_{1/\update[\jstep][\nRuns]}]
    -
    \sum_{\run=1}^{\nRuns}
    \norm{\vt[\jstate]-\update[\jstate]}_{1/(8\vt[\jstep])}^2
    +
    (12\nPlayers\lips^2
    +8\nPlayers\lips+2\lips+4)
    \sumplayers
    \ex
    \left[
    \sqrt{\vpt[\regpar][\play][\nRuns]}
    \right]\\
    &~~~~~
    +
    432\nPlayers^3\lips^6 + 24\nPlayers^2\gbound^2
    +(12\nPlayers\lips^2
    +8\nPlayers\lips+2\lips+8)
    (\nPlayers\gbound^2+\nPlayers\noisebound^2)
    \end{aligned}
\end{equation}
Using \cref{lem:adaptive-lr-cancel-out}, we can then further bound the \ac{RHS} of \eqref{eq:OptDA+-adapt-sum-proof-main} with
\begin{equation}
    \label{eq:OptDA+-adapt-sum-proof-main-RHS}
    \begin{aligned}[b]
    \norm{\vt[\jstate][1]-\jsol}^2_{1/\update[\jstep][\nRuns]}
    -
    \sum_{\run=1}^{\nRuns}
    \norm{\vt[\jstate]-\update[\jstate]}_{1/(8\vt[\jstep])}^2
    &=
    \sumplayers
    \left
    (
    \frac{\norm{\vpt[\state][\play][1]-\vp[\sol]}^2}{\vptupdate[\step][\play][\nRuns]}
    -\sum_{\run=1}^{\nRuns}
    \frac{\norm{\vpt[\state]-\vptupdate[\state]}^2}
    {8\vpt[\step]}
    \right)
    \\
    &\le
    \sumplayers
    \left(
    \norm{\vpt[\state][\play][1]-\vp[\sol]}^2
    \sqrt{1+\vpt[\regpar][\play][\nRuns-1]}
    +2\norm{\vp[\sol]}^4
    \right)\\
    &\le
    \nPlayers\solnorm^2
    + 2\nPlayers\solnorm^4
    + \sumplayers\solnorm^2\sqrt{\vpt[\regpar][\play][\nRuns-1]}.
    \end{aligned}
\end{equation}
Finally, using $\vt[\jstep]\le1$, $\vt[\jstate][3/2]=\vt[\jstate][1]$, and $\vt[\jstepalt][2]=\vt[\jstepalt][1]$,
the \ac{LHS} of \eqref{eq:OptDA+-adapt-sum-proof-main} can be bounded from below by
\begin{equation}
\label{eq:OptDA+-adapt-sum-proof-main-LHS}
\begin{aligned}[b]
    &\sum_{\run=2}^{\nRuns}
    \ex[\norm{\jvecfield(\inter[\jstate])}_{\vt[\jstepalt]}^2
    +\norm{\jvecfield(\past[\jstate])}_{\vt[\jstepalt]}^2]
    +
    \sum_{\run=1}^{\nRuns}
    \ex[\norm{\vt[\jstate]-\update[\jstate]}_{1/(8\vt[\jstep])}^2]\\
    &~~~\ge
    \sum_{\run=1}^{\nRuns}
    \ex[\norm{\jvecfield(\inter[\jstate])}_{\vt[\jstepalt]}^2]
    +
    \frac{1}{8}\sum_{\run=1}^{\nRuns}
    \ex[\norm{\vt[\jstate]-\update[\jstate]}^2].
\end{aligned}
\end{equation}
Combining
\eqref{eq:OptDA+-adapt-sum-proof-main},
\eqref{eq:OptDA+-adapt-sum-proof-main-RHS},
and
\eqref{eq:OptDA+-adapt-sum-proof-main-LHS}
gives the desired result.
\end{proof}


We also refine \cref{lem:OptDA+-lin-regret-bound} for the case of adaptive learning rates.
The next lemma suggests the terms that need to be bounded in expectation in order to control the regret.

\begin{lemma}[Bound on linearized regret]
\label{lem:OptDA+-lin-regret-bound-adapt}
Let \cref{asm:lips,asm:VS,asm:boundedness} hold and all players run OptDA+ with adaptive learning rates \eqref{adaptive-lr}.
Then, for all $\allplayers$, $\nRuns\in\N$, and bounded set $\vp[\cpt]\subset\vp[\points]$ with $\radius\ge\sup_{\vp[\arpoint]\in\vp[\cpt]}\norm{\vpt[\state][\play][1]-\vp[\arpoint]}$, it holds that
\begin{equation}
    \notag
    \label{eq:lem:OptDA+-linregret-bound-adapt}
    \begin{aligned}
    \max_{\vp[\arpoint]\in\vp[\cpt]}
    \ex\left[
    \sum_{\run=1}^{\nRuns}
    \product{
    \vp[\vecfield](\inter[\jstate])
    }
    {\vptinter[\state]-\vp[\arpoint]}
    \right]
    \le
    \ex\Bigg[&
    \left(
    \frac{\radius^2}{2}
    +\frac{\lips+1}{2}
    \right)
    \sqrt{\vpt[\regpar][\play][\nRuns]}
    +(6\lips^2+4\lips)
    \sumplayers[\playalt]
    \sqrt{\vpt[\regpar][\playalt][\nRuns-1]}
    \\
    &
    +
    \frac{\radius^2\sqrt{\vpt[\regparalt][\play][\nRuns-1]}}{2}
    +
    \frac{3\lips^2}{2}
    \sum_{\run=1}^{\nRuns-1}
    \norm{\vt[\jstate]-\update[\jstate]}^2
    \\
    &
    +\frac{\radius^2}{2}
    +(6\nPlayers\lips^2
    +4\nPlayers\lips+\lips+2)
    (\gbound^2+\noisebound^2)
    \Bigg].
    \end{aligned}
\end{equation}
\end{lemma}
\begin{proof}
We will derive inequality \eqref{eq:lem:OptDA+-linregret-bound-adapt} from \cref{lem:OptDA+-lin-regret-bound}.
To begin, by \cref{asm:noises}\ref{asm:noises-unbiased} the noises are conditionally unbiased and we can thus write
\begin{equation}
    \notag
    \ex_{\run-1}[
    \norm{\vptpast[\svecfield]}^2
    ]
    =
    \norm{\vp[\vecfield](\vptpast[\state])}^2
    +
    \ex_{\run-1}[\norm{\vptpast[\noise]}^2]
    \ge
    \ex_{\run-1}[\norm{\vptpast[\noise]}^2].
\end{equation}
Subsequently, $\vt[\jstepalt]$ being $\last[\filter]$-measurable, applying the law of total expectation gives
\begin{equation}
    \label{eq:bound-noise-by-Vhat}
    \begin{aligned}
    \ex[\norm{\past[\jnoise]}_{\vt[\jstepalt]^2}^2]
    &=
    \ex\left[
    \sumplayers
    (\vpt[\stepalt])^2
    \ex_{\run-1}[\norm{\vptpast[\noise]}^2]
    \right]
    \\
    &\le
    \ex\left[
    \sumplayers
    (\vpt[\stepalt])^2
    \ex_{\run-1}[\norm{\vptpast[\svecfield]}^2]
    \right]
    =
    \ex[\norm{\past[\jsvecfield]}_{\vt[\jstepalt]^2}^2].
    \end{aligned}
\end{equation}
Plugging the above two inequalities into the inequality of \cref{lem:OptDA+-lin-regret-bound} and using $\vpt[\stepalt]\le1$ results in
%
\begin{equation*}
    \begin{aligned}
    &\hspace{-1em}\max_{\vp[\arpoint]\in\vp[\cpt]}
    \ex\left[
    \sum_{\run=1}^{\nRuns}
    \product{
    \vp[\vecfield](\inter[\jstate])
    }
    {\vptinter[\state]-\vp[\arpoint]}
    \right]
    \\
    &\begin{aligned}
    \le
    \ex\Bigg[&
    \frac{
    \radius^2}{2\vptupdate[\step][\play][\nRuns]}
    +
    \sum_{\run=2}^{\nRuns}
    \vpt[\stepalt]\lips^2
    \left(
    3\norm{\past[\jsvecfield]}^2_{\vt[\jstepalt]^2}
    +\frac{3}{2}\norm{\vt[\jstate]-\last[\jstate]}^2
    \right)
    \\
    &+
    \frac{1}{2}\sum_{\run=2}^{\nRuns}
    ((\vpt[\stepalt])^2\lips\norm{\vptpast[\svecfield]}^2
    +4\lips
    \norm{\past[\jsvecfield]}_{\vt[\jstepalt]^2}^2)
    +
    \frac{1}{2}\sum_{\run=1}^{\nRuns}
    \vpt[\step]\norm{\vptinter[\svecfield]}^2
    \Bigg]
    \end{aligned}
    \\
    &\begin{aligned}
    \le
    \ex\Bigg[&
    \frac{\radius^2
    \sqrt{1
    +\vpt[\regpar][\play][\nRuns-1]
    +\vpt[\regparalt][\play][\nRuns-1]}
    }{2}
    +\sum_{\run=1}^{\nRuns-1}
    (3\lips^2+2\lips)
    \norm{\inter[\jsvecfield]}^2_{(\update[\jstepalt])^2}
    +\frac{3\lips^2}{2}
    \sum_{\run=1}^{\nRuns-1}
    \norm{\vt[\jstate]-\update[\jstate]}^2
    \\
    &
    +\frac{1}{2}\sum_{\run=1}^{\nRuns-1}
    (\vptupdate[\stepalt])^2\lips\norm{\vptinter[\svecfield]}^2
    +\frac{1}{2}\sum_{\run=1}^{\nRuns}
    \vpt[\step]\norm{\vptinter[\svecfield]}^2
    \Bigg].
    \end{aligned}
    \end{aligned}
\end{equation*}
Since we have both $(\vptupdate[\stepalt])^2\le1/\sqrt{1+\vpt[\regpar][\play][\run-1]}$
and $\vpt[\step]\le1/\sqrt{1+\vpt[\regpar][\play][\run-2]}$, applying \cref{lem:adagrad-lemma-with-lr} leads to
\begin{equation}
    \notag
    \sum_{\run=1}^{\nRuns-1}
    (\vptupdate[\stepalt])^2\lips\norm{\vptinter[\svecfield]}^2
    +\sum_{\run=1}^{\nRuns}
    \vpt[\step]\norm{\vptinter[\svecfield]}^2
    \le
    \lips
    \left(
    \sqrt{\vpt[\regpar][\play][\nRuns-1]}
    +2(\gbound^2+\noisebound^2)
    \right)
    +
    \sqrt{\vpt[\regpar][\play][\nRuns]}
    +4(\gbound^2+\noisebound^2).
\end{equation}
Similarly, using \cref{cor:adagrad-lemma-with-lr} we deduce
\begin{equation}
    \notag
    \sum_{\run=1}^{\nRuns-1}
    (3\lips^2+2\lips)
    \norm{\inter[\jsvecfield]}^2_{(\update[\jstepalt])^2}
    \le
    (3\lips^2+2\lips)
    \left(
    2\nPlayers(\gbound^2+\noisebound^2)
    +
    \sumplayers[\playalt]
    2\sqrt{\vpt[\regpar][\playalt][\nRuns-1]}
    \right)
\end{equation}
Putting the above inequalities together and using
$\sqrt{1
+\vpt[\regpar][\play][\nRuns-1]
+\vpt[\regparalt][\play][\nRuns-1]}
\le 1 +\sqrt{\vpt[\regpar][\play][\nRuns-1]} + \sqrt{\vpt[\regparalt][\play][\nRuns-1]}$
gives the desired result.
\end{proof}


\subsubsection{The Case of Additive Noise}

From \cref{lem:OptDA+-sum-bound-adapt} and \cref{lem:OptDA+-lin-regret-bound-adapt} we can readily derive our main results for the case of additive noise.

\begin{theorem}
\label{thm:OptDA+-adapt-abs-noise-sumV2}
Let \cref{asm:lips,asm:VS,asm:boundedness} hold and all players run OptDA+ with adaptive learning rates \eqref{adaptive-lr}.
Then,
\begin{equation*}
    \sum_{\run=1}^{\nRuns}
    \ex[\norm{\jvecfield(\inter[\jstate])]
    }^2
    =\bigoh\left(\nRuns^{1-\exponent}\right).
\end{equation*}
\end{theorem}
\begin{proof}
With \cref{lem:regpar-bound}, for $\run\in\intinterval{1}{\nRuns}$, we can lower bound the learning rate $\vpt[\stepalt]$ by
\begin{equation}
    \notag
    \vpt[\stepalt]
    =
    \frac{1}{(1+\vpt[\regpar][\play][\run-2])^{\frac{1}{2}-\exponent}}
    \ge
    \frac{1}{
    (1+2\max(\run-2,0)(\gbound^2+\noisebound^2))^{\frac{1}{2}-\exponent}
    }
    \ge
    \frac{1}{
    (1+2\nRuns(\gbound^2+\noisebound^2))^{\frac{1}{2}-\exponent}
    }.
\end{equation}
\cref{lem:OptDA+-sum-bound-adapt} thus guarantees
\begin{equation*}
    \frac{
    \sum_{\run=1}^{\nRuns}
    \ex[\norm{\jvecfield(\inter[\jstate])
    }^2]
    }{
    (1+2\nRuns(\gbound^2+\noisebound^2))^{\frac{1}{2}-\exponent}
    }
    \le
    \cst_1
    \sumplayers
    \ex\left[
    \sqrt{\vpt[\regpar]}
    \right]
    +\cst_2.
\end{equation*}
We conclude by using again \cref{lem:regpar-bound}.
\end{proof}

\begin{theorem}
\label{thm:OptDA+-adapt-absnoise-linregret}
Let \cref{asm:lips,asm:VS,asm:boundedness} hold and all players run OptDA+ with adaptive learning rates \eqref{adaptive-lr}.
Then, for any $\allplayers$ and bounded set $\vp[\cpt]\subset\vp[\points]$, we have
\begin{equation*}
    \max_{\vp[\arpoint]\in\vp[\cpt]}
    \ex\left[
    \sum_{\run=1}^{\nRuns}
    \product{
    \vp[\vecfield](\inter[\jstate])
    }
    {\vptinter[\state]-\vp[\arpoint]}
    \right]
    =\bigoh\left(
    \sqrt{\nRuns}
    \right).
\end{equation*}
\end{theorem}
\begin{proof}
This follows from \cref{lem:OptDA+-lin-regret-bound-adapt}.
To begin, with \cref{lem:regpar-bound}, we have clearly
\begin{equation*}
    \ex\Bigg[
    \left(
    \frac{\radius^2}{2}
    +\frac{\lips+1}{2}
    \right)
    \sqrt{\vpt[\regpar]}
    +(6\lips^2+4\lips)
    \sumplayers[\playalt]
    \sqrt{\vpt[\regpar][\playalt][\nRuns-1]}
    \Bigg]
    =\bigoh
    \left(
    \sqrt{\nRuns}
    \right).
\end{equation*}
Next, thanks to \cref{lem:OptDA+-sum-bound-adapt} we can bound
\begin{equation}
    \notag
    \begin{aligned}
    \ex\left[
    \frac{\radius^2\sqrt{\vpt[\regparalt][\play][\nRuns-1]}}{2}
    +
    \frac{3\lips^2}{2}
    \sum_{\run=1}^{\nRuns-1}
    \norm{\vt[\jstate]-\update[\jstate]}^2
    \right]
    &\le
    \ex\left[
    \left(
    \frac{\radius^2}{2}
    +\frac{3\lips^2}{2}
    \right)
    \sum_{\run=1}^{\nRuns-1}
    \norm{\vt[\jstate]-\update[\jstate]}^2
    \right]\\
    &\le
    (4\radius^2
    +12\lips^2)
    \left(
    \cst_1
    \sumplayers
    \ex\left[
    \sqrt{\vpt[\regpar]}
    \right]
    +\cst_2
    \right).
    \end{aligned}
\end{equation}
This is again in $\bigoh(\sqrt{\nRuns})$.
Plugging the above into \cref{lem:OptDA+-lin-regret-bound-adapt} concludes the proof.
\end{proof}


\subsubsection{The Case of Multiplicative Noise}

The case of multiplicative noise is more delicate. As explained in \cref{sec:adaptive}, the main step is to establish an inequality in the form of \eqref{eq:adapt-main-ineq}.
This is achieved in \cref{lem:OptDA+-adapt-relnoise-finite} by using \cref{lem:OptDA+-sum-bound-adapt}.
Before that, we derive a lemma to show how inequality \eqref{eq:adapt-main-ineq} implies boundedness of the relevant quantities.

\begin{lemma}
\label{lem:different-powers}
Let $p,r,\cst\in\R_+$ such that $p>r$, $\cst\in\R_+$, and $(\vp[\csta][1],\ldots,\vp[\csta][\nPlayers])$ be a collection of $\nPlayers$ non-negative real-valued random variables. If
\begin{equation}
    \label{eq:ineq-different-powers}
    \sumplayers \ex[(\vp[\csta])^p]
    \le
    \cst
    \sumplayers \ex[(\vp[\csta])^r],
\end{equation}
Then
$\sumplayers \ex[(\vp[\csta])^p]
\le \nPlayers\cst^{\frac{p}{p-r}}$
and
$\sumplayers \ex[(\vp[\csta])^r]
\le \nPlayers\cst^{\frac{r}{p-r}}$
.
\end{lemma}
\begin{proof}
Since $p>r$, the function $\scalar\in\R_+\union\{0\}\mapsto\scalar^{\frac{r}{p}}$ is concave.
Applying Jensen's inequality for the expectation gives
$\ex[(\vp[\csta])^r]
\le\ex[(\vp[\csta])^p]^{\frac{r}{p}}
$.
Next, we apply Jensen's inequality for the average to obtain
\begin{equation}
    \frac{1}{\nPlayers}
    \sumplayers
    \ex[(\vp[\csta])^p]^{\frac{r}{p}}
    \le
    \left(
    \frac{1}{\nPlayers}
    \sumplayers\ex[(\vp[\csta])^p]
    \right)^{\frac{r}{p}}.
\end{equation}
Along with inequality \eqref{eq:ineq-different-powers} we then get
\begin{equation}
    \label{eq:diff-powers-middle}
    \sumplayers \ex[(\vp[\csta])^p]
    \le
    \cst
    \sumplayers \ex[(\vp[\csta])^r]
    \le
    \cst
    \nPlayers^{1-\frac{r}{p}}
    \left(
    \sumplayers\ex[(\vp[\csta])^p]
    \right)^{\frac{r}{p}}.
\end{equation}
In other words
\begin{equation*}
    \left(
    \sumplayers\ex[(\vp[\csta])^p]
    \right)^{1-\frac{r}{p}}
    \le
    \cst
    \nPlayers^{1-\frac{r}{p}}.
\end{equation*}
Taking both sides of the inequality to the power of $p/(p-r)$, we obtain effectively
\begin{equation*}
    \sumplayers \ex[(\vp[\csta])^p]
    \le \nPlayers\cst^{\frac{p}{p-r}}
\end{equation*}
The second inequality combines the above with second part of \eqref{eq:diff-powers-middle}.
\end{proof}

In the next lemma we build inequality \eqref{eq:adapt-main-ineq}, and combined with \cref{lem:different-powers} we obtain the boundedness of various quantities. This is also where the factor $1/\exponent$ shows up.

\begin{lemma}
\label{lem:OptDA+-adapt-relnoise-finite}
Let \cref{asm:lips,asm:VS,asm:boundedness} hold and all players run OptDA+ with adaptive learning rates \eqref{adaptive-lr}.
Assume additionally \cref{asm:noises} with $ \noisedev=0$.
Then, for any $\nRuns\in\N$, we have
\begin{gather}
    \label{eq:bound-regpar-3/4}
    \sumplayers
    \ex\left[
    (1+\vpt[\regpar][\play][\nRuns])^{\frac{1}{2}+\exponent}
    \right]
    \le
    \nPlayers
    \left(
    (1+\noisecontrol)\cst_1
    +
    \frac{(1+\noisecontrol)\cst_2+1}{\nPlayers}
    \right)^{1+\frac{1}{2\exponent}},
    \\
    \label{eq:bound-regpar-1/2}
    \sumplayers
    \ex\left[
    \sqrt{1+\vpt[\regpar][\play][\nRuns]}
    \right]
    \le
    \nPlayers
    \left(
    (1+\noisecontrol)\cst_1
    +
    \frac{(1+\noisecontrol)\cst_2+1}{\nPlayers}
    \right)^{\frac{1}{2\exponent}},
    \\
    \label{eq:bound-regparalt}
    \sumplayers
    \ex[\vpt[\regparalt][\play][\nRuns]]
    \le
    8\nPlayers\cst_1
    \left(
    (1+\noisecontrol)\cst_1
    +
    \frac{(1+\noisecontrol)\cst_2+1}{\nPlayers}
    \right)^{\frac{1}{2\exponent}}
    +8\cst_2.
\end{gather}
\end{lemma}
\begin{proof}
From \cref{lem:OptDA+-sum-bound-adapt} we know that
\begin{equation*}
    \sum_{\run=1}^{\nRuns}
    \ex[\norm{\jvecfield(\inter[\jstate])}_{\vt[\jstepalt]}^2]
    \le
    \cst_1
    \sumplayers
    \ex\left[
    \sqrt{\vpt[\regpar][\play][\nRuns]}
    \right]
    +\cst_2,
\end{equation*}
Since $\vt[\jstepalt]$ is $\vt[\filter]$-measurable (it is even $\last[\filter]$-measurable), using the relative noise assumption and the law of total expectation we get 
\begin{align*}
    \ex[\norm{\jvecfield(\inter[\jstate])}_{\vt[\jstepalt]}^2]
    =\sumplayers
    \ex[\vpt[\stepalt]
    \ex_{\run}[
    \norm{\vp[\vecfield](\inter[\jstate])]}^2]
    \ge
    \sumplayers
    \ex\left[
    \vpt[\stepalt]
    \ex_{\run}
    \left[
    \frac{\norm{\vptinter[\svecfield]}^2
    }{1+\noisecontrol}
    \right]
    \right]
    =\frac{
    \norm{\inter[\jsvecfield]}^2_{\vt[\jstepalt]}
    }{1+\noisecontrol}.
\end{align*}
The learning rates $\vt[\jstepalt]$ being non-increasing, we can then bound from below the sum of $\ex[\norm{\jvecfield(\inter[\jstate])}_{\vt[\jstepalt]}^2]$ by
\begin{equation}
    \notag
    \begin{aligned}
    \sum_{\run=1}^{\nRuns}
    \ex[\norm{\jvecfield(\inter[\jstate])}_{\vt[\jstepalt]}^2]
    &\ge
    \frac{1}{1+\noisecontrol}
    \sum_{\run=1}^{\nRuns}
    \ex[\norm{\inter[\jsvecfield]}_{\vt[\jstepalt]}^2]
    \\
    &\ge
    \frac{1}{1+\noisecontrol}
    \sum_{\run=1}^{\nRuns}
    \ex[\norm{\inter[\jsvecfield]}_{\vt[\jstepalt][\nRuns+2]}^2]
    \\
    &=
    \frac{1}{1+\noisecontrol}
    \sumplayers
    \ex\left[
    \frac{\sum_{\run=1}^{\nRuns}
    \norm{\vptinter[\svecfield]}^2
    }
    {(1+\vpt[\regpar][\play][\nRuns])^{\frac{1}{2}-\exponent}}
    \right]
    \\
    &=
    \frac{1}{1+\noisecontrol}
    \sumplayers
    \ex\left[
    \frac{\vpt[\regpar][\play][\nRuns]+1-1}
    {(1+\vpt[\regpar][\play][\nRuns])^{\frac{1}{2}-\exponent}}
    \right]\\
    &\ge
    -\frac{1}{1+\noisecontrol}+
    \frac{1}{1+\noisecontrol}
    \sumplayers
    \ex\left[
    (1+\vpt[\regpar][\play][\nRuns])^{\frac{1}{2}+\exponent}
    \right].
    \end{aligned}
\end{equation}
As a consequence, we have shown that
\begin{equation}
    \notag
    \sumplayers
    \ex\left[
    (1+\vpt[\regpar][\play][\nRuns])^{\frac{1}{2}+\exponent}
    \right]
    \le
    (1+\noisecontrol)\cst_1
    \sumplayers
    \ex\left[
    \sqrt{\vpt[\regpar][\play][\nRuns]}
    \right]
    +(1+\noisecontrol)\cst_2+1,
\end{equation}
Subsequently,
\begin{equation}
    \notag
    \sumplayers
    \ex\left[
    (1+\vpt[\regpar][\play][\nRuns])^{\frac{1}{2}+\exponent}
    \right]
    \le
    \left(
    (1+\noisecontrol)\cst_1
    +
    \frac{(1+\noisecontrol)\cst_2+1}{\nPlayers}
    \right)
    \sumplayers
    \ex\left[
    \sqrt{1+\vpt[\regpar][\play][\nRuns]}
    \right].
\end{equation}
We deduce \eqref{eq:bound-regpar-3/4} and \eqref{eq:bound-regpar-1/2} with the help of \cref{lem:different-powers} taking
$p\subs 3/4-\exponent$,
$r\subs 1/2$,
$\cst \subs (1+\noisecontrol)\cst_1
+((1+\noisecontrol)\cst_2+1)/\nPlayers$, and
$\vp[\csta]\subs1+\vpt[\regpar][\play][\nRuns]$.
Plugging \eqref{eq:bound-regpar-1/2} into \cref{lem:OptDA+-sum-bound-adapt} gives
\eqref{eq:bound-regparalt}.
\end{proof}

Now, as an immediate consequence of all our previous results, we obtain the constant regret bound of adaptive OptDA+ under multiplicative noise.

\begin{theorem}
\label{thm:OptDA+-adapt-relnoise-linregret}
Let \cref{asm:lips,asm:VS,asm:boundedness} hold and all players run OptDA+ with adaptive learning rates \eqref{adaptive-lr}.
Assume additionally \cref{asm:noises} with $ \noisedev=0$.
Then, for any $\allplayers$ and bounded set $\vp[\cpt]$, we have
\begin{equation*}
    \ex\left[
    \max_{\vp[\arpoint]\in\vp[\cpt]}\sum_{\run=1}^{\nRuns}
    \product{
    \vp[\vecfield](\inter[\jstate])
    }
    {\vptinter[\state]-\vp[\arpoint]}
    \right]
    =\bigoh
    \left(
    \exp\left(\frac{1}{2\exponent}
    \right)
    \right).
\end{equation*}
\end{theorem}
\begin{proof}
Combining \cref{lem:OptDA+-lin-regret-bound-adapt} and \cref{lem:OptDA+-adapt-relnoise-finite} gives the desired result.
\end{proof}

%% file: appendices/convergence.tex
We close our appendix with proofs on almost-sure last-iterate convergence of the trajectories.
The global proof schema was sketched in \cref{sec:trajectory} for the particular case of \cref{thm:cvg-relative-noise} (which corresponds to the upcoming \cref{thm:cvg-OptDA+}).
To prove last-iterate convergence we make heavy use of the different results that we derived in previous sections.

\input{appendices/aux-lemmas}

\subsection{Trajectory Convergence of OG+ under Additive Noise}

We start by proving the almost sure last-iterate convergence of \ac{OG+} under additive noise.
This proof is, in a sense, the most technical once the results of the previous sections are established.
This is because with vanishing learning rates, we cannot show that every cluster point of $(\inter[\jstate])_{\run\in\N}$ is a Nash equilibrium with probability $1$.
Instead we need to work with subsequences.

We will prove the convergence of $\vt[\jstate]$ and $\inter[\jstate]$ separately, and under relaxed learning rate requirements.
For the convergence of $\vt[\jstate]$, a learning rate condition introduced in \cite{HIMM20} for double step-size \ac{EG} is considered.

\begin{theorem}
\label{thm:OG+-as-convergence}
Let \cref{asm:noises,asm:lips,asm:VS} hold and all players run \eqref{OG+} with non-increasing learning rate sequences $(\vt[\stepalt])_{\run\in\N}$ and $(\vt[\step])_{\run\in\N}$ satisfying \eqref{OG+-lr} and
\begin{equation}
    \label{OG+-lr-as}
    \sumtoinf \vt[\stepalt]\update[\step] = \infty,
    ~~
    \sumtoinf \vt[\stepalt]^2\update[\step] < \infty,
    ~~
    \sumtoinf \vt[\step]^2 < \infty.
\end{equation}
Then, $\vt[\jstate]$ converges almost surely to a Nash equilibrium.
\end{theorem}
\begin{proof}
Our proof is divided into four steps.
To begin, let us define $\vt[\surr{\jstate}][1]=\vt[\jstate][1]$ and for all $\run\ge2$, 
\begin{equation*}
\vt[\surr{\jstate}]=\vt[\jstate]+\vt[\step]\past[\jnoise]=\last[\jstate]-\vt[\step]\jvecfield(\past[\jstate])
\end{equation*}
Notice that $\vt[\surr{\jstate}]$ is $\last[\filter]$-measurable.
This surrogate of $\vt[\jstate]$ plays an important role in the subsequent analysis.

\vspace{0.4em}
\noindent
(1) \textit{With probability $1$, $\norm{\vt[\surr{\jstate}]-\jsol}$ converges for all $\jsol\in\sols$.}
\enspace
Let $\sol\in\sols$.
We would like to apply Robbins-Siegmund's theorem (\cref{lem:Robbins-Siegmund}) to the inequality of \cref{lem:OG+-quasi-descent}
with
\begin{align*}
    &\vt[\filteralt]\subs\last[\filter],
    ~~~
    \vt[\srv]\subs\ex_{\run-1}[\norm{\vt[\jstate]-\jsol}^2],
    ~~~
    \vt[\srvmul]\subs0,\\
    &\vt[\srvmi]
    \subs\vt[\stepalt]\update[\step]
    (\ex_{\run-1}[\norm{\jvecfield(\inter[\jstate])}^2]
    +\norm{\jvecfield(\past[\jstate])}^2),\\
    &\vt[\srvp]
    \subs\ex_{\run-1}
    [
    \begin{aligned}[t]
    &3\vt[\stepalt]\update[\step]\nPlayers\lips^2(
    (\vt[\step]^2+\vt[\stepalt]^2)\norm{\past[\jsvecfield]}^2
    +(\last[\stepalt])^2\norm{\ancient[\jsvecfield]}^2)\\
    &+(\vt[\stepalt]^2\update[\step]+\nPlayers\update[\step](\vt[\step]+\vt[\stepalt])^2)\lips\norm{\past[\jnoise]}^2
    +(\update[\step])^2\norm{\inter[\jsvecfield]}^2].
    \end{aligned}
\end{align*}
As \cref{lem:OG+-quasi-descent} only applies to $\run\ge2$, for $\run=1$ we use inequality \eqref{eq:OG+-quasi-descent-1}.
We thus choose $\vt[\srvmi][1]=0$ and $\vt[\srvp][1]=\vt[\step][2]^2\norm{\vt[\jsvecfield][3/2]}^2$.

We claim that $\sumtoinf\ex[\vt[\srvp]]<+\infty$.
In fact, following the proof of \cref{lem:OG+-sum-bound}, we can deduce
\begin{equation}
    \notag
    \begin{aligned}
    \sumtoinf\ex[\vt[\srvp]]
    \le 
    \underbrace{\sumtoinf 
    \vt[\stepalt]\update[\step](
    \vt[\csta](1+\noisecontrol)
    +\vt[\cstb]\noisecontrol)
    \ex[\norm{\jvecfield(\inter[\jstate])}^2]}_{(A)}
    +\underbrace{\sumtoinf
    \vt[\stepalt]\update[\step](\vt[\csta]+\vt[\cstb])\nPlayers\noisevar}_{(B)},
    \end{aligned}
\end{equation}
for $\vt[\csta]=\update[\step]/\vt[\stepalt]+9\vt[\stepalt]^2\nPlayers\lips^2$ and $\vt[\cstb]=\vt[\stepalt](4\nPlayers+1)\lips$.
With our learning rate requirements it is true that $ \vt[\csta](1+\noisecontrol)
    +\vt[\cstb]\noisecontrol\le3/2$, so \cref{lem:OG+-sum-bound} implies (A) is finite.
On the other hand, from $\sumtoinf\vt[\stepalt]^2\update[\step]<+\infty$ and $\sumtoinf\vt[\step]^2<+\infty$ we deduce that (B) is also finite.
We then conclude that it is effectively true that $\sumtoinf\ex[\vt[\srvp]]<+\infty$.

As a consequence, applying Robbins-Siegmund theorem gives the almost sure convergence of $\ex_{\run-1}[\norm{\vt[\jstate]-\jsol}^2]$ to a finite random variable $\limp{\srv}$.
To proceed, we use the equality
\begin{equation}
    \notag
    \ex_{\run-1}[\norm{\vt[\jstate]-\jsol}^2]
    = \ex_{\run-1}[\norm{\vt[\surr{\jstate}]
    -\vt[\step]\past[\jnoise]-\jsol}^2]
    = \norm{\vt[\surr{\jstate}]-\jsol}^2
    + \vt[\step]^2\ex_{\run-1}[\norm{\past[\jnoise]}^2].
\end{equation}
Accordingly,
\begin{equation}
    \label{eq:OG+-sum-diff-bounded}
    \begin{aligned}[b]
    \sumtoinf[2]
    \ex[\ex_{\run-1}[\norm{\vt[\jstate]-\jsol}^2]-\norm{\vt[\surr{\jstate}]-\jsol}^2]
    &= \sumtoinf[2]
    \ex[\vt[\step]^2
    \ex_{\run-1}[\norm{\past[\jnoise]}^2]]\\
    &\le 
    \sumtoinf[2]
    \vt[\step]^2
    \ex[\noisecontrol\norm{\jvecfield(\past[\jstate])}^2
    +\nPlayers\noisevar]\\
    &\le
    \sumtoinf
    (\vt[\stepalt]\update[\step]\noisecontrol
    \ex[\norm{\jvecfield(\inter[\jstate])}^2]
    +(\update[\step])^2\nPlayers\noisevar)\\
    &<+\infty.
    \end{aligned}
\end{equation}
To obtain the last inequality we have applied
\begin{enumerate*}[\itshape i\upshape )]
\item \cref{lem:OG+-sum-bound}; and
\item the summability of $(\vt[\step]^2)_{\run\in\N}$.
\end{enumerate*}
Invoking \cref{lem:non-decreasing-bounded}, we deduce that 
$\ex_{\run-1}[\norm{\vt[\jstate]-\jsol}^2]-\norm{\vt[\surr{\jstate}]-\jsol}^2$ converges to $0$ almost surely. 
This together with the almost sure convergence of $\ex_{\run-1}[\norm{\vt[\jstate]-\jsol}^2]$ to $\limp{\srv}$ we obtain the almost sure convergence of $\norm{\vt[\surr{\jstate}]-\jsol}^2$ to $\limp{\srv}$.

To summarize, we have shown that
for all $\jsol\in\sols$, the distance $\norm{\vt[\surr{\jstate}]-\jsol}$ almost surely converges. 
Applying \cref{cor:seq-cvg-almost-surely}, we conclude that the event $\{\norm{\vt[\surr{\jstate}]-\jsol} \text{ converges for all } \jsol\in\sols\}$ happens with probability $1$.

\vspace{0.4em}
\noindent
(2) \textit{There exists and increasing function $\extr\from\N\to\N$ such that $\norm{\jvecfield(\inter[\jstate][\extr(\run)])}^2+\norm{\inter[\jstate][\extr(\run)]-\vt[\surr{\jstate}][\extr(\run)]}^2$ converges to $0$ almost surely.}
\enspace
From \cref{lem:subsequence-cvg-0}, we know it is sufficient to show that 
\begin{equation*}
    \liminf_{\toinf}
    \ex[\norm{\jvecfield(\inter[\jstate])}^2+\norm{\inter[\jstate]-\vt[\surr{\jstate}]}^2]=0.
\end{equation*}
Since $\sumtoinf\vt[\stepalt]\update[\step]=+\infty$ in all the cases, the above is implied by
\begin{equation}
    \label{eq:OG+-summable-norm+diff}
    \sumtoinf[2]
    \vt[\stepalt]\update[\step]
    \ex[\norm{\jvecfield(\inter[\jstate])}^2+\norm{\inter[\jstate]-\vt[\surr{\jstate}]}^2] < +\infty.
\end{equation}
Using \cref{asm:noises} and $\vt[\step]<\vt[\stepalt]$, we have
\begin{equation}
    \notag
    \begin{aligned}
    \ex[\norm{\inter[\jstate]-\vt[\surr{\jstate}]}^2]
    &= \ex[\norm{\vt[\stepalt]\jvecfield(\past[\jstate])
    +(\vt[\step]+\vt[\stepalt])\past[\jnoise]}^2]\\
    &=\vt[\stepalt]^2\ex[\norm{\jvecfield(\past[\jstate])}^2]
    +(\vt[\step]+\vt[\stepalt])^2\ex[\norm{\past[\jnoise]}^2]\\
    &\le
    \vt[\stepalt]^2(1+4\noisecontrol)\ex[\norm{\jvecfield(\past[\jstate])}^2]
    + 4\vt[\stepalt]^2\nPlayers\noisevar.
    \end{aligned}
\end{equation}
Subsequently, with \cref{lem:OG+-sum-bound}, the summability of $(\vt[\stepalt]^2\update[\step])_{\run\in\N}$ and the fact that the learning rates are non-increasing, we obtain
\begin{equation}
    \notag
    \begin{aligned}
    \sumtoinf[2]
    \vt[\stepalt]\update[\step]
    \norm{\inter[\jstate]-\vt[\surr{\jstate}]}^2]
    &\le
    \sumtoinf[2]
    \vt[\stepalt]\update[\step]
    \ex[\vt[\stepalt]^2(1+4\noisecontrol)\norm{\jvecfield(\past[\jstate])}^2
    + 4\vt[\stepalt]^2\nPlayers\noisevar]\\
    &\le
    \sumtoinf
    \vt[\stepalt][1]^2\vt[\stepalt]\update[\step](1+4\noisecontrol)\norm{\jvecfield(\inter[\jstate])}^2
    +\sumtoinf
    4\vt[\stepalt][1]\vt[\stepalt]^2\update[\step]\nPlayers\noisevar\\
    &<+\infty.
    \end{aligned}
\end{equation}
Invoking \cref{lem:OG+-sum-bound} again gives $ \sumtoinf[2]
\vt[\stepalt]\update[\step]
\ex[\norm{\jvecfield(\inter[\jstate])}^2]<+\infty$ and thus we have effectively \eqref{eq:OG+-summable-norm+diff}.
This concludes the proof of this step.

\vspace{0.4em}
\noindent
(3) \textit{$\seqinf[\surr{\jstate}]$ converges to a point in $\sols$ almost surely.}
\enspace
Let us define the event
\begin{equation*}
    \event
    = \{\norm{\vt[\surr{\jstate}]-\jsol} \textit{ converges for all } \jsol\in\sols;
    ~\norm{\jvecfield(\inter[\jstate][\extr(\run)])}^2
    +\norm{\inter[\jstate][\extr(\run)]-\vt[\surr{\jstate}][\extr(\run)]}^2 \textit{ converges to } 0
    \}
\end{equation*}
Combining the aforementioned two points we know that $\prob(\event)=1$.
It is thus sufficient to show that $\seqinf[\surr{\jstate}]$ converges to a point in $\sols$ for any realization $\event$.

Let us consider a realization of $\event$.
The set $\sols$ being non-empty, the convergence of $\norm{\vt[\surr{\jstate}]-\jsol}$ for a $\jsol\in\sols$ implies the boundedness of $\seqinf[\surr{\jstate}]$.
Therefore, we can extract a subsequence of $(\vt[\surr{\jstate}][\extr(\run)])_\run$, which we denote by $(\vt[\surr{\jstate}][\extr(\extralt(\run)]))_\run$ that converges to a point $\limp{\jaction}\in\points$.
As $\lim_{\toinf}\norm{\inter[\jstate][\extr(\extralt(\run))]-\vt[\surr{\jstate}][\extr(\extralt(\run))]}^2=0$, we deduce that $(\inter[\jstate][\extr(\extralt(\run)]))_\run$ also converges to $\limp{\jaction}\in\points$.
Moreover, we also have $\lim_{\toinf}\norm{\jvecfield(\inter[\jstate][\extr(\extralt(\run))])}^2=0$.
By continuity of $\jvecfield$ we then know that $\jvecfield(\limp{\jaction})=0$, \ie $\limp{\jaction}\in\sols$.
By definition of $\event$, this implies the convergence of $\norm{\vt[\surr{\jstate}]-\limp{\jaction}}$.
The limit $\lim_{\toinf}\norm{\vt[\surr{\jstate}]-\limp{\jaction}}$ is thus well defined and $\lim_{\toinf}\norm{\vt[\surr{\jstate}]-\limp{\jaction}}=\lim_{\toinf}\norm{\vt[\surr{\jstate}][\extr(\extralt(\run))]-\limp{\jaction}}$.
However, $\lim_{\toinf}\norm{\vt[\surr{\jstate}][\extr(\extralt(\run))]-\limp{\jaction}}=0$ by the choice of $\limp{\jaction}$.
We have therefore $\lim_{\toinf}\norm{\vt[\surr{\jstate}]-\limp{\jaction}}=0$.
Recalling that $\limp{\jaction}\in\sols$, we have indeed shown that $\seqinf[\surr{\state}]$ converges to a point in $\sols$.

\vspace{0.4em}
\noindent
(4) \textit{Conclude: $\seqinf[\jstate]$ converges to a point in $\sols$ almost surely .}
\enspace
We claim that $\norm{\vt[\jstate]-\vt[\surr{\jstate}]}$ converges to $0$.
In fact, similar to \eqref{eq:OG+-sum-diff-bounded}, it holds that
\begin{equation}
    \notag
    \sumtoinf
    \ex[\norm{\vt[\jstate]-\vt[\surr{\jstate}]}^2]
    =\sumtoinf[2]\vt[\step]^2\ex[\norm{\past[\jnoise]}^2]
    <+\infty.
\end{equation}
Invoking \cref{lem:non-decreasing-bounded} we get 
almost sure convergence of $\norm{\vt[\jstate]-\vt[\surr{\jstate}]}$ to $0$.
Moreover, we have shown in the previous point that $\seqinf[\surr{\jstate}]$ converges to a point in $\sols$ almost surely.
Combining the above two arguments we obtain the almost sure convergence of $\seqinf[\jstate]$ to a point in $\sols$.
\end{proof}


Provided that the players use larger extrapolation steps, the convergence of $\vt[\jstate]$ does not necessarily imply the convergence of $\inter[\jstate]$.
The next theorem derives sufficient condition for the latter to hold.

\begin{theorem}
\label{thm:OG+-as-convergence-inter}
Let \cref{asm:noises,asm:lips,asm:VS} hold and all players run \eqref{OG+} with non-increasing learning rate sequences $(\vt[\stepalt])_{\run\in\N}$ and $(\vt[\step])_{\run\in\N}$ satisfying \eqref{OG+-lr} and \eqref{OG+-lr-as}.
Assume further that $\vt[\stepalt]^3=\bigoh(\vt[\step])$ and there exists $\probmoment\in(2,4]$ and $\noisedevalt>0$ such that $\ex[\norm{\vt[\jnoise]}^\probmoment]\le\noisedevalt^\probmoment$ for all $\run$ and $\sumtoinf \stepalt^\probmoment<\infty$. Then, the actual point of play
$\inter[\jstate]$ converges almost surely to a Nash equilibrium.
\end{theorem}

\begin{proof}
Since we already know that $\seqinf[\jstate]$ converges to a point in $\sols$ almost surely, it is sufficient to show that $\lim_{\toinf}\norm{\vt[\jstate]-\inter[\jstate]}=0$ almost surely.
By the update rule of OG+, we have, for $\run\ge2$, $\vt[\jstate]-\inter[\jstate]=\vt[\stepalt]\jvecfield(\past[\jstate])+\vt[\stepalt]\past[\jnoise]$.
We will deal with  
the two terms separately.
For the noise term, we notice that under the additional assumptions we have
\begin{equation*}
    \sumtoinf[2]
    \ex[\norm{\vt[\stepalt]\past[\jnoise]}^\probmoment]
    \le
    \sumtoinf[2] \vt[\stepalt]^\probmoment\noisedevalt^\probmoment
    <+\infty.
\end{equation*}
Therefore, applying \cref{lem:non-decreasing-bounded} gives the almost sure convergence of $\norm{\vt[\stepalt]\past[\jnoise]}$ to $0$.
As for the operator term, for $\run\ge3$ we bound
\begin{equation}
    \notag
    \norm{\vt[\stepalt]\jvecfield(\past[\jstate])}
    \le \vt[\stepalt]\norm{\jvecfield(\past[\jstate])-\jvecfield(\last[\jstate])}
    +\vt[\stepalt]\norm{\jvecfield(\last[\jstate])}.
\end{equation}
On one hand, as $\seqinf[\jstate]$ converges to a point in $\sols$ almost surely, the term $\vt[\stepalt]\norm{\jvecfield(\last[\jstate])}$ converges to $0$ almost surely by continuity of $\jvecfield$.
On the other hand, by Lipschitz continuity of $\jvecfield$ we have
\begin{equation}
    \label{eq:OG+-absnoise-Xt1/2-1}
    \begin{aligned}[b]
        \sumtoinf[2] \ex[(\update[\stepalt])^2\norm{
        \jvecfield(\inter[\jstate])
        -\jvecfield(\vt[\jstate])}^2]
        &\le
        \sumtoinf[2] (\update[\stepalt])^2\vt[\stepalt]^2\nPlayers\lips^2
        \ex[\norm{\past[\jsvecfield]}^2]\\
        &\le
        \sumtoinf[2] \vt[\stepalt]^4\nPlayers\lips^2
        \ex[\norm{\jvecfield(\past[\jstate])}^2]
        +\sumtoinf[2] \vt[\stepalt]^4\nPlayers\lips^2\ex[\norm{\past[\jnoise]}^2].
    \end{aligned}
\end{equation}
Since $\vt[\stepalt]^3=\bigoh(\vt[\step])$, there exists $\Cst\in\R_+$ such that $\vt[\stepalt]^3\le\Cst\vt[\step]$ for all $\run\in\N$. Along with \cref{lem:OG+-sum-bound} we get
\begin{equation}
    \label{eq:OG+-absnoise-Xt1/2-2}
    \sumtoinf[2] \vt[\stepalt]^4\nPlayers\lips^2
    \ex[\norm{\jvecfield(\past[\jstate])}^2]
    \le
    \sumtoinf[2] \last[\stepalt]\vt[\step]\Cst\nPlayers\lips^2
    \ex[\norm{\jvecfield(\past[\jstate])}^2]
    < +\infty.
\end{equation}
Since $\probmoment>2$, by Jensen's inequality $\ex[\norm{\vt[\jnoise]}^\probmoment]\le\noisedevalt^\probmoment$ implies $\ex[\norm{\vt[\jnoise]}^2]\le\noisedevalt^2$.
Along with $\probmoment\le4$ and $\sumtoinf \vt[\stepalt]^\probmoment<+\infty$ we deduce
\begin{equation}
    \label{eq:OG+-absnoise-Xt1/2-3}
    \sumtoinf[2] \vt[\stepalt]^4\nPlayers\lips^2\ex[\norm{\past[\jnoise]}^2]
    \le
    \sumtoinf[2] \vt[\stepalt]^\probmoment\vt[\stepalt][1]^{4-\probmoment}\nPlayers\lips^2\noisedevalt^2
    <+\infty.
\end{equation}
Combining \eqref{eq:OG+-absnoise-Xt1/2-1}, \eqref{eq:OG+-absnoise-Xt1/2-2}, and \eqref{eq:OG+-absnoise-Xt1/2-3} we obtain
$\sumtoinf[2] \ex[(\update[\stepalt])^2\norm{
        \jvecfield(\inter[\jstate])
        -\jvecfield(\vt[\jstate])}^2]<+\infty$,
which implies $\lim_{\toinf}\update[\stepalt]\norm{
        \jvecfield(\inter[\jstate])
        -\jvecfield(\vt[\jstate])}=0$ using \cref{lem:non-decreasing-bounded}.
In summary, we have shown the three sequences $(\vt[\stepalt]\norm{\past[\jnoise]})_{\run\in\N}$, $(\vt[\stepalt]\norm{\jvecfield(\last[\jstate])})_{\run\in\N}$, and $(\vt[\stepalt]\norm{
        \jvecfield(\past[\jstate])
        -\jvecfield(\last[\jstate])})_{\run\in\N}$
converge almost surely to $0$.
As we have
\begin{equation*}
    \norm{\vt[\jstate]-\inter[\jstate]}=
    \norm{\vt[\stepalt]\jvecfield(\past[\jstate])+\vt[\stepalt]\past[\jnoise]}
    \le
    \vt[\stepalt]\norm{\jvecfield(\past[\jstate])-\jvecfield(\last[\jstate])}
    +\vt[\stepalt]\norm{\jvecfield(\last[\jstate])}
    +\vt[\stepalt]\norm{\past[\jnoise]},
\end{equation*}
we can indeed conclude that $\lim_{\toinf}\norm{\vt[\jstate]-\inter[\jstate]}=0$ almost surely.
\end{proof}


\subsection{Trajectory Convergence of Non-Adaptive OptDA+ under Multiplicative Noise}

We now turn to the case of multiplicative noise and prove almost sure last-iterate convergence with constant learning rates.

\begin{theorem}
\label{thm:cvg-OptDA+}
Let \cref{asm:noises,asm:lips,asm:VS} hold with $\noisedev=0$
and
all players run \eqref{OG+} / \eqref{OptDA+} with learning rates given in \cref{thm:OptDA+-regret}\ref{thm:OptDA+-regret-mul-main}.
Then, both $\vt[\jstate]$ and $\inter[\jstate]$ converge almost surely to a Nash equilibrium.
\end{theorem}

\begin{proof}
As in the proof \cref{thm:OG+-as-convergence}, we define $\vt[\surr{\jstate}][1]=\vt[\jstate][1]$ and for all $\allplayers$, $\run\ge2$,
\begin{equation*}
\vpt[\surr{\state}]
=\vpt[\state]+\vp[\step]\vptpast[\noise]
=-\vp[\step]\sum_{\runalt=1}^{\run-2}\vptinter[\svecfield]
-\vp[\step]\vp[\vecfield](\past[\jstate]).
\end{equation*}
$\vt[\surr{\jstate}]$ serves a surrogate for $\vt[\jstate]$ and is $\last[\filter]$-measurable.
Our first step is to show that

\begin{center}
\textit{With probability $1$, $\norm{\vt[\surr{\jstate}]-\jsol}_{1/\jstep}$ converges for all $\jsol\in\sols$.}
\end{center}
For this, we fix $\jsol\in\sols$ and apply Robbins-Siegmund's theorem (\cref{lem:Robbins-Siegmund}) to inequality \eqref{eq:lem:OptDA+-quasi-descent} of \cref{lem:OptDA+-quasi-descent}
with
\begin{align*}
    &\vt[\filteralt]\subs\last[\filter],
    ~~~
    \vt[\srv]\subs\ex_{\run-1}[\norm{\vt[\jstate]-\jsol}^2_{1/\jstep}],
    ~~~
    \vt[\srvmul]\subs0,
    ~~~
    \vt[\srvmi]\subs
    \ex_{\run-1}[\norm{\jvecfield(\inter[\jstate])}_{\jstepalt}^2]
    +\norm{\jvecfield(\past[\jstate])}_{\jstepalt}^2,\\
    &\vt[\srvp]
    \subs\ex_{\run-1}
    [
    \begin{aligned}[t]
    &3\norm{
    \jvecfield(\vt[\jstate])-\jvecfield(\last[\jstate])
    }_{\jstepalt}^2
    +(4
    \nPlayers+1)\lips\norm{\past[\jnoise]}
    _{\jstepalt^2}^2
    \\
    &
    +3\lips^2
    (\onenorm{\jstepalt}\norm{\past[\jsvecfield]}_{
    \jstepalt^2}^2
    +
    \norm{\jstepalt}_1
    \norm{\ancient[\jsvecfield]}_{
    \jstepalt^2}^2)
    +2\norm{\inter[\jsvecfield]}_{\jstep}^2].
    \end{aligned}
\end{align*}
For $\run=1$ we use \eqref{eq:OptDA+-t=1}; thus $\vt[\srvmi]=0$ and $\vt[\srvp]=\norm{\vt[\jsvecfield][3/2]}_{\jstep}^2$.
To see that Robbins-Siegmund's theorem is effectively applicable, we use \cref{asm:noises,asm:lips} with $\noisedev=0$ to establish\footnote{For $\run=1$ and $\run=2$, we remove the terms that involve either $\vt[\jstate][1/2]$, $\vt[\jstate][0]$, or $\vt[\jstate][-1/2]$.}
\begin{equation}
    \notag
    \begin{aligned}
    \ex[\vt[\srvp]]
    \le
    \ex[
    &3\infnorm{\jstepalt}\lips^2
    \norm{\vt[\jstate]-\last[\jstate]}^2
    \\
    &
    +
    (\infnorm{\jstepalt}
    (4\nPlayers+1)\lips\noisecontrol
    +
    3\infnorm{\jstepalt}^2
    \nPlayers\lips^2(1+\noisecontrol))
    \norm{\jvecfield(\past[\jstate])}_{\jstepalt}^2
    ]
    \\
    &
    +3\infnorm{\jstepalt}^2
    \nPlayers\lips^2(1+\noisecontrol)
    \norm{\jvecfield(\ancient[\jstate])}_{\jstepalt}^2
    +2(1+\noisecontrol)
    \norm{\jvecfield(\inter[\jstate])}_{\jstep}^2].
    \end{aligned}
\end{equation}
With $2(1+\noisecontrol)\jstep\le\jstepalt$, it follows immediately from \cref{lem:OptDA+-sum-bound-nonadapt} that $\sumtoinf\ex[\vt[\srvp]]<+\infty$.
Robbins-Siegmund's theorem thus ensures the almost sure convergence of $\ex_{\run-1}[\norm{\vt[\jstate]-\jsol}^2]$ to a finite random variable.
By definition of $\vpt[\surr{\state}]$, we have
\begin{equation}
    \notag
    \ex_{\run-1}[\norm{\vpt[\state]-\vp[\sol]}^2]
    = \ex_{\run-1}[\norm{\vpt[\surr{\state}]
    -\vp[\step]\vptpast[\noise]-\vp[\sol]}^2]
    = \norm{\vpt[\surr{\state}]-\vp[\sol]}^2
    + (\vp[\step])^2
    \ex_{\run-1}[\norm{\vptpast[\noise]}^2].
\end{equation}
Subsequently
\begin{equation}
    \notag
    \ex_{\run-1}[\norm{\vt[\jstate]-\jsol}_{1/\jstep}^2]
    =
    \norm{\vt[\surr{\jstate}]-\jsol}_{1/\jstep}^2
    +
    \ex_{\run-1}[
    \norm{\past[\jnoise]}_{\jstep}^2].
\end{equation}
Therefore, by \cref{asm:noises} with $\noisedev=0$ and \cref{lem:OptDA+-sum-bound-nonadapt} we get
\begin{equation}
    \notag
    \sumtoinf[2]
    \ex[\ex_{\run-1}[\norm{\vt[\jstate]-\jsol}_{1/\jstep}^2]
    -\norm{\vt[\surr{\jstate}]-\jsol}_{1/\jstep}^2]
    = \sumtoinf[2]
    \ex[\norm{\past[\jnoise]}_{\jstep}^2]
    \le \sumtoinf[2]
    \noisecontrol
    \ex[\norm{\jvecfield(\past[\jstate])}_{\jstep}^2]
    < +\infty.
\end{equation}
Following the proof of \cref{thm:OG+-as-convergence}, we deduce with the help of \cref{lem:non-decreasing-bounded} and \cref{cor:seq-cvg-almost-surely}
that the claimed argument is effectively true, \ie
with probability $1$, $\norm{\vt[\surr{\jstate}]-\jsol}_{1/\jstep}$ converges for all $\jsol\in\sols$.

Since $\norm{\vt[\jstate]-\vt[\surr{\jstate}]}^2=\norm{\past[\jnoise]}_{\jstep^2}^2$ and $\norm{\inter[\jstate]-\vt[\surr{\jstate}]}^2=\sumplayers\norm{\vp[\stepalt]\vptpast[\svecfield]+\vp[\step]\vptpast[\noise]}^2$ (for $\run\ge2$),
applying the multiplicative noise assumption,
\cref{lem:OptDA+-sum-bound-nonadapt}, and \cref{lem:non-decreasing-bounded} we deduce that both $\norm{\vt[\jstate]-\vt[\surr{\jstate}]}$ and $\norm{\inter[\jstate]-\vt[\surr{\jstate}]}$ converge to $0$ almost surely.
Moreover, \cref{lem:OptDA+-sum-bound-nonadapt} along with \cref{lem:non-decreasing-bounded} also implies the almost sure convergence of $\norm{\jvecfield(\inter[\state])}$ to $0$.
In summary, we have shown that the event
\begin{equation}
    \notag
    \event
    \defeq
    \left\{
    \begin{gathered}
    \norm{\vt[\surr{\jstate}]-\jsol}_{1/\jstep}
    \text{ converges for all } \jsol\in\sols,\\
    \lim_{\toinf}\norm{\vt[\jstate]-\vt[\surr{\jstate}]}=0,
    ~~
    \lim_{\toinf}\norm{\inter[\jstate]-\vt[\surr{\jstate}]}=0,
    ~~
    \lim_{\toinf}\norm{\jvecfield(\inter[\jstate])}=0
    \end{gathered}
    \right\}
\end{equation}
happens almost surely.
To conclude, we just need to show that $\vt[\jstate]$ and $\inter[\jstate]$ converge to a point in $\sols$ whenever $\event$ happens.
The convergence of $\norm{\vt[\surr{\jstate}]-\jsol}_{1/\jstep}$ for a point $\jsol$ in particular implies the boundedness of $\seqinf[\surr{\jstate}]$.
Therefore, $\seqinf[\surr{\jstate}]$ has at least a cluster point, which we denote by $\limp{\jaction}$.
Provided that $\lim_{\toinf}\norm{\inter[\jstate]-\vt[\surr{\jstate}]}=0$, the point $\limp{\jaction}$ is clearly also a cluster point of $(\inter[\jstate])_{\run\in\N}$.
By $\lim_{\toinf}\norm{\jvecfield(\inter[\jstate])}=0$ and the continuity of $\jvecfield$ we then have $\jvecfield(\limp{\jaction})=0$, \ie $\limp{\jaction}\in\sols$.
This in turn implies that $\norm{\vt[\surr{\jstate}]-\limp{\jaction}}_{1/\jstep}$ converges, so this limit can only be $0$.
In other words, $\seqinf[\surr{\jstate}]$ converges to $\limp{\jaction}$; we conclude by $\lim_{\toinf}\norm{\vt[\jstate]-\vt[\surr{\jstate}]}=0$ and $\lim_{\toinf}\norm{\inter[\jstate]-\vt[\surr{\jstate}]}=0$.
\end{proof}


\subsection{Trajectory Convergence of Adaptive OptDA+ under Multiplicative Noise}

In closing, we prove the almost sure last-iterate convergence of adaptive OptDA+ under multiplicative noise.
As claimed in \cref{sec:trajectory}, we first show that the learning rates almost surely converge to positive constant.
This intuitively means that the analysis of the last section should apply as well.

\begin{lemma}
\label{lem:OptDA+-adapt-relnoise-lr-as}
Let \cref{asm:lips,asm:VS,asm:noises,asm:boundedness} hold with $\noisedev=0$ and all players run OptDA+ with adaptive learning rates \eqref{adaptive-lr}. 
Then, 
\begin{enumerate}[(a)]
    \item
    With probability $1$,
    for all $\allplayers$,
    $(\vpt[\regpar])_{\run\in\N}$
    and
    $(\vpt[\regparalt])_{\run\in\N}$
    converge to finite constant.
    \item
    With probability $1$,
    for all $\allplayers$,
    the learning rates $(\vpt[\stepalt])_{\run\in\N}$ and $(\vpt[\step])_{\run\in\N}$ converge to positive constants.
\end{enumerate}
\end{lemma}
\begin{proof}
We notice that (b) is a direct consequence of (a) so we will only show (a) below.
For this, we make use of \cref{lem:OptDA+-adapt-relnoise-finite} and \cref{lem:non-decreasing-bounded}.
In fact, $(\sqrt{\vpt[\regpar]})_{\run\in\N}$ is clearly non-decreasing and by \cref{lem:OptDA+-adapt-relnoise-finite}, $\sup_{\run\in\N}\ex[\sqrt{\vpt[\regpar]}]<+\infty$.
Therefore, \cref{lem:non-decreasing-bounded} ensures the almost sure convergence of $(\sqrt{\vpt[\regpar]})_{\run\in\N}$ to a finite random variable, which in turn implies that $(\vpt[\regpar])_{\run\in\N}$ converges to a finite constant almost surely.
Similarly, $(\vpt[\regparalt])_{\run\in\N}$ is non-decreasing and $\sup_{\run\in\N}\ex[\vpt[\regparalt]]<+\infty$ by \cref{lem:OptDA+-adapt-relnoise-finite}. We thus deduce by \cref{lem:non-decreasing-bounded} that $(\vpt[\regparalt])_{\run\in\N}$ converges to finite constant almost surely.
\end{proof}

We now adapt the proof of \cref{thm:cvg-OptDA+} to the case of adaptive learning rates.
Note that the fact that the learning rates are not constant  also causes some additional challenges.

\begin{theorem}
\label{thm:OptDA+-adapt-as-convergence}
Let \cref{asm:lips,asm:VS,asm:noises,asm:boundedness} hold with $\noisedev=0$ and all players run \eqref{OptDA+} with adaptive learning rates \eqref{adaptive-lr}.
Then, 
\begin{enumerate}[(a)]
    \item It holds almost surely that $\sumtoinf\norm{\jvecfield(\inter[\jstate])}^2<+\infty$.
    \item Both $\seqinf[\jstate]$ and $(\inter[\jstate])_{\run\in\N}$ converge to a Nash equilibrium almost surely.
\end{enumerate}
\end{theorem}

\begin{proof}
In the following, we define 
$\limp{\jstepalt}=\lim_{\toinf}\vt[\jstepalt]$ and
$\limp{\jstep}=\lim_{\toinf}\vt[\jstep]$ as the limits of the learning rate sequences.
Since for each $\allplayers$, $(\vpt[\stepalt])_{\run\in\N}$ and $(\vpt[\step])_{\run\in\N}$ are non-negative non-increasing sequences, both $\limp{\jstepalt}$ and $\limp{\jstep}$ are well-defined.
Moreover, by \cref{lem:OptDA+-adapt-relnoise-lr-as} we know that $\limp{\jstepalt}$ and $\limp{\jstep}$ are positive almost surely.

\vspace{0.4em}
\noindent
(a) Combining \cref{lem:OptDA+-sum-bound-adapt} and \cref{lem:OptDA+-adapt-relnoise-finite} we get immediately $\sumtoinf
\ex[\norm{\jvecfield(\inter[\jstate])}_{\vt[\jstepalt]}^2]<+\infty$.
Therefore, using \cref{lem:non-decreasing-bounded} we deduce that $\sumtoinf\norm{\jvecfield(\inter[\jstate])}_{\vt[\jstepalt]}^2<+\infty$ almost surely.
By definition of $\limp{\jstepalt}$ we have
\begin{equation}
\notag
\sumtoinf\norm{\jvecfield(\inter[\jstate])}_{\vt[\jstepalt]}^2
\ge
\sumtoinf\norm{\jvecfield(\inter[\jstate])}_{\limp{\jstepalt}}^2
\ge
\min_{\allplayers}\vp[\limp{\stepalt}]
\sumtoinf
\norm{\jvecfield(\inter[\jstate])}^2
\end{equation}
As a consequence, whenever
\begin{enumerate*}[\itshape i\upshape)]
    \item $\Cst\defeq\sumtoinf\norm{\jvecfield(\inter[\jstate])}_{\vt[\jstepalt]}^2$ is finite; and
    \item $\min_{\allplayers}\vp[\limp{\stepalt}]>0$,
\end{enumerate*}
we have
\begin{equation*}
    \sumtoinf
    \norm{\jvecfield(\inter[\jstate])}^2
    \le \frac{\Cst}{\min_{\allplayers}\vp[\limp{\stepalt}]}
    <+\infty.
\end{equation*}
As both \emph{i}) and \emph{ii}) hold almost surely, we have indeed shown that $\sumtoinf\norm{\jvecfield(\inter[\jstate])}^2<+\infty$ almost surely.

\vspace{0.4em}
\noindent
(b) To prove this point, we follow closely the proof of \cref{thm:cvg-OptDA+}.
To begin, we fix $\jsol\in\sols$ and show that we can always apply Robbins-Siegmund's theorem (\cref{lem:Robbins-Siegmund}) to inequality \eqref{eq:lem:OptDA+-quasi-descent} of \cref{lem:OptDA+-quasi-descent} (or inequality \eqref{eq:OptDA+-t=1} for $\run=1$).
This gives, for $\run\ge2$,
\begin{align*}
    &\vt[\filteralt]=\last[\filter],
    ~~~
    \vt[\srv]=\ex_{\run-1}[\norm{\vt[\jstate]-\jsol}^2_{1/\vt[\jstep]}],
    ~~~
    \vt[\srvmul]=0,
    ~~~
    \vt[\srvmi]=\ex_{\run-1}[\norm{\jvecfield(\inter[\jstate])}_{\vt[\jstepalt]}^2]
    +\norm{\jvecfield(\past[\jstate])}_{\vt[\jstepalt]}^2,\\
    &\vt[\srvp]
    =\ex_{\run-1}
    [
    \begin{aligned}[t]
    &3\norm{
    \jvecfield(\vt[\jstate])-\jvecfield(\last[\jstate])
    }_{\vt[\jstepalt]}^2
    +\norm{\vt[\jstate][1]-\jsol}_{1/\update[\jstep]-1/\vt[\jstep]}^2
    +(4\nPlayers+1)\lips\norm{\past[\jnoise]}
    _{\vt[\jstepalt]^2}^2
    +3\lips^2
    \\
    &
    (\norm{\vt[\jstepalt]}_1\norm{\past[\jsvecfield]}_{
    \vt[\jstepalt]^2}^2
    +
    \norm{\last[\jstepalt]}_1
    \norm{\ancient[\jsvecfield]}_{
    (\last[\jstepalt])^2}^2)
    +2\norm{\inter[\jsvecfield]}_{\vt[\jstep]}^2].
    \end{aligned}
\end{align*}
As for $\run=1$, we replace the above with $\vt[\srvmi]=0$ and $\vt[\srvp]=\norm{\vt[\jsvecfield][3/2]}_{\vt[\jstep][1]}^2$.
Using \cref{asm:lips}, \eqref{eq:OptDA+-adapt-sum-bound-Bt}, and \eqref{eq:OptDA+-adapt-sum-bound-Ct},
we can bound the sum of the expectation of $\vt[\srvp]$ by
\begin{align*}
    \sum_{\run=1}^{\nRuns}
    \ex[\vt[\srvp]]
    &\le
    \sum_{\run=1}^{\nRuns-1}
    3\lips^2
    \ex[\norm{\vt[\jstate]-\update[\jstate]}^2]
    +
    \sumplayers
    \left(
    \norm{\vpt[\state][\play][1]-\vp[\sol]}^2
    \ex\left[
    \sqrt{1+\vptlast[\regpar][\play][\nRuns]+\vptlast[\regparalt][\play][\nRuns]}
    \right]
    \right)
    \\
    &~~~+(6\nPlayers\lips^2
    +(4\nPlayers+1)\lips)
    \left
    (
    2\nPlayers(\gbound^2+\noisebound^2)
    +
    \sumplayers
    2
    \ex\left[
    \sqrt{\vpt[\regpar][\play][\nRuns-1]}
    \right]
    \right)
    +
    \\
    &
    ~~~
    +8\nPlayers(\gbound^2+\noisebound^2)
    +
    \sumplayers
    4\ex
    \left[
    \sqrt{\vpt[\regpar][\play][\nRuns]}
    \right]
\end{align*}
It then follows immediately from \cref{lem:OptDA+-adapt-relnoise-finite} that $\sumtoinf\ex[\vt[\srvp]]<+\infty$.
With Robbins-Siegmund's theorem we deduce that $\ex_{\run-1}[\norm{\vt[\jstate]-\jsol}^2_{1/\vt[\jstep]}]$ converges almost surely to a finite random variable.

As in the proof of \cref{thm:OG+-as-convergence,thm:cvg-OptDA+}, we next define $\vt[\surr{\jstate}][1]=\vt[\jstate][1]$ and for all $\allplayers$, $\run\ge2$,
\begin{equation*}
\vpt[\surr{\state}]
=\vpt[\state]+\vpt[\step]\vptpast[\noise]
=-\vpt[\step]\sum_{\runalt=1}^{\run-2}\vptinter[\svecfield]
-\vpt[\step]\vp[\vecfield](\past[\jstate]).
\end{equation*}
Then,
\begin{equation}
    \notag
    \ex_{\run-1}[\norm{\vt[\jstate]-\jsol}_{1/\vt[\jstep]}^2]
    =
    \norm{\vt[\surr{\jstate}]-\jsol}_{1/\vt[\jstep]}^2
    +
    \ex_{\run-1}[
    \norm{\past[\jnoise]}_{\vt[\jstep]}^2].
\end{equation}
Using
$\ex_{\run-1}[\norm{\vptpast[\noise]}^2]\le\ex_{\run-1}[\norm{\vptpast[\svecfield]}^2]$,
the law of total expectation,
the fact that $\vt[\jstep]$ is $\last[\filter]$-measurable,
\cref{cor:adagrad-lemma-with-lr},
and \cref{lem:OptDA+-adapt-relnoise-finite}, we then get
\begin{equation}
    \label{eq:OptDA+-adapt-et-1-surr-diff}
    \begin{aligned}[b]
    \sumtoinf[2]
    \ex[\ex_{\run-1}[\norm{\vt[\jstate]-\jsol}_{1/\vt[\jstep]}^2]
    -\norm{\vt[\surr{\jstate}]-\jsol}_{1/\vt[\jstep]}^2]
    &= \sumtoinf[2]
    \ex[\norm{\past[\jnoise]}_{\vt[\jstep]}^2]
    \\
    &\le
    \sumtoinf[2]
    \ex[\norm{\past[\jsvecfield]}_{\vt[\jstep]}^2]
    \\
    &
    \le
    2\nPlayers(\gbound^2+\noisebound^2)
    +
    \sup_{\run\in\N}
    \sumplayers
    2\ex\left[\sqrt{\vpt[\regpar]}\right]
    \\
    &<+\infty.
    \end{aligned}
\end{equation}
Invoking \cref{lem:non-decreasing-bounded} we deduce that $\ex_{\run-1}[\norm{\vt[\jstate]-\jsol}_{1/\vt[\jstep]}^2]-\norm{\vt[\surr{\jstate}]-\jsol}_{1/\vt[\jstep]}^2$
almost surely converges to $0$.
Since we have shown $\ex_{\run-1}[\norm{\vt[\jstate]-\jsol}_{1/\vt[\jstep]}^2]$ almost surely converges to a finite random variable, we now know that $\norm{\vt[\surr{\jstate}]-\jsol}_{1/\vt[\jstep]}^2$ almost surely converges to this finite random variable as well.
To summarize, we have shown that 
for any $\jsol\in\sols$, $\norm{\vt[\surr{\jstate}]-\jsol}_{1/\vt[\jstep]}$ converges almost surely.
Moreover, we also know that $(1/\limp{\jstep})$, the limit of $(1/\vt[\jstep])_{\run\in\N}$ is finite almost surely.
Therefore, applying \cref{lem:seq-cvg-almost-surely} with
$\cpt\subs\sols$, $\vt[\juvec]\subs\vt[\surr{\jstate}]$, and $\vt[\weights]\subs1/\vt[\jstep]$,
we deduce that with probability $1$, the vector
$1/\limp{\jstep}$ is finite and $\norm{\vt[\surr{\jstate}]-\jsol}_{1/\limp{\jstep}}$ converges for all $\jsol\in\sols$. 

Next, with $\norm{\vt[\jstate]-\vt[\surr{\jstate}]}^2=\norm{\past[\jnoise]}_{\vt[\jstep]^2}^2$ and $\norm{\inter[\jstate]-\vt[\surr{\jstate}]}^2=\sumplayers\norm{\vpt[\stepalt]\vptpast[\svecfield]+\vpt[\step]\vptpast[\noise]}^2$ (for $\run\ge2$), following the reasoning of \eqref{eq:OptDA+-adapt-et-1-surr-diff}, we get both $\sumtoinf\ex[\norm{\vt[\jstate]-\vt[\surr{\jstate}]}^2]<+\infty$ and $\sumtoinf\ex[\norm{\inter[\jstate]-\vt[\surr{\jstate}]}^2]<+\infty$.
By \cref{lem:non-decreasing-bounded} we then know that $\norm{\vt[\jstate]-\vt[\surr{\jstate}]}$ and $\norm{\inter[\jstate]-\vt[\surr{\jstate}]}$ converge to $0$ almost surely.
Finally, from point (a) we know that $\norm{\jvecfield(\inter[\jstate])}$ converges to $0$ almost surely.
To conclude, let us define the event
\begin{equation}
    \notag
    \event
    \defeq
    \left\{
    \begin{gathered}
    \text{
    $1/\limp{\jstep}$ is finite and 
    $\norm{\vt[\surr{\jstate}]-\jsol}_{1/\limp{\jstep}}$
    converges for all $\jsol\in\sols$},\\
    \lim_{\toinf}\norm{\vt[\jstate]-\vt[\surr{\jstate}]}=0,
    ~~
    \lim_{\toinf}\norm{\inter[\jstate]-\vt[\surr{\jstate}]}=0,
    ~~
    \lim_{\toinf}\norm{\jvecfield(\inter[\jstate])}=0
    \end{gathered}
    \right\}
\end{equation}
We have shown that $\prob(\event)=1$.
Moreover, with the arguments of \cref{thm:cvg-OptDA+} we deduce that whenever $\event$ happens both $\seqinf[\jstate]$ and $(\inter[\jstate])_{\run\in\N}$ converge to a point in $\sols$, and this ends the proof.
\end{proof}

%% file: appendices/aux-lemmas.tex
\subsection{Lemmas on Stochastic Sequences}

To begin, we state several basic lemmas concerning stochastic sequences.
The first one translates a bound of expectation into almost sure boundedness and convergence.
It is a special case of Doob's martingale convergence theorem~\cite{HH80}, but we also provide another elementary proof below.
For simplicity, throughout the sequel, we use the term finite random variable to refer to those random variables which are finite almost surely.

\begin{lemma}
\label{lem:non-decreasing-bounded}
Let $\seqinf[\srv]$ be a sequence of non-decreasing and non-negative real-valued random variables. If there exists constant $\Cst\in\R$ such that 
\[
\forall\, \run\in\N,~~ \ex[\vt[\srv]]\le\Cst.
\]
Then $\seqinf[\srv]$ converges almost surely to a finite random variable.
In particular, for any sequence of non-negative real-valued random variables $\seqinf[\srvp]$, the fact that $\sumtoinf\ex[\vt[\srvp]]<+\infty$ implies $\sumtoinf\vt[\srvp]<+\infty$ almost surely, and accordingly $\lim_{\toinf}\vt[\srvp]=0$ almost surely.
\end{lemma}
\begin{proof}
Let $\limp{\srv}$ be the pointwise limit of $\seqinf[\srv]$.
Applying Beppo Levi's lemma we deduce that $\limp{\srv}$ is also measurable and $\lim_{\toinf}\ex[\vt[\srv]]=\ex[\limp{\srv}]$.
Accordingly, $\ex[\limp{\srv}]\le\Cst$.
The random variable $\limp{\srv}$ being non-negative, $\ex[\limp{\srv}]\le\Cst<+\infty$ implies that $\limp{\srv}$ is finite almost surely, which concludes the first statement of the lemma. The second statement is derived from the first statement by setting $\vt[\srv]=\sum_{\runalt=1}^{\run}\vt[\srvp][\runalt]$.
\end{proof}

The next lemma is essential for building almost sure last-iterate converge in the case of vanishing learning rates, as it allows to extract a convergent subsequence.

\begin{lemma}
\label{lem:subsequence-cvg-0}
Let $\seqinf[\srv]$ be a sequence of non-negative real-valued random variables such that \[\liminf_{\toinf}\,\ex[\vt[\srv]]=0.\]
Then,
\begin{enumerate*}[\itshape i\upshape)]
\item there exists a subsequence $(\vt[\srv][\extr(\run)])_{\run\in\N}$ of $\seqinf[\srv]$ that converges to $0$ almost surely;%
\footnote{We remark that the choice of the subsequence does not depend on the realization but only the distribution of the random variables.}
and accordingly
\item it holds almost surely that $\liminf_{\toinf}\,\vt[\srv]=0$.
\end{enumerate*}
\end{lemma}
\begin{proof}
Since $\liminf_{\toinf}\,\ex[\vt[\srv]]=0$, we can 
extract a subsequence $(\vt[\srv][\extr(\run)])_{\run\in\N}$ such that for all $\run\in\N$,
$\ex[\vt[\srv][\extr(\run)]]\le2^{-\run}$.
This gives $\sum_{\run=1}^{+\infty}\ex[\vt[\srv][\extr(\run)]]<+\infty$ and 
invoking \cref{lem:non-decreasing-bounded} we then know that $\sum_{\run=1}^{+\infty}\vt[\srv][\extr(\run)]<+\infty$ almost surely, which in turn implies that $\vt[\srv][\extr(\run)]$ converges to $0$ almost surely.
To prove (\textit{ii}), we just notice that for any realization such that $\lim_{\toinf}\vt[\srv][\extr(\run)]=0$, we have
$0=\lim_{\toinf}\vt[\srv][\extr(\run)]
\ge \liminf_{\toinf}\,\vt[\srv]
\ge 0$ and thus the equalities must hold, \ie $\liminf_{\toinf}\,\vt[\srv]=0$.
\end{proof}

Another important building block is Robbins-Sigmunds's theorem that allows us to show almost sure convergence of the Lyapunov function to a finite random variable.

\begin{lemma}[\citet{RS71}]
\label{lem:Robbins-Siegmund}
Consider a filtration $\seqinf[\filteralt][\run]$ and four non-negative real-valued $\seqinf[\filteralt][\run]$-adapted
processes $\seqinf[\srv][\run]$, $\seqinf[\srvmul][\run]$, $\seqinf[\srvp][\run]$, $\seqinf[\srvmi][\run]$
such that $\ex[\vt[\srv][1]]<+\infty$, $\sumtoinf\ex[\vt[\srvmul]] < \infty$, $\sumtoinf\ex[\vt[\srvp]] < \infty$, and for all $ \run\in\N$,
\begin{equation}
    \notag
    \exof{\update[\srv]\given{\current[\filteralt]}}
    \le (1+\current[\srvmul])\current[\srv] + \current[\srvp] - \current[\srvmi].
\end{equation}
Then $\seqinf[\srv][\run]$ converges almost surely to a finite random variable 
and $\sumtoinf\vt[\srvmi]<\infty$ almost surely.
\end{lemma}

Finally, since the solution may not be unique, we need a to translate the result with respect to a single point to the one that applies to the entire set. This is achieved through the following lemma.

\begin{lemma}
\label{lem:seq-cvg-almost-surely}
Let $\cpt\subseteq\vecspace$ be a closed set, $\seqinf[\juvec]$ be a sequence of $\vecspace$-valued random variable,
and $\seqinf[\weights]$ be a sequence of $\R^{\nPlayers}$-valued random variable such that
\begin{enumerate}[(a)]
    \item For all $\allplayers$,
    $\vpt[\weight][\play][1]\ge1$, $(\vpt[\weight])_{\run\in\N}$ is non-decreasing and converges to a finite constant almost surely.
    \item For all $\jaction\in\cpt$, $\norm{\vt[\juvec]-\jaction}_{\vt[\weights]}$ converges almost surely.
\end{enumerate}
Then, with probability $1$, the vector $\limp{\weights}=\lim_{\toinf}\vt[\weights]$ is well-defined, finite, and $\norm{\vt[\juvec]-\jaction}_{\limp{\weights}}$ converges
for all $\jaction\in\cpt$.
\end{lemma}
\begin{proof}
As $\vecspace$ is a separable metric space, $\cpt$ is also separable and we can find a countable set $\denseset$ such that $\cpt=\cl(\denseset)$.
Let us define the event
\begin{equation}
    \event
    \defeq
    \{
    \limp{\weights}=\lim_{\toinf}\vt[\weights]
    \text{ is well-defined and finite};
    ~~
    \norm{\vt[\juvec]-\jzvec}_{\vt[\weights]}
    \text{ converges for all }
    \jzvec\in\denseset.
    \}
\end{equation}
The set $\denseset$ being countable, from (a) and (b) we then know that $\prob(\event)=1$.
In the following, we show that $\norm{\vt[\juvec]-\jaction}_{\limp{\weights}}$ converges
for all $\jaction\in\cpt$ whenever $\event$ happens, which concludes our proof.

Let us now consider a realization of $\event$.
We first establish the convergence of $\norm{\vt[\juvec]-\jzvec}_{\limp{\weights}}$ for any $\jzvec\in\denseset$.
To begin, the convergence of $\norm{\vt[\juvec]-\jzvec}_{\vt[\weights]}$ implies the boundedness of this sequence, from which we deduce immediately the boundedness of $\norm{\vt[\juvec]-\jzvec}$
as $\norm{\vt[\juvec]-\jzvec}\le\norm{\vt[\juvec]-\jzvec}_{\vt[\weights]}$ by $\vt[\weights]\ge\vt[\weights][1]\ge1$.
In other words, $\Cst=\sup_{\run\in\N}\norm{\vt[\juvec]-\jzvec}$ is finite.
Furthermore, we have
\begin{equation}
    \label{eq:lr-bound-diff}
    0
    \le
    \norm{\vt[\juvec]-\jzvec}_{\limp{\weights}}^2
    -
    \norm{\vt[\juvec]-\jzvec}_{\vt[\weights]}^2
    =
    \sumplayers
    (\vp[\limp{\weight}]
    -\vpt[\weight])
    \norm{\vpt[\uvec]-\vp[\zvec]}^2
    \le
    \sumplayers
    (\vp[\limp{\weight}]
    -\vpt[\weight])
    \Cst^2.
\end{equation}
Since $\vp[\limp{\weight}]-\vpt[\weight]$ converges to $0$ when $\run$ goes to infinity,
from \eqref{eq:lr-bound-diff} we get immediately
$\lim_{\toinf}(\norm{\vt[\juvec]-\jzvec}_{\limp{\weights}}^2
-\norm{\vt[\juvec]-\jzvec}_{\vt[\weights]}^2)=0$.
This shows that $\norm{\vt[\juvec]-\jzvec}_{\limp{\weights}}^2$ converges to $\lim_{\toinf}\norm{\vt[\juvec]-\jzvec}_{\vt[\weights]}^2$, which exists by definition of $\event$.
We have thus shown the convergence of $\norm{\vt[\juvec]-\jzvec}_{\limp{\weights}}$.

To conclude, we need to show that $\norm{\vt[\juvec]-\jaction}_{\limp{\weights}}$ in fact converges for all $\jaction\in\cpt$.
Let $\jaction\in\cpt$.
As $\denseset$ is dense in $\cpt$, there exists a sequence of points $\seqinf[\jzvec][\indg]$ with $\vt[\jzvec][\indg]\in\denseset$ for all $\indg\in\N$ such that $\lim_{\toinf[\indg]}\vt[\jzvec][\indg]=\jaction$.
For any $\run,\indg\in\N$, the triangular inequality implies
\begin{equation*}
    -\norm{
    \vt[\jzvec][\indg]
    -\jaction
    }_{\limp{\weights}}
    \le
    \norm{
    \vt[\juvec]
    -\jaction
    }_{\limp{\weights}}
    -\norm{
    \vt[\juvec]
    -\vt[\jzvec][\indg]
    }_{\limp{\weights}}
    \le
    \norm{
    \vt[\jzvec][\indg]
    -\jaction
    }_{\limp{\weights}}.
\end{equation*}
Since $\vt[\jzvec][\indg]\in\denseset$, we have shown that $\lim_{\toinf}\norm{\vt[\juvec]-\vt[\jzvec][\indg]}$ exists.
Subsequently, we get
\begin{align*}
    -\norm{
    \vt[\jzvec][\indg]
    -\jaction
    }_{\limp{\weights}}
    &\le
    \liminf_{\toinf}
    \norm{
    \vt[\juvec]
    -\jaction
    }_{\limp{\weights}}
    -
    \lim_{\toinf}
    \norm{
    \vt[\juvec]
    -\vt[\jzvec][\indg]
    }_{\limp{\weights}}
    \\
    &\le
    \limsup_{\toinf}
    \norm{
    \vt[\juvec]
    -\jaction
    }_{\limp{\weights}}
    -
    \lim_{\toinf}
    \norm{
    \vt[\juvec]
    -\vt[\jzvec][\indg]
    }_{\limp{\weights}}
    \\
    &\le
    \norm{
    \vt[\jzvec][\indg]
    -\jaction
    }_{\limp{\weights}}.
\end{align*}
Taking the limit as $\toinf[\indg]$, we deduce that $\lim_{\toinf[\indg]}
    \lim_{\toinf}
    \norm{
    \vt[\juvec]
    -\vt[\jzvec][\indg]
    }_{\limp{\weights}}$ exists and
\begin{equation*}
    \liminf_{\toinf}
    \norm{
    \vt[\juvec]
    -\jaction
    }_{\limp{\weights}}
    =
    \lim_{\toinf[\indg]}
    \lim_{\toinf}
    \norm{
    \vt[\juvec]
    -\vt[\jzvec][\indg]
    }_{\limp{\weights}}
    =\limsup_{\toinf}
    \norm{
    \vt[\juvec]
    -\jaction
    }_{\limp{\weights}}.
\end{equation*}
This shows the convergence of $\norm{\vt[\juvec]-\jaction}_{\limp{\weights}}$.
\end{proof}

\begin{corollary}
\label{cor:seq-cvg-almost-surely}
Let $\cpt\subseteq\vecspace$ be a closed set, $\seqinf[\juvec]$ be a sequence of $\vecspace$-valued random variable,
and $\weights\in\R^{\nPlayers}$ 
such that $\vp[\weight]\ge1$ for all $\allplayers$, and for all $\jaction\in\cpt$, $\norm{\vt[\juvec]-\jaction}_{\weights}$ converges almost surely.
Then, with probability $1$,
$\norm{\vt[\juvec]-\jaction}_{\weights}$ converges for all $\jaction\in\cpt$.
\end{corollary}